\documentclass[a4paper,10pt]{article}
\usepackage{amsmath,amsfonts,amssymb,amsthm}
\usepackage{stmaryrd}
\SetSymbolFont{stmry}{bold}{U}{stmry}{m}{n}
\usepackage{mathtools}
\usepackage{comment} 
\usepackage{physics}
\usepackage[mathcal]{euscript}
\usepackage{graphicx}
\usepackage[colorlinks,linkcolor=blue,citecolor=teal,psdextra]{hyperref}
\usepackage[top=2.8cm,bottom=2.8cm,left=2.8cm,right=2.8cm]{geometry}
\usepackage{tikz}
\usepackage[export]{adjustbox}
\usepackage{multirow}
\usepackage{appendix}
\usepackage{setspace}
\usepackage{tcolorbox}
\usepackage{rotating}
\usepackage{todonotes}
\usepackage{pdflscape}
\usepackage{titlesec}
\titleformat*{\section}{\bfseries\boldmath\Large}
\titleformat*{\subsection}{\bfseries\boldmath\large}

\numberwithin{equation}{section}

\usepackage{parskip}
\usepackage{tikz-cd}

\newtheoremstyle{customstyle}
  {6pt} %
  {2pt} %
  {\itshape} %
  {} %
  {\bfseries} %
  {.} %
  { .5em} %
  {} %

\newtheoremstyle{customdef}
  {6pt} %
  {0pt} %
  {} %
  {} %
  {\bfseries} %
  {.} %
  { .5em} %
  {} %

\usepackage{aliascnt}

\theoremstyle{customstyle}
\newtheorem{theorem}{Theorem}[section]

\newaliascnt{corollary}{theorem}
\newtheorem{corollary}[corollary]{Corollary}
\aliascntresetthe{corollary}

\newaliascnt{lemma}{theorem}
\newtheorem{lemma}[lemma]{Lemma}
\aliascntresetthe{lemma}

\newaliascnt{proposition}{theorem}
\newtheorem{proposition}[proposition]{Proposition}
\aliascntresetthe{proposition}

\theoremstyle{customdef}
\newtheorem{definition}{Definition}[section]
\newtheorem{example}{Example}[section] 
\newtheorem{remark}{Remark}[section]

\usepackage{cleveref}
\Crefname{theorem}{Theorem}{Theorems}
\Crefname{lemma}{Lemma}{Lemmas}
\Crefname{proposition}{Proposition}{Propositions}
\Crefname{corollary}{Corollary}{Corollaries}
\Crefname{definition}{Definition}{Definitions}
\Crefname{remark}{Remark}{Remarks}
\Crefname{example}{Example}{Examples}

\expandafter\def\expandafter\normalsize\expandafter{%
    \normalsize%
    \setlength\abovedisplayskip{6pt}%
    \setlength\belowdisplayskip{6pt}%
    \setlength\abovedisplayshortskip{4pt}%
    \setlength\belowdisplayshortskip{4pt}%
}

\usepackage{tikz}
\usetikzlibrary{calc}

\usepackage{braket}
\usepackage{enumerate}   
\usepackage{booktabs}

\newcommand{\cycl}{\stackrel{c_{ycl.}}{\scalebox{1.5}{\ensuremath\circlearrowleft}}}

\newcommand{\ZZ}{\mathbb{Z}}
\newcommand{\RR}{\mathbb{R}}
\newcommand{\MM}{\mathcal{M}}
\newcommand{\Trasl}{\mathcal{S}}
\newcommand{\Kop}{\mathcal{K}}
\newcommand{\gl}{\mathfrak{gl}}

\DeclareMathOperator{\SL}{SL}
\DeclareMathOperator{\Pol}{Pol}
\DeclareMathOperator{\diag}{diag}
\DeclareMathOperator{\ad}{ad}
\DeclareMathOperator{\nil}{nil}
\DeclarePairedDelimiter{\SB}{[\![}{]\!]}
\definecolor{electricviolet}{rgb}{0.56, 0.0, 1.0}

\usepackage[backend=biber,sorting=nyt,maxnames=99,giveninits=true,doi=false,citestyle=numeric-comp,url=true,doi=true,isbn=false]{biblatex}
\usepackage{xurl}

\newbibmacro*{journal+issuetitle+no-date}{%
  \usebibmacro{journal}%
  \setunit*{\addspace}%
  \printfield{series}%
  \setunit*{\addcomma\space}%
  \usebibmacro{volume+number+eid}%
  \setunit{\addcomma\space}%
  \printfield{issue}%
}

\newbibmacro*{note+pages+article}{%
  \printfield{pages}%
  \iffieldundef{note}
    {}
    {\setunit{\addcomma\space}\printtext[parens]{\printfield{note}}}%
  \setunit{\addspace}%
  \printtext[parens]{\printfield{year}}%
}

\DeclareBibliographyDriver{article}{%
  \usebibmacro{bibindex}%
  \usebibmacro{begentry}%
  \usebibmacro{author/translator+others}%
  \setunit{\printdelim{nametitledelim}}\newblock
  \usebibmacro{title}%
  \newunit
  \printlist{language}%
  \newunit\newblock
  \usebibmacro{byauthor}%
  \newunit\newblock
  \usebibmacro{bytranslator+others}%
  \newunit\newblock
  \printfield{version}%
  \newunit\newblock
  \usebibmacro{journal+issuetitle+no-date}%
  \setunit{\addcomma\space}%
  \usebibmacro{byeditor+others}%
  \setunit{\addcomma\space}%
  \usebibmacro{note+pages+article}%
  \newunit\newblock
  \iftoggle{bbx:isbn}
    {\printfield{issn}}
    {}%
  \newunit\newblock
  \usebibmacro{doi+eprint+url}%
  \newunit\newblock
  \usebibmacro{addendum+pubstate}%
  \setunit{\bibpagerefpunct}\newblock
  \usebibmacro{pageref}%
  \newunit\newblock
  \iftoggle{bbx:related}
    {\usebibmacro{related:init}%
     \usebibmacro{related}}
    {}%
  \usebibmacro{finentry}}

\title{Discrete differential-geometric Poisson brackets:
\\general theory and explicit constructions }

\usepackage{authblk}

\author[1]{Marta Dell'Atti}
\author[2]{Giorgio Gubbiotti}
\author[3]{Pierandrea Vergallo}

\affil[1]{\small Institute of Mathematics, University of Warsaw, Banacha 2, 02-097 Warsaw, Poland

\texttt{dellattimarta@gmail.com} 
\vspace{1.5ex} }

\affil[2]{Universit\`a degli Studi di Milano, Dipartimento di Matematica
``Federigo Enriques'', Via Cesare Saldini 50, 20133, Milano, Italy, \& INFN,
Sez.\ di Milano, Via Giovanni Celoria 16, 20133, Milano, Italy

\texttt{giorgio.gubbiotti@unimi.it} 
\vspace{1.5ex}}

\affil[3]{Department of Basic and Applied Sciences, University of Basilicata, 
V.le dell'Ateneo Lucano 10, 85100 Potenza, Italy \&
 INFN, 
Sez. di Napoli, Complesso Universitario di Monte S. Angelo, Edificio 6
Via Cintia
80126 Napoli (NA),
Italy

\texttt{pierandrea.vergallo@unibas.it} }

\date{\today}

\begin{document}

\allowdisplaybreaks

\maketitle

\begin{abstract}
We review the discrete theory of differential-geometric Poisson brackets as introduced by B.\ A.\ Dubrovin, and we present a complete proof of their characterisation in the non-degenerate case. We use this characterisation to show several novel examples of such structures, providing a complete classification for non-degenerate four-dimensional differential-geometric Poisson brackets under the additional assumption of the Lie algebra to be of quasi-Frobenius type. We also present a complete classification of quasi-Frobenius filiform  $\mathbb{N}$-graded Lie algebras, including two families in arbitrary dimensions. 
\end{abstract}

\setcounter{tocdepth}{2}

\renewcommand{\baselinestretch}{0.97}\normalsize
\tableofcontents
\renewcommand{\baselinestretch}{1.0}\normalsize

\section{Introduction}

In the late 80s, Boris A.\ Dubrovin obtained an interesting characterisation of the
non-degenerate Hamiltonian structures of $N$-component three-point differential-difference 
equations~\cite{Dubrovin1989}, also known as \emph{$N$-component Volterra-like 
equations}~\cite{Yamilov2006,LeviWinternitzYamilov2022Book}.
As pointed out by Dubrovin, this result was a semi-discrete extension of the
celebrated theory of the non-degenerate \emph{differential-geometric Poisson brackets} for 
\emph{hydrodynamic-type systems} introduced a
few years earlier by Sergei P.\ Novikov and Dubrovin himself~\cite{DubrovinNovikov1984}.
The geometric structure of hydrodynamic-type systems is intimately related to
the Riemannian geometry of flat manifolds, whereas its semi-discrete counterpart is related
to special cases of Poisson--Lie groups, Lie bi-algebras, and 
$r$-matrices~\cite{Kosmann2004,Semenov1983,Semenov2008,MartaRev}. In force
of this analogy, the Hamiltonian structures of $N$-component three-point 
differential-difference equations were called the \emph{discrete differential-geometric
Poisson brackets}.

Both the continuous and the discrete Hamiltonian structures are
very rich in terms of geometric and algebraic properties, but over 
the course of the years the continuous one drawn much more attention 
than its semi-discrete counterpart. Indeed, the results of~\cite{Dubrovin1989} 
were essentially forgotten with the exception of their appearance in
the Ph.D.\ thesis of Emanuele Parodi~\cite{ParodiThesis}, a student of Dubrovin himself who dealt with the problem of finding Hamiltonian structure
for higher-order semi-discrete equations~\cite{Parodi2012}. 
More recently, Dubrovin's results were reprised and discussed by 
Matteo Casati and Daniele Valeri~\cite{CasatiValeri}.

Additionally, in a very unfortunate manner, a detailed proof of Dubrovin's characterisation
of the discrete differential-geometric Poisson brackets has never been published. In the original 
publication by Dubrovin~\cite{Dubrovin1989} only the statement of the characterisation is
present. A more detailed introduction to the subject has been given in the review 
paper~\cite{DubrovinNovikov1989}, and a sketch of the proof has been presented in Parodi's Ph.D.\ 
thesis~\cite{ParodiThesis}, which otherwise dealt with a scalar higher-order analogue of Dubrovin's 
results. In the work of Casati and Valeri~\cite{CasatiValeri} some of the results of Dubrovin are 
presented again using the technique of Poisson $\Lambda$-brackets~\cite{DeSole_etal2020}.

In this paper, we address this gap and present a full proof of Dubrovin's characterisation
in the spirit of his original ideas. Far from being just a replication of known results we will
show how Dubrovin's construction naturally leads to the discovery of many new classes of
examples of discrete differential-geometric Poisson brackets, including infinite dimensional families.
Indeed, reinterpreting Dubrovin's results in a modern manner gives
access to all the works on classification and characterisation of Poisson--Lie groups, Lie bi-algebras,
and $r$-matrices made since the 90s. This gives to Dubrovin's own work a new actuality and
modernity. Moreover, we underline that in an upcoming work we will present another development of 
Dubrovin's theory by providing the general construction of
the continuum limits of the discrete differential-geometric Poisson brackets~\cite{contiDisc}.

The structure of the paper is the following: in~\Cref{sec:discretePB} we give
the basic definitions of the theory, starting from the semi-discrete analogue
of the loop spaces~\cite{Mokhov_1998}. In~\Cref{sec:geom} we give the proof
of Dubrovin's characterisation, including a fundamental corollary stating
that quasi-Frobenius Lie algebras naturally give rise to non-degenerate
discrete differential-geometric Poisson brackets. In~\Cref{sec:examples}
we use the aforementioned corollary to produce explicit examples of differential-geometric 
Poisson brackets. In particular, we give a complete 
classification of non-degenerate four-dimensional differential-geometric 
Poisson brackets, using the list of decomposable and indecomposable four-dimensional
Lie algebras as presented in~\cite{SnobWinternitz2017book}. Additionally, using the results
of~\cite{Millionschikov2004} on quasi-Frobenius $\mathbb{N}$-graded filiform Lie algebras, we give a complete
classification of non-degenerate differential-geometric 
Poisson brackets arising from the classification of these Lie algebras. This includes two families which are valid 
in arbitrary dimensions, namely the Lie algebras $\mathfrak{n}_{2k,1}$ and $\mathfrak{V}_{2k}$.
To the best of our knowledge, this is the first construction of 
a non-trivial infinite family of non-degenerate discrete differential-geometric Poisson brackets.
Finally, in~\Cref{sec:conclusions} we give some conclusions and outlook on the problems
left open by the resurgence of Dubrovin's theory.

\section{Discrete differential-geometric Poisson brackets}\label{sec:discretePB}

In this Section, we consider the discrete analogue of the differential-geometric 
Poisson brackets introduced by Dubrovin and Novikov in the continuous framework. 
The construction is realised on the discrete loop space $\mathcal M^N_\infty$, 
the lattice counterpart of loop space in the continuous theory~\cite{Mokhov_1998}, 
and following the approach of~\cite{Dubrovin1989,DubrovinNovikov1989}. We provide 
the notions of local functionals and local Poisson brackets on $\mathcal M^N_\infty$, 
and present two descriptions of the same structure: one based on brackets on 
coordinate functions, and one in terms of local difference operators. We then 
focus on the brackets of order~$1$, and state the conditions for such a bracket 
to be Poisson (\Cref{thm:dDNPBconds}), with the explicit proof in 
Appendix~\ref{app:dp} revising Parodi's results~\cite{ParodiThesis}. These 
conditions will be the starting point for the geometric interpretation of~\Cref{sec:geom}.

\subsection{General definition and construction}\label{sec:ddgPB_definition_construction}

Consider an $N$-dimensional manifold $\MM^N_n$ with local coordinates labelled by 
an integer $n\in\ZZ$, i.e.\ $\vb*{u}_n = (u^1_n,\ldots,u^N_n)$. The 
space we consider is given by a infinitely many copies of the above mentioned manifold:
\begin{equation}
    \MM_\infty^N = \bigoplus_{n\in\ZZ}\MM^N_n\,, 
\end{equation}
{the analogue of the loop space in the continuous case~\cite{Mokhov_1998}. }
We say that a function $f \in \mathcal{F}(\MM_\infty^N)$ is \emph{local} if it depends on a finite number of indices in $\ZZ$, that is
\begin{equation}
    \vb*{u}_{(n,m)} = \begin{cases}
        (\vb*{u}_n, \vb*{u}_{n+1}, \dots, \vb*{u}_{m-1}, \vb*{u}_m) \,, & n \le m\,, \\[1mm] 
        (\vb*{u}_m, \vb*{u}_{m+1}, \dots, \vb*{u}_{n-1}, \vb*{u}_n) \,, & n \ge m\,, 
    \end{cases}
\end{equation}
and we write $f=f_{(n,m)}:=f(\vb*{u}_{(n,m)})$ for brevity. 
In what follows we denote the space of local functions on
$\MM_{\infty}^{N}$ as $\mathcal{F}_\text{loc}(\MM_{\infty}^{N})$, and that of local functions of a given regularity $k$ as
$\mathcal{C}_\text{loc}^{k}(\MM_{\infty}^{N})$. In particular, the space of
smooth local functions will be denoted as
$\mathcal{C}_\text{loc}^{\infty}(\MM_{\infty}^{N})$, and the space of local
polynomials $\Pol_\text{loc}(\MM_{\infty}^{N})$.  

We emphasise that in this work, we will only consider local functions. On smooth local functions, we introduce a notion of local Poisson brackets. 

\begin{definition}\label{def:local_PB}
    A bilinear operation $\{\,\cdot\,\,,\,\cdot\,\} \colon
    \mathcal{C}^{\infty}_{\text{loc}}(\mathcal{M}^N_{\infty})\times\mathcal{C}^{\infty}_{\text{loc}}(\mathcal{M}^N_{\infty})
    \to \mathcal{C}^{\infty}_{\text{loc}}(\mathcal{M}^N_{\infty}) $ is a
    \emph{local Poisson bracket} if it satisfies the following conditions for every $f,g,h\in \mathcal{C}^\infty_{\text{loc}}(\mathcal{M}^N_{\infty})$:
\begin{enumerate}
    \item it is skew-symmetric:
    \begin{equation}\label{eq:cond_PB}
        \{ f\,, g\} = - \{ g\,, f\}; 
    \end{equation}
    \item satisfies the Jacobi identity, i.e.:
    \begin{equation}
        \{\{ f\,, g\}\,, h\} + \cycl =0\,
    \end{equation}
    where the symbol $\cycl$ refers to the sum on the cyclic permutations of the arguments;
    \item it is a derivation, i.e.\ it satisfies the Leibniz rule:
    \begin{equation}
        \pb{fg}{h} = f\pb{g}{h} + \pb{f}{h}g.
    \end{equation}
\end{enumerate}
\end{definition}

A way of constructing such a bracket is to consider a family of $2K+1$
matrices $g^{ij}_{\ell}$, with $\ell\in\set{-K,\dots, K}$, and local functions 
as coefficients: 
\begin{equation}
    \pb*{u_n^i}{u_m^j}_{\!(K)} = \begin{cases}
        g^{ij}_{(n-m)}(\vb*{u}_{(n,m)}),
        &\abs{n-m}\leq K, \\ 
    0,
    &\abs{n-m}>K,
    \end{cases}
\end{equation}
such that it satisfies the skew-symmetry condition: 
\begin{equation}
    \pb*{u_n^i}{u_m^j}_{\!(K)} = - \pb*{u_m^j}{u_n^i}_{\!(K)} ,
    \label{eq:ssbase}
\end{equation}
and the Jacobi identity:
\vspace*{-2ex}
\begin{equation}
\pb*{\pb*{u_n^i}{u_m^j}_{\!(K)}}{u_\ell^k}_{\!(K)} + 
        \cycl
        = 0.
        \label{eq:jacbase}
\end{equation}

Additionally, we say that a bracket of order~$K$ is \emph{homogeneous}, denoted as $\{\,\cdot\,\,,\, \cdot\,\}_{[K]}$, if it has the following expression:
\begin{equation}\label{eq:PB_order_k_only}
    \pb*{u_n^i}{u_m^j}_{\![K]} = \begin{cases}
        g^{ij}_{(n-m)}(\vb*{u}_{(n,m)}),
       &\quad \abs{n-m} = K, \\ 
    0,
    &\quad \abs{n-m} \neq K,
    \end{cases}
\end{equation}
i.e.\ it is non-zero only if two points on the lattice are \emph{exactly} at distance
$K$. This allows us to write a generic bracket of order~$K$ as the sum of homogeneous brackets
of order~$\ell$ with $0 \leq \ell \leq K$, i.e.:
\begin{equation}
    \pb*{u_n^i}{u_m^j}_{\!(K)}
    =
    \sum_{\ell=0}^{K}\pb*{u_n^i}{u_m^j}_{\![\,\ell\,]}.
\end{equation}

A Poisson bracket is \emph{degenerate} if the homogeneous highest order~is realised by a degenerate matrix, otherwise it is said non-degenerate.

If a bracket defined on the coordinate functions satisfies~\eqref{eq:ssbase}
and~\eqref{eq:jacbase}, then it can be extended on products of coordinate functions as
a derivation:
\begin{equation}
    \pb*{u^{i}_{n}u^{j}_{m}}{u_{\ell}^{k}}_{\!(K)} =
    u^{i}_{n}\pb*{u^{j}_{m}}{u_{\ell}^{k}}_{\!(K)} +
    \pb*{u^{i}_{n}}{u_{\ell}^{k}}_{\!(K)}u^{j}_{m}.
    \label{eq:derbase}
\end{equation}
The iterated application of this rule allows one to introduce a bracket on local
polynomials as well, i.e.\ on $\Pol_\text{loc}(\MM_{\infty}^{N})$ and, by density,
on local smooth function through the following expression:  
\begin{equation}
    \pb{f_{(k,\ell)}}{g_{(k',\ell')}}_{\!(K)} = 
    \sum_{n,m\in\ZZ}
    \pdv{f_{(k,\ell)}}{u_n^i}\,\pb*{u_n^i}{u_m^j}_{\!(K)}\,\pdv{g_{(k',\ell')}}{u_m^j}\,,
    \quad 
    f_{(k,\ell)},g_{(k',\ell')}\in\mathcal{C}^{\infty}_\text{loc}(\MM_{\infty}^{N}).
    \label{eq:pblocfunc}
\end{equation}
We note that the sum in~\eqref{eq:pblocfunc} is in fact finite,
since the functions are local. That is, for fixed pairs~$(k,\ell)$, $(k',\ell')$, the computations are analogous to those for the standard Poisson brackets defined on finite dimensional spaces.    

By construction, \eqref{eq:pblocfunc} defines a bracket on local
functions. To prove that it is a local Poisson
bracket in the sense of Definition~\ref{def:local_PB}, t is enough to verify the two properties \eqref{eq:ssbase} and \eqref{eq:jacbase} on coordinate functions, then the corresponding properties for arbitrary local functions follow. 
A Poisson bracket introduced in this way, namely the bracket~\eqref{eq:pblocfunc}, is 
called \emph{discrete differential-geometric Poisson}~(dDGP) bracket of order~$(K)$. 

Alternatively, one can take into account the locality condition
by writing the elementary dDGP bracket using Kronecker deltas:
\begin{equation}
    \pb*{u_n^i}{u_m^j}_{\!(K)} =
    \sum_{k=-K}^{K} g^{ij}_{(k)}(\vb*{u}_{(n,n+k)})\,\delta_{n,m-k}. 
    \label{eq:pbkronbas0}
\end{equation}
We introduce the shift operator $\Trasl\colon \mathcal{M}_\infty^N \to \mathcal{M}_\infty^N $ acting on the coordinate functions as 
\begin{equation}
    \Trasl\, u^i_n = u^i_{n+1} \,,  
\end{equation}
and then extended by the homomorphism property. 
It is easy to see from the skew-symmetry condition on the elementary
brackets~\eqref{eq:pbkronbas0} that the following necessary conditions on
the functions $g_{(k)}$ in~\eqref{eq:pbkronbas0} hold: 
\begin{equation}
    g^{ij}_{(-k)}(\vb*{u}_{(n,n-k)}) = -\Trasl^{-k}g^{ji}_{(k)}(\vb*{u}_{(n,n+k)}). 
\end{equation}
From this last equality, we write for the sake of simplicity:
\begin{equation}
    \left(g_{(k)}\right)^{ij}_{n} \coloneqq  
    g^{ij}_{(k)}(\vb*{u}_{(n,n+k)}), 
\end{equation}
so that a necessary condition for an elementary dDGP bracket~\eqref{eq:pbkronbas0} 
is to be of the following form:
\begin{equation}
    \pb*{u_n^i}{u_m^j}_{\!(K)} = \left(g_{(0)}\right)^{ij}_n\,\delta_{n,m}+
    \sum_{k=1}^{K}
\left[\left(g_{(k)}\right)^{ij}_{n}\delta_{n,m-k}-\left(g_{(k)}\right)^{ji}_{n-k}\delta_{n,m+k}\right]. 
    \label{eq:pbkronbas1}
\end{equation}
In the following, we will consider this to be the standard form of 
the elementary dDGP bracket of order $(K)$, albeit it is important to note
that this form is not sufficient, as the Jacobi identity is not automatically
satisfied.

Finally, we recall that given a local Poisson bracket, any local function 
$f_{(n,n+k)}\in \mathcal{C}^{\infty}_\text{loc}(\MM^{N}_\infty)$ defines an infinite-dimensional dynamical system with respect to a continuous 
``time'' variable $t$ as:
\begin{equation}
    \dot{u}_n^i = \pb*{u_n^i}{f_{(n,n+k)}},
    \label{eq:hamevol}
\end{equation}
where with the dot we mean the derivative with respect to $t$. In this
context, the local function $f_{(n,n+k)}$ is called a \emph{Hamiltonian density}.

\subsection{Operatorial form}

Analogously to its continuous counterpart, the discrete differential geometric Poisson bracket
can be induced by a \emph{local difference operator}, i.e.\ an operator written as a combination of finitely many powers of 
the shift operator $\Trasl$:
\begin{equation} \label{eq:diffopgen}
    \mathcal{K}^{ij} = \sum_{\ell = K'}^{K}
    (\kappa_{(\ell)})^{ij}\,\Trasl^k,
\end{equation}
where $K'\leq K$ are integers and $\kappa_{(\ell)}$ matrices with local functions as coefficients. Then, one associates with the difference operator $\mathcal{K}$ the following bracket
on local functions:
\begin{equation}
    \pb{f_{(k,\ell)}}{g_{(k',\ell')}}_{\mathcal{K}} = 
    \sum_{n,m\in\ZZ}
    \pdv{f_{(k,\ell)}}{u_n^i}\,\mathcal{K}^{ij}\!\left(\pdv{g_{(k',\ell')}}{u_m^j}\right),
    \quad 
    f_{(k,\ell)},g_{(k',\ell')}\in\mathcal{C}^{\infty}_\text{loc}(\MM_{\infty}^{N}).
    \label{eq:pblocfuncop}
\end{equation}
If the bracket defined in this way is a local Poisson bracket in the sense of
\Cref{def:local_PB}, we say that the difference operator $\mathcal{K}$ is a
Hamiltonian operator. Moreover, a Hamiltonian operator is said to be \emph{degenerate} 
if the associated Poisson bracket is degenerate, and non-degenerate otherwise.

The approach with difference operators is completely equivalent to the one with the elementary
brackets on coordinate functions, and one can go back and forth from one
to the other. Indeed, on the one hand it is enough to compute the elementary brackets
with formula~\eqref{eq:pblocfuncop}, and on the other hand one uses the matrices~$g_{(k)}$ to construct the operator. In particular, following the discussion in~\Cref{sec:ddgPB_definition_construction}, to be Hamiltonian, the form of a discrete local operator~\eqref{eq:diffopgen} is restricted: the lower bound in the sum is $K'=-K$, and the  matrices  $\kappa_{(\ell)}$ must satisfy the same conditions as the matrices~$g_{(k)}$, hence we can safely identify the two objects. Summarising, a discrete Hamiltonian operator has the following form:
\begin{equation}
    \mathcal{K}^{ij} = \left(g_{(0)}\right)^{ij}_n
    + \sum_{k=1}^{K}
\left[\left(g_{(k)}\right)^{ij}_{n}\Trasl^{k}-\left(g_{(k)}\right)^{ji}_{n-k}\Trasl^{-k}\right]. 
    \label{eq:hamopgen}
\end{equation}
For instance, this implies that the degeneracy condition can be expressed by
saying that the matrix~$g_{(k)}$ is degenerate.

\begin{remark}
    The theory of discrete Hamiltonian operators can be carried
    out in a completely algebraic way using the concepts we explained in this
    subsection, together with those of \emph{functionals} and equivalence
    classes, see e.g.~\cite{CasatiWang2020,CasatiValeri}. However, in this paper we
    will not make use of this approach.
\end{remark}

\medskip 

Let us now see some examples. In the following, we set $h:=g_{(0)}$ and $g:=g_{(1)}$.

\begin{example} \label{ex:volterra}
    Consider the following first order~difference operator on a one-dimensional 
    space with local coordinates $\vb*{u}_n\equiv u_n$:
    \begin{equation}
        \mathcal{K} = u_{n}u_{n+1}\,\Trasl - u_{n}u_{n-1}\,\Trasl^{-1}.
        \label{eq:volterraop}
    \end{equation}
    In such a case we have that $g$ is a $1\times1$ matrix just given by $g_n = u_{n}u_{n+1}$. We can rewrite this operator in bracket form as a bracket of order~$1$
    with only one non-zero relations:
    \begin{equation}
             \pb{u_n}{u_{n+1}}_{\!(1)} = u_n u_{n+1}.
        \label{eq:volterrabracket}
    \end{equation}
    It is possible to prove directly that such a bracket satisfy the Jacobi
    identity. Indeed, the only possible non-zero relation in the Jacobi identity
    is the following:
    \begin{equation}
        \pb*{\pb*{u_n}{u_{n+1}}_{\!(1)}}{u_{n+2}}_{\!(1)}
        + \pb*{\pb*{u_{n+1}}{u_{n+2}}_{\!(1)}}{u_n}_{\!(1)}
        + \pb*{\pb*{u_{n+2}}{u_{n}}_{\!(1)}}{u_{n+1}}_{\!(1)} =0.
        \label{eq:volterrajac}
    \end{equation}
    The last term is identically zero since $\pb*{u_{n}}{u_{n+2}}\equiv0$.
    Using~\eqref{eq:volterrabracket} in~\eqref{eq:volterrajac}:
    \begin{equation}
        \pb*{u_{n}u_{n+1}}{u_{n+2}}_{\!(1)}
        + \pb*{u_{n+1}u_{n+2}}{u_n}_{\!(1)}
        =0.
        \label{eq:volterrajac2}
    \end{equation}
    Using the derivation property~\eqref{eq:derbase} on the first term 
    of~\eqref{eq:volterrajac2}:
    \begin{equation}
        \begin{aligned}
            \pb*{u_{n}u_{n+1}}{u_{n+2}}_{\!(1)}
            &= u_{n}\underbrace{\pb*{u_{n+1}}{u_{n+2}}_{\!(1)}}_{=u_{n+1}u_{n+2}}
            +\underbrace{\pb*{u_{n}}{u_{n+2}}_{\!(1)}}_{=0}u_{n+1}
            \\
            &=u_{n}u_{n+1}u_{n+2}.
        \end{aligned}
        \label{eq:volterrajac2a}
    \end{equation}
    In the same way on the second term of~\eqref{eq:volterrajac2}:
    \begin{equation}
        \begin{aligned}
            \pb*{u_{n+1}u_{n+2}}{u_{n}}_{\!(1)}
            &= u_{n+1}\underbrace{\pb*{u_{n+2}}{u_{n}}_{\!(1)}}_{=0}
            +\underbrace{\pb*{u_{n+1}}{u_{n}}_{\!(1)}}_{=-u_{n}u_{n+1}}u_{n+2}
            \\
            &=-u_{n}u_{n+1}u_{n+2}.
        \end{aligned}
        \label{eq:volterrajac2b}
    \end{equation}
    This shows that~\eqref{eq:volterrajac2} is in fact an identity, proving
    that the bracket~\eqref{eq:volterrabracket} is a Poisson bracket.

    The Hamiltonian operator~\eqref{eq:volterraop}, or its bracket 
    equivalent~\eqref{eq:volterrabracket}, is the Hamiltonian structure
    of the so-called Volterra lattice:
    \begin{equation}
        \dot{u}_{n}= u_n(u_{n+1}-u_{n-1}),        
        \label{eq:volterra}
    \end{equation}
    with the Hamiltonian density:
    \begin{equation}
        f_n = u_n.
    \end{equation}
    For additional properties of the Volterra lattice we refer to~\cite[\S 4.1]{KhanizadehMikhailovWang},
    \cite[\S 3.1]{Yamilov2006}, and the book~\cite[\S 3.3]{LeviWinternitzYamilov2022Book}.
\end{example}

\begin{example}
    Consider the following first order~difference operator on a two-dimensional 
    space with local coordinates $\vb*{u}_n=(u_n,v_n)$:
    \begin{equation}
        \mathcal{K}
        =
        \begin{pmatrix}
            0 & u_n(\Trasl -1)
            \\
            (1-\Trasl)u_n & 0
        \end{pmatrix}.
        \label{eq:optoda1st}
    \end{equation}
    From this form we obtain the following values for the matrices
    $g_n$ and $h_n$:
    \begin{equation}
        g_n =
        \begin{pmatrix}
            0 & u_n
            \\
            0 & 0
        \end{pmatrix},
        \qquad
        h_n =
        \begin{pmatrix}
            0 & -u_n
            \\
            u_n & 0
        \end{pmatrix}.
        \label{eq:toda1ststruct}
    \end{equation}
    The same structure can be expressed in bracket form as a bracket of order~$1$
    with non-zero relations:
    \begin{equation}
        \begin{array}{ll}
             \pb{u_n}{v_{n+1}}_{\!(1)} = u_n, & \pb{u_n}{v_n}_{\!(1)} = -u_n.
        \end{array}
        \label{eq:todabracket1}
    \end{equation}
    We observe that the matrix $g_n$ in~\eqref{eq:toda1ststruct} is degenerate, hence the bracket and the operator
    are both degenerate. However, like in~\Cref{ex:volterra} it can be shown by an explicit 
    computation that the bracket~\eqref{eq:todabracket1} satisfy the Jacobi identity.

    The Hamiltonian operator~\eqref{eq:optoda1st} is known in the literature as
    the first Hamiltonian structure of the Toda lattice~\cite{Toda1967}:
    \begin{equation}
        \left\{\begin{array}{l}
        \dot{u}_{n}= u_n(v_{n+1}-v_n) \\[1mm]
        \dot{v}_{n}= u_n-u_{n-1},
        \end{array}\right.        
        \label{eq:toda}
    \end{equation}
    with the Hamiltonian density:
    \begin{equation}
        f = u_n + \frac{v_n^2}{2}.
    \end{equation}
    We recall that the Toda lattice in the form~\eqref{eq:toda} is written
    in the so-called Flaschka--Manakov 
    variables~\cite{Flaschka1974int,Flaschka1974lax,Manakov1974},
    which are defined from an infinite number of pairs $(q_i,p_i)$ of conjugate variables as:
    \begin{equation}
        u_n = \exp(q_{n+1}-q_n),
        \qquad 
        v_n = p_n;
        \label{eq:fmvars}
    \end{equation}
    In fact, this Hamiltonian structure comes from transferring the usual 
    symplectic structure of the physical variables $\pb{q_i}{p_j}=\delta_{ij}$ 
    into the Flaschka--Manavok variables. For additional properties
    of the Toda lattice we refer to~\cite[\S 4.6]{KhanizadehMikhailovWang},
    \cite[\S 3.2]{Yamilov2006}, and the book~\cite[\S 3.2]{LeviWinternitzYamilov2022Book}.
    
    Moreover, it is known, for instance see again~\cite[\S 4.6]{KhanizadehMikhailovWang}, that the Toda lattice admits a second 
    Hamiltonian structure, which in operatorial form reads as:
    \begin{equation}
        \mathcal{K} =
        \begin{pmatrix}
            u_n(\Trasl-\Trasl^{-1})u_n & u_n(\Trasl-1)v_n
            \\[1mm]
            v_n(1-\Trasl^{-1})u_n & u_n\Trasl - \Trasl^{-1}u_n
        \end{pmatrix}.
    \end{equation}
    From this form we obtain the following values for the matrices
    $g_n$ and $h_n$:
    \begin{equation}
        g_n =
        \begin{pmatrix}
            u_n u_{n+1} & u_n v_{n+1}
            \\[1mm]
            0 & u_n
        \end{pmatrix},
        \qquad
        h_n =
        \begin{pmatrix}
            0 & -u_nv_n
            \\[1mm]
            u_nv_n & 0
        \end{pmatrix}.
        \label{eq:toda2ndstruct}
    \end{equation}
    Observe that in this case the matrix $g_n$ is non-degenerate. 
    The same structure can be expressed in bracket form as a bracket of order~$1$
    with non-zero elementary dDGP bracket of order $(1)$:
    \begin{equation}
        \begin{aligned}
             \pb*{u_{n}}{u_{n+1}}_{\!(1)} &= u_{n}u_{n+1}, 
             &~~
             \pb*{u_n}{v_{n+1}}_{\!(1)} &= u_n v_{n+1},
             \\[1mm]
             \pb*{v_n}{v_{n+1}}_{\!(1)} &= u_n,
             &
             \pb*{u_n}{v_n}_{\!(1)} &= -u_n v_n.
        \end{aligned}
        \label{eq:todabracket2}
    \end{equation}
    Like in~\Cref{ex:volterra}, also in this case it can be shown by an explicit 
    computation that the bracket~\eqref{eq:todabracket1} satisfy the Jacobi identity.
    Finally, the Hamiltonian density associated with this structure is simply:
    \begin{equation}
        f = v_n.
    \end{equation}
    \label{ex:toda}
\end{example}

\begin{example}
    Consider the following first order~difference operator acting
    on a 3-dimensional space with local coordinates $\vb*{u}_n=(u_n,v_n,w_n)$:
    \begin{equation}
        \mathcal K
        =
        \begin{pmatrix}
            \Trasl-\Trasl^{-1} & 0 & 0
            \\
            0 & 0 & (\Trasl^{-1}-1)w_n
            \\
            0 & -w_n(\Trasl-1) & 0
        \end{pmatrix}.
    \end{equation}
    It is known that this operator is Hamiltonian, see e.g.~\cite[\S 4.19]{KhanizadehMikhailovWang}. The value of the matrices $g_n$ and $h_n$ is
    the following
    \begin{equation}\label{eq:three_dim_example}
        g_n =
        \begin{pmatrix}
            1&0&0
            \\
            0&0&0
            \\
            0&-w_n&0
        \end{pmatrix},
        \qquad
        h_n =
        \begin{pmatrix}
            0&0&0
            \\
            0&0&-w_n
            \\0&w_n&0
        \end{pmatrix}.
    \end{equation}
    Note that like in the case of the first Hamiltonian structure of the 
    Toda lattice~\eqref{eq:toda1ststruct} the matrix~$g_n$ in~\eqref{eq:three_dim_example} is degenerate.
    Additionally, the same structure can be expressed in bracket form as 
    a bracket of order~$1$ with non-zero relations:
    \begin{equation}
        \begin{array}{lll}
             \pb*{u_{n}}{u_{n+1}}_{\!(1)} = 1, 
             &~~
             \pb*{w_n}{v_{n+1}}_{\!(1)} = -w_{n},
             &~~
             \pb*{v_n}{w_{n}}_{\!(1)} = -w_n.
        \end{array}
        \label{eq:bmbracket}
    \end{equation}
    Like in~\Cref{ex:volterra}, also in this case it can be shown by an explicit 
    computation that the bracket~\eqref{eq:bmbracket} satisfy the Jacobi identity.
    
    This example is the Hamiltonian structure of the so-called 
    \emph{B{\l}aszak--Marciniak lattice}~\cite{BlaszakMarciniak1994}:
    \begin{equation}
        \left\{\begin{array}{l}
        \dot{u}_{n}=w_{n+1}-w_{n-1} \\[1mm]
        \dot{v}_{n}=u_{n-1} w_{n-1}-u_n \,w_n \\[1mm]
        \dot{w}_{n}=w_n\left(v_n-v_{n+1}\right)
        \end{array}\right.        
    \end{equation}
    with Hamiltonian density:
    \begin{equation} \label{eq:BM_lattice_discrete}
        f=u_n w_n+\frac{1}{2} v^{2}_n.
    \end{equation}
\end{example}

\subsection{\texorpdfstring{Order ${(1)}$: Dubrovin--Novikov analogue}{order1}}\label{sec:ham_conds}
The main case of interest in this paper is the class of dDGP brackets~\eqref{eq:pbkronbas1} of order $(1)$, i.e. the Dubrovin--Novikov Poisson brackets:
\begin{equation}
    \pb*{u_n^i}{u_m^j}_{\!(1)} =
    g^{ij}_{n}\,\delta_{n,m-1}-g^{ji}_{n-1}\,\delta_{n,m+1}+ h_{n}^{ij}\,\delta_{n,m},
    \label{eq:discretePB}
\end{equation}
where we introduced $h \coloneqq g_{(0)}$ and $g\coloneqq g_{(1)}$. Alternatively, we can write:
\begin{equation}
    \pb*{u_n^i}{u_{n+1}^j}_{\!(1)} =
    g^{ij}_{n},
    \qquad
    \pb*{u_n^i}{u_{n}^j}_{\!(1)} =
    h_{n}^{ij},
    \label{eq:discretePBexp}
\end{equation}
or in operatorial form~\eqref{eq:hamopgen}:
\begin{equation}
    \Kop^{ij} =  g_n^{ij}\,\Trasl - g_{n-1}^{ji}\,\Trasl^{-1} + h^{ij}_n\,. 
    \label{eq:ddnop}
\end{equation}

We emphasise that the bracket $\{\,\cdot\,\,,\,\cdot\,\}_{(1)}$ is given by the linear combination of the brackets for the homogeneous orders $\{\,\cdot\,\,,\,\cdot\,\}_{[\,1\,]}$ and $\{\,\cdot\,\,,\,\cdot\,\}_{[\,0\,]}$. Moreover, since from
now on we will only deal with this kind of bracket, we will not make use
of the subscript $(1)$ anymore.

The following theorem provides necessary and sufficient conditions for a bracket of this form to be a dDGP bracket, and hence the conditions for the associated operator~\eqref{eq:ddnop} to be Hamiltonian.

\begin{theorem}\label{thm:dDNPBconds}
    The bracket~\eqref{eq:discretePB} is a Poisson bracket if and only if the following conditions on $g_n$, $h_n$ are satisfied: 
    \begin{subequations}
        \begin{align}
            \pdv{g_n^{ij}}{u_{n+1}^s}g_{n+1}^{sk}
            -
            g_{n}^{is}\pdv{g_{n+1}^{jk}}{u_{n+1}^s}&=0,
            \label{eq:pdncond1}
            \\
            \pdv{h^{ij}_n}{u_{n}^{s}}g_{n}^{sk}
            +
            \pdv{g_{n}^{jk}}{u_{n}^{s}}h_{n}^{si}
            -
            \pdv{g_{n}^{jk}}{u_{n+1}^{s}}g_{n}^{is}
            -
            \pdv{g_{n}^{ik}}{u_{n}^{s}}h_{n}^{sj}
            +
            \pdv{g_{n}^{ik}}{u_{n+1}^{s}}g_{n}^{js}
            &=0,
            \label{eq:pdncond2}
            \\
            \pdv{h^{ij}_{n}}{u_{n}^{s}}g_{n-1}^{ks}
            +\pdv{g_{n-1}^{kj}}{u_{n-1}^{s}}
            g_{n-1}^{si}
            +
            \pdv{g_{n-1}^{kj}}{u_{n}^{s}}h_{n}^{si}
            -\pdv{g_{n-1}^{ki}}{u_{n-1}^{s}}
            g_{n-1}^{sj}
            -
            \pdv{g_{n-1}^{ki}}{u_{n}^{s}}h_{n}^{sj}
            &=0,
            \label{eq:pdncond3}
            \\
            \pdv{h^{ij}_n}{u_{n}^s}h^{sk}_n
            +
            \pdv{h^{jk}_n}{u_{n}^s}h^{si}_n
            +
            \pdv{h^{ki}_n}{u_{n}^s}h^{sj}_n
            &=0\,.  
            \label{eq:pdncond4}
        \end{align}
        \label{eq:pdncond}
    \end{subequations}
\end{theorem}

\Cref{thm:dDNPBconds} was originally announced without an explicit proof by Dubrovin
in~\cite{Dubrovin1989}, together with a geometric characterisation of the bracket in the 
case where $g_n$ is non-degenerate (see also~\Cref{sec:geom}). The same result
appeared without proofs in the review again by Dubrovin and Novikov~\cite[\S3]{DubrovinNovikov1989}.
Moreover,~\Cref{thm:dDNPBconds} 
is present as a sub-result in Parodi's Ph.D.\ thesis~\cite[Theorem 2.4.4]{ParodiThesis},
making explicit the ideas of~\cite{Dubrovin1989}. More recently, such a result was derived
again by Casati and Valeri in~\cite[Theorem 7]{CasatiValeri}, with a different proof
bases on the theory of $\Lambda$-Poisson brackets~\cite{DeSole_etal2020}. To be more specific,
upon the identification $A^{ij}\equiv g^{ij}_{n}$ and $B^{ij}\equiv h_n^{ij}$, 
equation~\eqref{eq:pdncond1} coincides with \cite[eq.(3.11a)]{CasatiValeri}, \eqref{eq:pdncond2} with 
\cite[eq.(3.11c)]{CasatiValeri} and \eqref{eq:pdncond4} with \cite[eq.(3.11d)]{CasatiValeri}. The condition~\eqref{eq:pdncond3} is equivalent 
to \cite[eq.(3.11b)]{CasatiValeri} upon the application of the translation operator~$\Trasl$. We present the proof of~\Cref{thm:dDNPBconds} in Appendix~\ref{app:dp} in the line of the original
proof carried out by Dubrovin and Parodi.

\begin{remark}\label{conditions}
The four conditions~\eqref{eq:pdncond} are not symmetric in $g_n$ and $h_n$. Condition~\eqref{eq:pdncond4} is the Jacobi identity for the bracket $h_n$ alone: indeed, $h_n$ defines a Poisson bracket on each site. This is equivalent to require that the tensor $h^{ij}_n$ is Poisson, i.e.\ the Schouten brackets between $h_n$ and itself vanishes. Condition~\eqref{eq:pdncond1} is the only Jacobi-type condition involving $g_n$ alone, and it comes from the triple 
$(u_n,u_{n+1},u_{n+2})$, in which all three points are pairwise neighbours. The remaining conditions~\eqref{eq:pdncond2} and~\eqref{eq:pdncond3} are mixed: they encode the compatibility between $h_n$ and $g_n$. 
\end{remark}

Before going on, we present some examples of solution of the conditions given in~\Cref{thm:dDNPBconds}.

\begin{example}
    Let us consider the case $N=1$, i.e.\ $\vb*{u}_n\equiv u_n$. In such a case the skew-symmetry of
    $h_n$ implies $h_n\equiv 0$, so that the only condition surviving is~\eqref{eq:pdncond1},
    which writing $g_n\equiv g^{11}_n$ reads as:
    \begin{equation}
            {\pdv{g}{u_{n+1}}}(u_n,u_{n+1}) g(u_{n+1},u_{n+2})
            =
            g(u_n,u_{n+1}){\pdv{g}{u_{n+1}}} (u_{n+1},u_{n+2}).
            \label{eq:pdncond1onedim}
    \end{equation}
    This condition can be easily solved by ``separating the variables'', a technique
    that we will see in~\Cref{sec:geom} how to extend to the general case. Indeed, dividing
    by $g(u_n,u_{n+1})g(u_{n+1},u_{n+2})$ equation~\eqref{eq:pdncond1onedim} can be
    rewritten as:
    \begin{equation}
        \frac{1}{g(u_n,u_{n+1})}{\pdv{g}{u_{n+1}}}(u_n,u_{n+1})
        =
        \frac{1}{g(u_{n+1},u_{n+2})}{\pdv{g}{u_{n+1}}} (u_{n+1},u_{n+2}),
        \label{eq:pdncond1onedimb}
    \end{equation}
    implying that both sides equate to an arbitrary function of $u_{n+1}$ only, i.e.:
    \begin{equation}
        \frac{1}{g(u_n,u_{n+1})}{\pdv{g}{u_{n+1}}}(u_n,u_{n+1})
        = p(u_{n+1}).
        \label{eq:pdncond1onedimc}
    \end{equation}
    Equation~\eqref{eq:pdncond1onedimc} is readily solved to give:
    \begin{equation}
        g(u_n,u_{n+1}) = a(u_n)b(u_{n+1}),
        \label{eq:pdncond1onedimd}
    \end{equation}
    where $a$ is an arbitrary function of its argument we used the arbitrariness of 
    $p(u_{n+1})$ to write 
    $b(u_{n+1})=\exp(\int^{u_{n+1}} p(\upsilon)\dd \upsilon)$. 
    Inserting~\eqref{eq:pdncond1onedimd} in~\eqref{eq:pdncond1onedimb} we obtain:
    \begin{equation}
        \frac{b'(u_{n+1})}{b,u_{n+1})}
        =
        \frac{a'(u_{n+1})}{a(u_{n+1})},
    \end{equation}
    that is:
    \begin{equation}
        b(\upsilon) = C a(\upsilon),
        \quad
        C\in\RR.
    \end{equation}
    Absorbing the constant $C$ into the arbitrary function $a$ we obtain
    that:
    \begin{equation}
        g(u_n,u_{n+1}) = a(u_n)a(u_{n+1}).
        \label{eq:pdncond1onedimfin}
    \end{equation}
    Therefore, the most general discrete differential-geometric Hamiltonian operator in one component has the following form:
    \begin{equation}
        \mathcal{K} =  a(u_n)a(u_{n+1}) \Trasl - a(u_{n-1})a(u_{n})\Trasl^{-1},
        \label{eq:pdncond1onedimop}
    \end{equation}
    i.e.\ it is a generalisation of the operator of the Volterra equation
    considered in~\Cref{ex:volterra}, which is obtained for $a(\upsilon)=\upsilon$.
    This gives a simple proof of a result that was presented also in the
    context of the Poisson $\Lambda$-brackets in~\cite{DeSole_etal2020} and in 
    the one of Poisson cohomology in~\cite{CasatiWang2020}.
\end{example}

\begin{remark}
    It is worth observing that, following again~\cite{DeSole_etal2020,CasatiWang2020},
    the operator of the preceeding example is in fact equivalent to the constant
    operator:
    \begin{equation}
        \mathcal{K}_c = \Trasl - \Trasl^{-1},
    \end{equation}
    upon the gauge transformation:
    \begin{equation}
        u_n \longrightarrow \widetilde{u}_n = \int^{u_n} \frac{\dd \upsilon}{a(\upsilon)}.
    \end{equation}
    \label{rem:gauge}
\end{remark}

\begin{example}
    Let us consider the case of $N$ components, but assume that 
    $g$ is constant matrix. Then, the condition~\eqref{eq:pdncond1}
    is identically satisfied, and in turn conditions (\ref{eq:pdncond2}--\ref{eq:pdncond4}) reduce to:
    \begin{equation}
            \pdv{h^{ij}_n}{u_{n}^{s}}g^{sk}
            =0,
            \qquad
            \pdv{h^{ij}_n}{u_{n}^s}h^{sk}_n
            +
            \pdv{h^{jk}_n}{u_{n}^s}h^{si}_n
            +
            \pdv{h^{ki}_n}{u_{n}^s}h^{sj}_n
            =0.  
        \label{eq:pdncondconst}
    \end{equation}
    In particular, observe that since $g$ is constant $g_{n}=g_{n-1}$, and conditions~\eqref{eq:pdncond2} and~\eqref{eq:pdncond3} produce the
    same equation. If we further assume that $g$ is non-degenerate, then the first condition in~\eqref{eq:pdncondconst} multiplying
    by $g^{-1}$ reduces to $\pdv*{h^{ij}_n}{u_{n}^{s}}=0$, that is $h$ is
    any constant skew-symmetric matrix. In the degenerate case, non-trivial
    solutions are possible considering the kernel of the matrix $g$.
    \label{ex:constg}
\end{example}

\begin{remark}
    We remark that~\Cref{ex:constg} is analogous to what happens in the continuous
    setting for (non-homogeneous) first-order~Hamiltonian differential operators:
    to a constant metrics all constant skew-symmetric $(2,0)$ tensors are compatible,
    i.e.\ are Killing--Yano tensors~\cite{Yano1952}. The converse does not hold in general in the continuous context, indeed for constant metrics $g$ all the Killing-Yano tensors are linear (see \cite{GOSV_lie}).
\end{remark}

\begin{example}
    Let us assume that the matrix $g$ is diagonal (not necessarily constant),
    i.e.:
    \begin{equation}
        g = \diag (G^{1}(\vb*{u}_{n},\vb*{u}_{n+1}),\ldots,G^{N}(\vb*{u}_{n},\vb*{u}_{n+1})).
        \label{eq:gdiag}
    \end{equation}
    We show that this implies that the functions $G^{i}$ factorise
    as in~\eqref{eq:pdncond1onedimfin}, i.e.\ $G^{i}=a^{i}(u_n^i)a^{i}(u_{n+1}^i)$
    and that $h^{ij}=C^{ij}a^{i}(u_n^i)a^{j}(u_n^j)$, where $C^{ij}$ are constants
    and no summation is implied.

    Indeed, let us start analysing~\eqref{eq:pdncond1} with $g$ given by~\eqref{eq:gdiag}. We have:
    \begin{subequations}
        \begin{align}
            0&={\pdv{G^{i}}{u^{i}_{n+1}}}(\vb*{u}_{n},\vb*{u}_{n+1})
            G^{1}(\vb*{u}_{n+1},\vb*{u}_{n+2})
        -G^{1}(\vb*{u}_{n},\vb*{u}_{n+1})
        {\pdv{G^{i}}{u^{i}_{n+1}}}(\vb*{u}_{n+1},\vb*{u}_{n+2}),
        \label{eq:cond1diaga}
        \\
        0&={\pdv{G^{i}}{u^{j}_{n+1}}}(\vb*{u}_{n},\vb*{u}_{n+1}) 
        G^{j}(\vb*{u}_{n+1},\vb*{u}_{n+2}), \quad i\neq j,
        \label{eq:cond1diagb}
        \\
        0&={\pdv{G^{i}}{u^{j}_{n+1}}}(\vb*{u}_{n+1},\vb*{u}_{n+2}) 
        G^{j}(\vb*{u}_{n},\vb*{u}_{n+1}), \quad i\neq j.
        \label{eq:cond1diagc}
        \end{align}
    \end{subequations}
    The non-degeneracy condition implies that $G^{i}\neq 0$ for all $i$,
    so that equations~\eqref{eq:cond1diagb} and~\eqref{eq:cond1diagc} yield
    $G^{i}=G^{i}(u^i_n,u^{i}_{n+1})$. Then, upon substitution in 
    equation~\eqref{eq:cond1diaga} we obtain the same equation as 
    in~\eqref{eq:pdncond1onedim}. This readily implies that $G^{i}=a^{i}(u_n)a^{i}(u_{n+1})$. Inserting these values into the conditions~\eqref{eq:pdncond2}
    and~\eqref{eq:pdncond3} we obtain (no summation on repeated indices):
    \begin{subequations}
        \begin{align}
            a^\ell(u_n){\pdv{h_n^{ij}}{u_n^\ell}}(\vb*{u}_n)
            -{\dv{a^\ell}{u_n^\ell}}(u_n^\ell) h^{ij}_n(\vb*{u}_n) &=0,  & \ell&= i,j,
        \label{eq:cond2diag_ab}
        \\
        a^k(u_n){\pdv{h_n^{ij}}{u_n^k}}(\vb*{u}_n) &=0, &  k&\neq i,j.
        \label{eq:cond2diagc}
        \end{align}
    \end{subequations}
    The condition~\eqref{eq:cond2diagc} implies that $h^{ij}=h^{ij}(u_n^i,u_n^j)$,
    which in turn when inserted in~\eqref{eq:cond2diag_ab} implies that 
    $h^{ij}=C^{ij}a^{i}(u_n^i)a^{j}(u_n^j)$. Hence, we obtained
    the discrete differential-geometric Hamiltonian operator:
    \begin{equation}
        \mathcal{K}^{ij} =
        \left[a^{i}(u_n)a^{i}(u_{n+1}) \Trasl -a^{i}(u_n)a^{i}(u_{n+1}) \Trasl^{-1}\right]\delta^{ij} 
        + C^{ij}a^{i}(u_n)a^{j}(u_n).
    \end{equation}
    However, we can observe that
    this operator gives nothing new with respect to the constant one considered in~\Cref{ex:constg}.
    Indeed, reasoning in the same way as in~\Cref{rem:gauge} we can perform the gauge transformation:
    \begin{equation}
        u_n^i \longrightarrow \widetilde{u}_n^i = \int^{u_n^i} \frac{\dd \upsilon}{a^i(\upsilon)},
    \end{equation}
    which brings the operator into the constant form:
    \begin{equation}
        \mathcal{K}^{ij}_{c} =
        \left[ \Trasl - \Trasl^{-1}\right]\delta^{ij} 
        + C^{ij}.
    \end{equation}
\end{example}

The last example shows why solving the system of conditions~\eqref{eq:pdncond} can
be misleading: what seems to be a genuine solution is, in fact, a known solution in disguise.
For this reason, it is clear that a better, coordinate independent, characterisation of the
dDGP bracket is needed. This is the content of the next Section.

\section{Geometric interpretation of dDGP brackets}
\label{sec:geom}
In this Section, we discuss the geometric interpretation of the brackets 
$\{u^i_n,u^j_{n+1}\}_{(1)}$ and $\{u^i_n,u^j_{n}\}_{(1)}$ 
in~\eqref{eq:discretePBexp} following~\cite{DubrovinNovikov1989,ParodiThesis}, via Poisson--Lie groups\footnote{Poisson--Lie groups were originally called Hamiltonian Lie groups (or Hamilton–Lie groups) in references such as~\cite{Drinfeld1983} and~\cite{DubrovinNovikov1989}, and the two nomenclatures refer to the same structure.} and Lie bialgebras~\cite{Kosmann2004}. We will see that the two brackets are determined by the pair $(q,k)$, where $q:\mathfrak{g}^* \to \mathfrak{g}$ is a Lie algebra homomorphism, and~$k \in \mathfrak{g} \wedge \mathfrak{g}$ defines a coboundary deformation of the dual Lie algebra. The general admissible case in~\Cref{thm:Dubrovin_general} corresponds to a self-dual Lie bialgebra, where $q$ is an isomorphism. The~\Cref{thm:Dubrovin_quasiF} is obtained when $k$ coincides with a non-degenerate classical $r$-matrix, and therefore the pair $(q,k)$ reduces to the triangular case~($(q,k)=(r,-r)$), or quasi-Frobenius.  

\subsection{Preliminary notions} \label{sec:PoissonLie_preliminary}
We start by collecting the necessary background from the theory of Poisson--Lie groups and Lie bialgebras. We then introduce the classical and generalised classical Yang--Baxter equations, which govern the deformation of the bialgebra structure. Building on these notions, we introduce the two objects that will play a central role in the geometric interpretation: a Lie algebra homomorphism $q$, whose transpose is also a homomorphism, and a skew-symmetric cocycle deformation $k$. 

\begin{definition}
A Lie group $G$ is called a Poisson--Lie group
if it is equipped with a Poisson bracket such that the group 
multiplication
\begin{equation}
m : G \times G \to G, \qquad (X,Y) \mapsto XY,    
\end{equation}
is a Poisson map, i.e.\ it preserves the product Poisson structure 
on $G \times G$
\begin{equation}
    m(\pb{(x,y)}{(x',y')}_{G\times G}) = \pb{m(x,y)}{m(x',y')}_G.
\end{equation}
\end{definition}

Let $\mathfrak{g}$ be the Lie algebra of $G$, with basis $\{e_\alpha\}$ and structure constants $c^{\sigma}_{\alpha\beta}$ 
defined by
\begin{equation}\label{eq:Lie_algebra_c}
[e_\alpha, e_\beta] = c^{\sigma}_{\alpha\beta} e_\sigma,     
\end{equation}
A Poisson--Lie structure on $G$ is completely determined by a Lie bracket 
on the dual vector space $\mathfrak{g}^*$, with basis $\{e^{\alpha}\}$ (satisfying $\langle e^\alpha, e_\beta \rangle = \delta^\alpha_\beta$) and structure constants $f^{\alpha \beta}_{\sigma}$ defined by:
\begin{equation}\label{eq:Lie_algebra_dual_f}
[e^\alpha, e^\beta]_{*} = f^{\alpha\beta}_{\sigma} e^\sigma,
\end{equation}
such that the pair $(\mathfrak{g}\,;\mathfrak{g}^*)$ forms a Lie bialgebra.

\begin{definition}
    A pair $(\mathfrak{g}\,,c^{\sigma}_{\alpha\beta};\mathfrak{g}^*,f^{\alpha\beta}_\sigma)$ is a Lie bialgebra if the cobracket\footnote{The action of the cobracket on an element of the base of $\mathfrak{g}$ is $\delta(e_{\sigma}) = f^{\alpha \beta}_{\sigma} e_{\alpha} \otimes e_{\beta}$. } $\delta : \mathfrak{g} \to \mathfrak{g} \otimes \mathfrak{g}$, 
    defined as the transpose of the bracket on $\mathfrak{g}^*$, satisfies the 1-cocycle condition: 
    \begin{equation}
    \begin{split}
        \delta([x,y]) &= \text{ad}_{x}^{(2)} \delta(y) - \text{ad}_{y}^{(2)} \delta(x), \qquad  x,y \in \mathfrak{g}, 
        \end{split}
    \end{equation}
    where $\text{ad}_{x}^{(2)}(a \otimes b) = [x,a] \otimes b + a \otimes [x,b]$. In terms of the structure constants $c^{\sigma}_{\alpha\beta}$ and $f^{\alpha\beta}_{\sigma}$, 
the 1-cocycle condition reads as
\begin{equation}
    c^{\sigma}_{\mu\nu} f^{\alpha \beta}_{\sigma} = c^{\alpha}_{\sigma\nu} f^{\sigma \beta}_{\mu} + c^{\beta}_{\sigma\nu} f^{\alpha \sigma}_{\mu} - c^{\alpha}_{\mu\sigma} f^{\sigma \beta}_{\nu} - c^{\beta}_{\mu\sigma} f^{\alpha \sigma}_{\nu} \,,
\end{equation}
and the Lie algebras $\mathfrak{g}$ and $\mathfrak{g}^*$ are said to be \emph{compatible}. 
\end{definition}

A significant class of Lie bialgebras for our purposes consists of those equipped with an exact cobracket, i.e.\ there exists an element $r \in \mathfrak{g} \otimes \mathfrak{g}$ such that 
\begin{equation}
   \delta (x)= \delta_r(x) \;=\; [x \otimes 1 + 1 \otimes x, \, r],
\qquad x \in \mathfrak{g}\,.
\end{equation}
The Lie algebra is said coboundary Lie bialgebra, and $r$ is called a classical $r$-matrix. For the bracket given by the coboundary on $\mathfrak{g}^*$ to define a Lie algebra, the Schouten bracket $\SB{r,r}\in\mathfrak{g} \otimes \mathfrak{g} \otimes\mathfrak{g}$ needs to be ad-invariant. In operatorial form this is:
\begin{equation}
    [x \otimes 1 \otimes 1 + 1 \otimes x \otimes 1 + 1 \otimes 1 \otimes x, \, \SB{r, r} ] = 0 \,, \qquad x \in\mathfrak{g}\,,
\end{equation}
taking the name of generalised classical Yang--Baxter equation (GYBE). 
A stronger constraint implying GYBE is the classical Yang--Baxter equation (CYBE), for which $\SB{r,r}=0$. In coordinates ($\SB{r,r}:=[ r,r]^{\alpha \beta \sigma}$) these conditions reads as
\begin{align}
\label{eq:CYBE}
   \text{(CYBE)}\colon \quad &[ r,r]^{\alpha \beta \sigma} = 0 \,, \\
   \label{eq:GYBE}
   \text{(GYBE)}\colon \quad &[ r,r]^{\alpha \beta \sigma} = f^{\alpha \beta}_{\rho} r^{\rho \sigma}+f^{\beta \sigma}_{\rho} r^{\rho \alpha}+ f^{\sigma \alpha}_{\rho} r^{\rho \beta}
\end{align}
where the left hand side is explicitly given by 
\begin{equation}
    [ r,r]^{\alpha \beta \sigma}:= c^{\alpha}_{\rho\nu}\,r^{ \rho \beta }r^{ \nu \sigma}+c^{\beta}_{\rho\nu}\,r^{ \alpha \rho}r^{\nu \sigma}+c^{\sigma}_{\rho\nu}\,r^{\alpha \rho }r^{\beta \nu}\,. 
\end{equation}
In~\eqref{eq:CYBE} $r$ is skew-symmetric and takes the name of \emph{triangular} $r$-matrix\footnote{The CYBE is often written for $r \in \mathfrak{g} \otimes \mathfrak{g}$ as 
\begin{equation}
    [r_{12},r_{23}]+[r_{12},r_{13}]+[r_{13},r_{23}] = 0 \,,
\end{equation}
with $r_{12}=r^{ij} (e_{i}\otimes e_j \otimes 1)$, $r_{13}=r^{ij} (e_{i}\otimes 1 \otimes e_j)$, $r_{23}=r^{ij} (1\otimes e_{i} \otimes e_j)$ elements in $\mathfrak{g} \otimes \mathfrak{g} \otimes \mathfrak{g}$. In the basis $\{e_a \otimes e_b \otimes e_c\}$ for $\mathfrak{g} \otimes \mathfrak{g} \otimes \mathfrak{g}$ this reduces to~\eqref{eq:CYBE}. }, whereas in~\eqref{eq:GYBE} $r$ has no defined parity.

Following~\cite{Dubrovin1989,DubrovinNovikov1989}, we consider the Poisson--Lie group $G$, its bialgebra $(\mathfrak{g}\,,c^{\sigma}_{\alpha\beta};\mathfrak{g}^*,f^{\alpha\beta}_\sigma)$, and additionally a skew-symmetric $k \in \mathfrak{g} \wedge \mathfrak{g}$ ($k^{\alpha\beta}=-k^{\beta \alpha}$). The latter is introduced such that the deformed cohomologous 1-cocycle
\begin{equation}
    \underline{f}^{\alpha \beta}_{\sigma} = f^{\alpha \beta}_{\sigma} + d^{\alpha \beta}_{\sigma} \equiv f^{\alpha \beta}_{\sigma} + c^{\alpha}_{\rho \sigma}\,k^{\rho \beta} + c^{\beta}_{\rho \sigma}\,k^{\alpha \rho}
\end{equation}
also defines a Lie algebra structure on the dual $(\mathfrak{g}^*,\underline{f}^{\alpha \beta}_{\sigma})$. The condition for $k$ to be such that the Jacobi identity is satisfied for $(\mathfrak{g}^*,\underline{f}^{\alpha \beta}_{\sigma})$ coincides with the GYBE~\eqref{eq:GYBE}. If $k$ is such that the deformation $d^{\alpha \beta}_\sigma$ coincides with $f^{\alpha \beta}_\sigma$, $k$ is a classical $r$-matrix and GYBE reduces to CYBE~\eqref{eq:CYBE} ($k = -r$). 

We introduce the Lie homomorphism $q \colon (\mathfrak{g}^*,f^{\alpha\beta}_\sigma) \to (\mathfrak{g}\,,c^{\sigma}_{\alpha\beta})$ such that its transpose $^t q\colon (\mathfrak{g}^*,\underline{f}^{\alpha\beta}_\sigma) \to (\mathfrak{g}\,,c^{\sigma}_{\alpha\beta})$ is as well a Lie algebra homomorphism ($q=q^{\alpha\beta}$,$^t q=q^{\beta\alpha}$). In the case where $k$ is a classical $r$-matrix, then $\underline{f}^{\alpha \beta}_\sigma = 2f^{\alpha \beta}_\sigma$ and $q$ becomes a classical $r$-matrix as well ($q = r$). Both the cases where the pair $(q,k)$ is general and where it reduces to $(r,-r)$ are depicted in diagrammatic form in~\Cref{fig:placeholder}.

\begin{figure}
    \centering
    \begin{tikzpicture} \def\len{1.5cm}
    \begin{scope} 
        \node (A) at (0,0) {$(\mathfrak{g}^*,f^{\alpha\beta}_\sigma)$};
        \node (B) at (2*\len,0) {$(\mathfrak{g}^*,\underline{f}^{\alpha\beta}_\sigma)$};
        \node (C) at (\len,-1.2*\len) {$(\mathfrak{g}\,,c_{\alpha\beta}^\sigma)$};
        \draw[->] (A) -- (C) node[left=.2,midway] {$q$};
        \draw[->] (B) -- (C) node[right=.2,midway] {$^t q$};
        \draw[->] (A) -- (B) node[above=.2,midway] {$k$};
        \node[below=.4*\len] (D) at (C) {general case};
         \node[below=.7*\len] (D) at (C) {(self-dual)};
    \end{scope} 
    \begin{scope}
    \node (A) at (4*\len,0) 
        {$(\mathfrak{g}^*,f^{\alpha\beta}_\sigma)$};
        \node (B) at (6*\len,0) {$(\mathfrak{g}^*,\underline{f}^{\alpha\beta}_\sigma)$};
        \node (C) at (5*\len,-1.2*\len) {$(\mathfrak{g}\,,c_{\alpha\beta}^\sigma)$};
        \draw[->] (A) -- (C) node[left=.2,midway] {$r$};
        \draw[->] (C) -- (B) node[right=.2,midway] {$r$};
        \draw[->] (B) -- (A) node[above=.2,midway] {$r$};
        \node[below=.4*\len] (D) at (C) {quasi-Frobenius};
        \node[below=.7*\len] (D) at (C) {(symplectic)};
    \end{scope}
    \end{tikzpicture}
    \caption{Commutative diagrams describing the general case
    and the quasi-Frobenius case.}
    \label{fig:placeholder}
\end{figure}
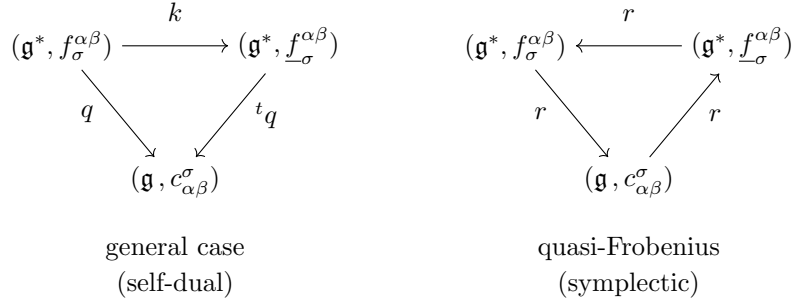

Moreover, we recall that a Lie algebra is called quasi-Frobenius 
(resp.\ Frobenius) if it is equipped with a non-degenerate
closed (resp.\ exact) antisymmetric bilinear form 
$\omega \in \mathfrak{g}\wedge\mathfrak{g}$, see e.g.\ \cite{Millionschikov2004}. 
As usual, the form $\omega$ is called a \emph{symplectic form}. In fact,
this condition is equivalent to the existence of a non-degenerate
classical $r$-matrix:

\begin{lemma}
   If $r \in \mathfrak{g} \wedge \mathfrak{g}$ is a non-degenerate skew-symmetric solution of CYBE~\eqref{eq:CYBE}, then its inverse $r^{-1} \in \mathfrak{g}^* \wedge \mathfrak{g}^*$ is a symplectic form on $\mathfrak{g}$. Conversely, the
   inverse of a symplectic form $\omega\in \mathfrak{g}\wedge\mathfrak{g}$ 
   yields a non-degenerate skew-symmetric solution to the CYBE~\eqref{eq:CYBE}. 
\end{lemma}

\begin{remark}
    As it is clear from the fact that the classical $r$-matrix
    (resp.\ the symplectic form $\omega$)
    is an element of $\mathfrak{g}^{*}\wedge\mathfrak{g}^{*}$ 
    (resp.\ of $\mathfrak{g}\wedge\mathfrak{g}$) we have
    that quasi-Frobenius Lie algebras have necessarily 
    even dimension.
\end{remark}

We denote by $\{L_\alpha\}$ and $\{R_\alpha\}$ the left  and right  invariant vector fields on $G$, respectively. 
They form global bases of the tangent bundle $TG$ and are related by the adjoint action of the group:
\begin{equation}
\label{eq:L_R_relation}
L_\alpha(g) = \operatorname{Ad}_{g} R_\alpha(g), \qquad g \in G.
\end{equation}
Equivalently, in terms of the left invariant basis, the right  invariant fields are obtained by conjugating the left invariant ones with the inverse group element $g^{-1}$.

We introduce $p_0$, the canonical (multiplicative) Poisson--Lie bivector on $G$ associated with the Lie bialgebra~$(\mathfrak{g}\,,c^{\sigma}_{\alpha\beta};\mathfrak{g}^*,f^{\alpha\beta}_\sigma)$. 
In the frame $\{L_{\alpha}\}$ it is written as 
\begin{equation}\label{eq:poisson_bivector}
p_0(g) = p_0^{\alpha \beta}\,L_{\alpha}(g) \wedge L_{\beta}(g)\,, \qquad p_0(e) = 0,
\end{equation}
with the initial condition on the identity of the group $e \in G$, and $p_0^{\alpha \beta} \in \mathcal{C}^{\infty}(G)$. By construction, the linearisation of $p_0$ at the identity coincides with the Lie--Poisson bracket on the dual $\mathfrak{g}^*$ with structure constants $f^{\alpha \beta}_{\sigma}$.

\subsection{Statement and proof of the results}
 We now specialise the construction of~\Cref{sec:PoissonLie_preliminary} to the case in which the manifold $\mathcal M^N_n$ is the Poisson--Lie group $G$, i.e.\ we identify $G = \mathcal M^N_n$. We use the left and right invariant vector fields on~$\mathcal M^N_n$ to translate its structural data into a bracket on the coordinate functions $u^i_n$ of the discrete loop space~$\mathcal M^N_\infty$. Two kinds of brackets will appear: the \emph{on-site} brackets $\{u_n^i, u_n^j\}$, between coordinates at the same lattice site, and the \emph{inter-site} brackets $\{u_n^i, u_{n+1}^j\}$, between coordinates at neighbouring sites. These are the only two families of elementary brackets that are non-zero, that in~\eqref{eq:discretePBexp} are denoted as $h_n^{ij}$ and $g_n^{ij}$ respectively. Imposing the Hamiltonianity conditions~\eqref{eq:pdncond} on~$h_n^{ij}$ and $g_n^{ij}$ is what determines the geometric form of the bracket $\{u^i_n,u^j_m\}$: the on-site bracket is forced to be governed by the bivector~$p_0$ evaluated at $\vb*{u}_n$, together with the skew-symmetric element $k$, while the inter-site bracket by the homomorphism~$q$.

The first result gives the general form of such a bracket. It was stated without a 
proof by Dubrovin in~\cite{Dubrovin1989}, and appears again without proofs in the 
review by Dubrovin and Novikov \cite[\S3]{DubrovinNovikov1989}. A sketched proof was 
later given by Parodi in his Ph.D.\ thesis \cite[Thm.~2.4.4]{ParodiThesis}. We 
reproduce it here in corrected form and supplement it with~\Cref{lem:factorisation}, 
which is not explicit there.

\begin{theorem}[General case]\label{thm:Dubrovin_general}
Let $q : (\mathfrak g^*, \underline f^{\alpha\beta}_\sigma) \to (\mathfrak g, c^\sigma_{\alpha\beta})$ be a Lie algebra homomorphism whose transpose ${}^tq$ is also a homomorphism, and let $k \in \mathfrak g \wedge \mathfrak g$ be a skew-symmetric solution of the GYBE~\eqref{eq:GYBE}. Then the brackets
    \begin{subequations}\label{eq:general_brackets}
    \begin{align}
     \label{eq:gen_onsite}
        h^{ij}_n&=\{ u^i_n , u^j_{n} \} = p_0^{\alpha \beta} L^i_\alpha(\vb*{u}_n)\, L^j_\beta(\vb*{u}_{n}) + k^{\alpha \beta} R^i_\alpha(\vb*{u}_n)\, R^j_\beta(\vb*{u}_{n}) \,, \\[1mm]
    \label{eq:gen_intersite}
        g^{ij}_n&=\{ u^i_n , u^j_{n+1} \} = q^{\alpha \beta} L^i_\alpha(\vb*{u}_n)\, R^j_\beta(\vb*{u}_{n+1})\,, 
    \end{align}
    \end{subequations}
    where $p_0^{\alpha\beta} \in \mathcal C^\infty(\mathcal M^N_n)$ are the coefficients of the Poisson--Lie bivector~\eqref{eq:poisson_bivector}, define a dDGP bracket and vice versa.
\end{theorem}

\begin{remark}\label{rem:non_PL}
Each site $n$ carries a copy of the Poisson--Lie group $\mathcal M^N_n$, with on-site bracket given by the Poisson--Lie bivector $p_0$ deformed by $k$. Note that $q$ being non-degenerate, hence is necessarily non-zero the inter-site bracket is always present and the structure is never ultra-local. Specifically, the whole bracket $\{u^i_n,u^j_m\}$ on $\mathcal M^N_\infty$ is a Poisson structure built from the (deformed) Poisson--Lie structures on the individual sites together with an inter-site coupling, and hence it is not a Poisson--Lie structure on $\mathcal M^N_\infty$ itself. 
\end{remark}

\Cref{thm:Dubrovin_general} gives the geometric form of both brackets: the inter-site bracket is determined by the condition~\eqref{eq:pdncond1} alone, while the on-site bracket and the compatibility of the two pieces are determined by the remaining conditions~\eqref{eq:pdncond2}, \eqref{eq:pdncond3} and~\eqref{eq:pdncond4}. 

In the following,~\Cref{lem:factorisation} and~\Cref{lem:factor_vectorfields} are built on condition~\eqref{eq:pdncond1}. 

\begin{lemma}\label{lem:factorisation}
    If the matrix $g^{ij}_n$ is non-degenerate, then  the 
    condition~\eqref{eq:pdncond1} implies the following factorisation:
    \begin{equation}\label{eq:gfact}
        g^{ij}_n = q^{\alpha\beta}\,L^i_\alpha(\vb*{u}_{n})\, R^{j}_\beta(\vb*{u}_{n+1}),
    \end{equation}
    where $L$, $R$, and $q$ are all invertible matrices, and $q$ is constant.  
\end{lemma}

\begin{proof} The non-degeneracy condition on $g_n$ allows one to rewrite the condition~\eqref{eq:pdncond1} as 
\begin{equation}\label{eq:cond1_non_deg}
    (g_n)_{\ell i}\,\pdv{g_{n}^{ij}}{u_{n+1}^m}
            =
            \pdv{g_{n+1}^{jk}}{u_{n+1}^\ell} \,(g_{n+1})_{km} \,,  
\end{equation}
where the lower indices identifies the elements of $(g_n)^{-1}$. The triples $(j,\ell,m)$ identify the free indices, yielding $N^3$ conditions. These can be recast in terms of the equivalence between (elements of) rows of the $2N$ $N$-dimensional matrices defined as 
\begin{equation} \label{eq:auxiliary}
    A^{(j)} = g(\vb*{u},\vb*{v}) ^{-1}\,\frac{\partial g}{\partial v^j}(\vb*{u}, \vb*{v}),\qquad B^{(k)} = g(\vb*{v}, \vb*{w})^{-1}\,\frac{\partial g}{\partial v^k}(\vb*{v}, \vb*{w}), \qquad  j,k \in \{1, \dots, N\}\,,
\end{equation} 
where $(\vb*{u},\vb*{v},\vb*{w}) \equiv (\vb*{u}_n,\vb*{u}_{n+1},\vb*{u}_{n+2},)$. With this notation, the conditions~\eqref{eq:cond1_non_deg} are equivalent to 
\begin{equation}
    \text{row}_k(A^{(j)}) = \text{row}_j(B^{(k)}) \,, 
\end{equation}
such that, given to the dependencies in~\eqref{eq:auxiliary}, all the matrices can depend on $\vb*{v}$ only, i.e.\ for any $j$ and any $k$:
\begin{equation}
    A^{(j)} = A^{(j)}(\vb*{v})\,, \qquad B^{(k)} = B^{(k)}(\vb*{v})\,. 
\end{equation}
With focus on the first of~\eqref{eq:auxiliary}, we can write the system of ODEs 
\begin{equation}\label{eq:ODEs_g}
    \frac{\partial g}{\partial v^j}(\vb*{u},\vb*{v}) = g(\vb*{u},\vb*{v})\,A^{(j)}(\vb*{v})\,, \qquad j \in \{1,\dots, N\}\,,
\end{equation} 
and then ask whether the following Cauchy problem in the invertible matrix $f(\vb*{v})$ is well posed 
\begin{equation}\label{eq:ODEs_f}
\begin{aligned}
    \frac{\partial f}{\partial v^j}(\vb*{v}) &= f(\vb*{v})\,A^{(j)}(\vb*{v})\,, \qquad j \in \{1,\dots, N\},  
\end{aligned} 
\end{equation}
with initial condition $f(\vb*{0})= I_N$, and $I_N$ the $N$-dimensional identity matrix. The compatibility condition for the system of ODEs implies the zero-curvature condition for the matrices $A^{(j)}$, $A^{(k)}$. In particular, we have 
\begin{equation*}
    \begin{split}
        \frac{\partial^2 f}{\partial v^k \partial v^j} (\vb*{v}) & = \partial_k f(\vb*{v})\,A^{(j)}(\vb*{v})+ f(\vb*{v})\,\partial_k A^{(j)}(\vb*{v}) = f(\vb*{v})\bigl( A^{(k)}(\vb*{v}) A^{(j)}(\vb*{v}) + \partial_k A^{(j)}(\vb*{v}) \bigr) \,, \\
    \end{split}
\end{equation*}
and imposing the equivalence of mixed derivatives $\partial_k \partial_j f=\partial_j \partial_k f $ we obtain
\begin{equation}
    \frac{\partial A^{(j)}}{\partial v^k} (\vb*{v}) - \frac{\partial A^{(k)}}{\partial v^j}(\vb*{v}) + \bigl[ A^{(k)}(\vb*{v})\,, A^{(j)}(\vb*{v})\bigr]=0\,, 
\end{equation}
which is verified because of the system of ODEs in~\eqref{eq:ODEs_g}. Therefore, a solution $f(\vb*{v})$ exists. Next, we compare $f(\vb*{v})$ and $g(\vb*{u},\vb*{v})$ by introducing the additional function $\ell(\vb*{u},\vb*{v})$ defined as 
\begin{equation}
    \ell(\vb*{u},\vb*{v}) = g(\vb*{u},\vb*{v})\, f(\vb*{v})^{-1}\,, 
\end{equation} 
and evaluating the $N$ derivatives with respect to the components of $\vb*{v}$, i.e. 
\begin{equation}
    \begin{aligned}
         \frac{\partial \ell}{\partial v^j}(\vb*{u},\vb*{v}) &= \frac{\partial \ell}{\partial v^j} g(\vb*{u},\vb*{v})\, f(\vb*{v})^{-1} + g(\vb*{u},\vb*{v})\,\frac{\partial \ell}{\partial v^j} f(\vb*{v})^{-1} \,.
    \end{aligned}
\end{equation} 
Since $\partial_j (f f^{-1}) = 0$, we have $\partial_j f^{-1} = -f^{-1} (\partial_j f)\,f^{-1}$. Using \eqref{eq:ODEs_g} and \eqref{eq:ODEs_f} we obtain 
\begin{equation*}
    \begin{split} 
        \frac{\partial \ell}{\partial v^j}(\vb*{u},\vb*{v}) &= g(\vb*{u},\vb*{v}) \bigl( A^{(j)}(\vb*{v})\,f(\vb*{v})^{-1} - \,f(\vb*{v})^{-1}\frac{\partial f}{\partial v^j}(\vb*{v})\, f(\vb*{v})^{-1}  \bigr) \\[1mm] 
        &= g(\vb*{u},\vb*{v}) \bigl( A^{(j)}(\vb*{v})\,f(\vb*{v})^{-1} - A^{(j)}(\vb*{v})\, f(\vb*{v})^{-1} \bigr) = 0, ~~ \implies ~~  \ell(\vb*{u},\vb*{v})=\ell(\vb*{u}).
    \end{split}
\end{equation*}
The function $g(\vb*{u},\vb*{v})$ is then factorised as $g(\vb*{u},\vb*{v}) = \ell(\vb*{u})\,f(\vb*{v})$, and can taken into the form~\eqref{eq:gfact}.   
\end{proof}

\begin{remark}\label{rem:isomorphism}
The non-degeneracy of the matrix $g_n^{ij}$ forces $q$ to be an isomorphism, hence the Lie algebra $\mathfrak g$ to be self-dual  (see \cite[\S3, Rem.~1]{DubrovinNovikov1989} and \cite[\S2.4]{ParodiThesis}). 
\end{remark}

\begin{remark}
In the general case of \Cref{thm:Dubrovin_general}, the on-site bracket~\eqref{eq:gen_onsite} vanishes identically if and only if the dual Lie algebra $\mathfrak{g^*}$ is abelian ($f=0$ implies $p_0=0$), and the skew-symmetric $k$ is zero as well. 
In the non-degenerate case, where the inter-site homomorphism $q: \mathfrak{g}^* \to \mathfrak{g}$ in~\eqref{eq:gen_intersite} is a Lie algebra isomorphism, 
this forces the original Lie algebra $\mathfrak{g}$ to be abelian as well. Therefore, the only self-dual Lie bialgebra with a vanishing  on-site bracket is the trivial one, where both $\mathfrak{g}$ and $\mathfrak{g}^*$ are abelian.
\end{remark}

\begin{lemma}\label{lem:factor_vectorfields}
    Using the factorisation~\eqref{eq:gfact} the 
    condition~\eqref{eq:pdncond1} can be rewritten as:
    \begin{equation}
        \comm{L_\alpha}{R_\beta}^{k}=0,
        \label{eq:pdncond1b}
    \end{equation}
    hence, the matrices $L$ and $R$ are made of commuting vector fields.
\end{lemma}
\begin{proof}
We consider the factorisation~\eqref{eq:gfact},
with $q$ constant and invertible, into condition~\eqref{eq:pdncond1}. Writing $\partial_s := \partial/\partial u_{n+1}^s$, we have
\begin{equation}
\frac{\partial g_n^{ij}}{\partial u_{n+1}^s} = q^{\alpha\beta} L_\alpha^i(\vb*{u}_n)\, \partial_s R_\beta^j(\vb*{u}_{n+1}),
\qquad
\frac{\partial g_{n+1}^{jk}}{\partial u_{n+1}^s} = q^{\alpha\beta}\, \partial_s L_\alpha^j(\vb*{u}_{n+1})\, R_\beta^k(\vb*{u}_{n+2}).
\end{equation}
The left hand side and right hand side of~\eqref{eq:pdncond1} then become
\begin{subequations}
\begin{align}
\textup{LHS}&:~ q^{\alpha\beta} q^{\gamma\delta}\, L_\alpha^i(\vb*{u}_n)\, \partial_s R_\beta^j(\vb*{u}_{n+1})\, L_\gamma^s(\vb*{u}_{n+1})\, R_\delta^k(\vb*{u}_{n+2}) \\
\textup{RHS}&:~ 
q^{\alpha\beta} q^{\gamma\delta}\, L_\alpha^i(\vb*{u}_n)\, R_\beta^s(\vb*{u}_{n+1})\, \partial_s L_\gamma^j(\vb*{u}_{n+1})\, R_\delta^k(\vb*{u}_{n+2}).
\end{align} 
\end{subequations}
Since the factors $L_\alpha^i(\vb*{u}_n)$ and $R_\delta^k(\vb*{u}_{n+2})$ are independent and $q$ is invertible, condition~\eqref{eq:pdncond1} is equivalent to
\begin{equation}
L_\gamma^s(\vb*{u}_{n+1})\, \partial_s R_\beta^j(\vb*{u}_{n+1}) - R_\beta^s(\vb*{u}_{n+1})\, \partial_s L_\gamma^j(\vb*{u}_{n+1}) = 0, \qquad \forall~ \beta, \gamma, j,
\end{equation}
that is, $\comm{L_\gamma}{R_\beta}^j = 0$.
\end{proof}

\

\Cref{lem:conds_b_c} encodes the geometric interpretation of the brackets, adapting the notation used in \Cref{sec:PoissonLie_preliminary}. 
The on-site bracket can be written both in the left invariant frame $\{L_\alpha\}$ and in the right  invariant frame $\{R_\alpha\}$, with the two frames related by the adjoint $R_\alpha = \textup{Ad}_g L_\alpha$. In the proof below we will use both representations, but in the end we will refer to the $L$-frame, in order to match the way the bracket is stated in \Cref{thm:Dubrovin_general}, where the term involving $p_0$ is already expressed in the left invariant frame.

\begin{lemma}\label{lem:conds_b_c}
Assume the factorisations
\begin{equation}\label{eq:fact}
\begin{split} 
g_n^{ij} &= q^{\alpha\beta}\, L^i_\alpha(\vb*{u}_n)\, R^j_\beta(\vb*{u}_{n+1}),\\[1mm]
h_n^{ij} &= k^{\alpha\beta}(\vb*{u}_n)\,L^i_\alpha(\vb*{u}_n)\,  L^j_\beta(\vb*{u}_n) = \underline{k}^{\alpha\beta}(\vb*{u}_n)\,R^i_\alpha(\vb*{u}_n)\,  R^j_\beta(\vb*{u}_n),
\end{split} 
\end{equation}
with $q$ constant and invertible, $k,\underline{k}$ skew-symmetric. Then conditions~\eqref{eq:pdncond2} and~\eqref{eq:pdncond3} imply the following.
\begin{enumerate}
\item[\textup{($i$)}] From \eqref{eq:pdncond2}: the left invariant vector fields close,
\begin{equation}
[L_\alpha, L_\beta]^k = c^\sigma_{\alpha\beta}\, L^k_\sigma,
\end{equation}
with $c^\sigma_{\alpha\beta}$ constant and skew-symmetric in $(\alpha,\beta)$. Moreover, there exists a constant tensor $f^{\alpha\beta}_\lambda$, skew-symmetric in $(\alpha,\beta)$, such that
\begin{equation}
q^{\mu\alpha} q^{\nu\beta}\, c^\sigma_{\mu\nu} = f^{\alpha\beta}_\lambda\, q^{\sigma\lambda},
\end{equation}
i.e.\ $q$ is a Lie algebra homomorphism from $(\mathfrak g^*, f)$ to $(\mathfrak g, c)$. The coefficient $k^{\alpha\beta}(\vb*{u}_n)$ satisfies the system
\begin{equation}\label{eq:system_k_statement}
\mathcal L_\gamma k^{\alpha\beta}
= c^\alpha_{\gamma\varepsilon}\, k^{\varepsilon\beta}
+ c^\beta_{\gamma\varepsilon}\, k^{\alpha\varepsilon}
+ f^{\alpha\beta}_\gamma,
\end{equation}
where $\mathcal L_\gamma := L^s_\gamma \partial_s$ is the Lie derivative along the left invariant vector field. 
\item[\textup{($ii$)}] From \eqref{eq:pdncond3}: the right  invariant vector fields close,
\begin{equation}
[R_\alpha, R_\beta]^k = -c^\sigma_{\alpha\beta}\, R^k_\sigma,
\end{equation}
with the same constants $c^\sigma_{\alpha\beta}$. Moreover, there exists a constant tensor $\underline f^{\alpha\beta}_\lambda$, skew-symmetric in~$(\alpha,\beta)$, such that
\begin{equation}
q^{\alpha\mu} q^{\beta\nu}\, c^\sigma_{\mu\nu} = -\underline f^{\alpha\beta}_\lambda\, q^{\lambda\sigma},
\end{equation}
i.e.\ ${}^tq$ is a Lie algebra homomorphism from $(\mathfrak g^*, \underline f)$ to $(\mathfrak g, c)$. The coefficient $\underline k^{\alpha\beta}(\vb*{u}_n)$ satisfies the system
\begin{equation}
\mathcal R_\gamma \underline k^{\alpha\beta}
= c^\alpha_{\gamma\varepsilon}\, \underline k^{\varepsilon\beta}
+ c^\beta_{\gamma\varepsilon}\, \underline k^{\alpha\varepsilon}
+ \underline f^{\alpha\beta}_\gamma, 
\end{equation}
where $\mathcal R_\gamma := R^s_{\gamma} \partial_s$ is the Lie derivative along the right invariant vector field. 
\item[\textup{($iii$)}] The two structures on $\mathfrak g^*$ are related by
\begin{equation}
\underline f^{\alpha\beta}_\rho =  f^{\alpha\beta}_\rho + \kappa^{\beta\sigma}\, c^\alpha_{\rho\sigma} + \kappa^{\sigma\alpha}\, c^\beta_{\rho\sigma},
\end{equation}
and the general solution of the system~\eqref{eq:system_k_statement} is
\begin{equation}
k^{\alpha\beta}(\vb*{u}_n) = p_0^{\alpha\beta}(\vb*{u}_n) + \bigl(\mathrm{Ad}^{(2)}_{\vb*{u}_n^{-1}} \kappa\bigr)^{\alpha\beta},
\end{equation}
with $\kappa^{\alpha\beta} \coloneqq k^{\alpha\beta}(\vb*{e})$ initial datum.
\end{enumerate}
\end{lemma}

\begin{proof}
The proof is a case-by-case analysis of the two mixed Jacobi conditions, followed by a comparison of the two structures on $\mathfrak g^*$ obtained from them.

\medskip
($i$)~
We substitute the factorisations~\eqref{eq:fact} into condition~\eqref{eq:pdncond2} evaluated on $(\vb*{u},\vb*{v})\equiv (\vb*{u}_n,\vb*{u}_{n+1})$. Introducing the notations $\partial_s := \partial/\partial u^s$ and $\bar{\partial_s} := \partial/\partial v^s$, after a suitable relabelling of the indices of some of the terms, condition~\eqref{eq:pdncond2} takes the form:
\begin{equation}\label{eq:b_factorised}
A^{ij\delta}(\vb*{u})\,R_\delta^\ell(\vb*{v}) = q^{\alpha\beta} q^{\gamma\delta}\, L^i_\alpha(\vb*{u}) L^j_\gamma(\vb*{u})\,
\bigl[R^s_\beta\, (\bar{\partial_s} R^\ell_\delta) - R^s_\delta\, (\bar{\partial_s} R^\ell_\beta)\bigr](\vb*{v}) ,
\end{equation}
where $A^{ij\delta}$ is the $\vb*{u}$-dependent tensor
\begin{equation}
\begin{split} 
A^{ij\delta}(\vb*{u}) &= \bigl[k^{\alpha\beta}(\partial_s L^i_\alpha)  L^j_\beta
+ (\partial_s k^{\alpha\beta})L^i_\alpha \, L^j_\beta
+ k^{\alpha\beta} L^i_\alpha  (\partial_s L^j_\beta)\bigr]\,q^{\gamma\delta}\, L^s_\gamma\,  \\
&~~+ \bigl[k^{\gamma\beta}\,(\partial_s L^j_\alpha)\,  L^i_\beta
- k^{\gamma\beta}(\partial_s L^i_\alpha)\, L^j_\beta \bigl]q^{\alpha\delta} \, L^s_\gamma .
\end{split} 
\end{equation}
The bracket in the RHS of~\eqref{eq:b_factorised} is, by definition,
\begin{equation}\label{eq:vector_R}
R^s_\beta\, (\bar{\partial_s} R^\ell_\delta) - R^s_\delta\, (\bar{\partial_s} R^\ell_\beta)
= -[R_\beta, R_\delta]^\ell,
\end{equation}
i.e.\ a vector in $\vb*{v}$. Hence, since $\{R_\alpha(\vb*{v})\}$ is a frame, the right hand side can be expressed as 
\begin{equation}\label{eq:frame_R}
    [R_\alpha, R_\beta]^\ell(\vb*{v}) = C^\delta_{\alpha \beta}(\vb*{v}) \,R_{\delta}^\ell(\vb*{v}),
\end{equation}
with $C^\delta_{\alpha \beta}(\vb*{v})$ coefficients. Of course, $[R_\alpha,R_\beta]$ is a right  invariant vector field itself, hence its coefficients in the right  invariant frame are constant $C^\delta_{\alpha \beta}(\vb*{v})=C^\delta_{\alpha \beta}$, skew-symmetric in $(\alpha,\beta)$ for the Lie bracket. The equation~\eqref{eq:b_factorised} then becomes
\begin{equation}
    A^{ij\delta}(\vb*{u})\,R_\delta^\ell(\vb*{v}) +q^{\alpha\beta} q^{\gamma\sigma}\, L^i_\alpha(\vb*{u}) L^j_\gamma(\vb*{u})\,
C^{\delta}_{\beta \sigma} R^\ell_{\delta}(\vb*{v})=0,
\end{equation}
from which we can equate the coefficients of $R^\ell_\delta$. Since $\{L_\alpha(\vb*{u})\}$ is a frame, the terms $\partial_s L$ in $A^{ij\delta}$ can be expressed as 
\begin{equation}\label{eq:vect_frame_L}
    \partial_s L^i_\alpha (\vb*{u}) =  B^\lambda_{s \alpha}(\vb*{u}) L^i_{\lambda}(\vb*{u})
\end{equation}
with $B^\lambda_{s \alpha}(\vb*{u})$ coefficients. Then, after renaming indinces in some of the terms, we obtain 
\begin{equation}
\begin{split} 
      &\bigl(k^{\rho\gamma}\,B^\alpha_{s \rho} {L^i_{\alpha}\,  L^j_\gamma}
+ (\partial_s k^{\alpha\gamma}){L^i_\alpha  \,L^j_{\gamma}}
+ k^{\alpha\beta} B^\gamma_{s \beta} {L^i_\alpha  \,L^j_{\gamma}}\bigr)\,q^{\sigma\delta}\, L^s_\sigma\,+ \\
&+ \bigl(k^{\beta\alpha} \, B^\gamma_{s \sigma}  \,  {L^i_{\alpha}\,  L^j_\gamma}
-  k^{\beta\gamma} \, B^\alpha_{s \sigma} \, {L^i_{\alpha}\,  L^j_\gamma} \bigr)\,q^{\sigma\delta}\,L^s_\beta + q^{\alpha\beta} q^{\gamma\sigma}\, {L^i_\alpha L^j_\gamma}\,
C^{\delta}_{\beta \sigma}=0,
\end{split} 
\end{equation}
and by using the non-degeneracy of $L^i_\alpha$ and $L^j_\gamma$ we get 
\begin{equation}
\begin{split} 
&q^{\sigma\delta}\bigl(k^{\rho\gamma}\,B^\alpha_{s \rho} 
+ (\partial_s k^{\alpha\gamma})
+ k^{\alpha\beta} B^\gamma_{s \beta} \bigr)L^s_\sigma + q^{\sigma\delta} \bigl(k^{\beta\alpha} \, B^\gamma_{s \sigma} 
-  k^{\beta\gamma} \, B^\alpha_{s \sigma}   \bigl)L^s_\beta =- q^{\alpha\beta} q^{\gamma\sigma}\,
C^{\delta}_{\beta \sigma}. 
\end{split} 
\end{equation}
The LHS of this expression is linear in $L^s_\sigma$, while the RHS is constant. Setting
\begin{equation}\label{eq:c_definition}
    c^\lambda_{\alpha\beta}(\vb*{u}) \coloneqq L^s_\alpha\, B^\lambda_{s\beta} - L^s_\beta\, B^\lambda_{s\alpha},
\end{equation}
and using the skew-symmetry of $k$, after renaming the indices we obtain
\begin{equation}\label{eq:renamed}
      q^{\sigma\delta}\, L^s_\sigma\, (\partial_s k^{\alpha\gamma})
    +q^{\sigma\delta}\, k^{\beta\gamma}\,c^{\alpha}_{\sigma \beta} + q^{\sigma\delta}\,k^{\alpha\beta} \, c^\gamma_{\sigma \beta} + q^{\alpha\beta} q^{\gamma\sigma}\, C^{\delta}_{\beta \sigma}=0. 
\end{equation}
Using~\eqref{eq:vect_frame_L} on the Lie bracket, we get
\begin{equation}
    \bigl[L_\alpha,L_\beta\bigr]^i = L^s_\alpha\partial_sL^i_\beta-L^s_\beta\partial_sL^i_\alpha=\left(L^s_\alpha B^\lambda_{s\beta}-L^s_\beta B^\lambda_{s\alpha}
\right)L^i_\lambda=c^\lambda_{\alpha\beta}(\vb*{u})L^i_\lambda,
\end{equation}
hence the coefficients $c^\lambda_{\alpha\beta}(\vb*{u})$ are the coefficients of the bracket in the left invariant frame $\{L_\alpha\}$. By the same argument used before $c^\lambda_{\alpha\beta}(\vb*{u})=c^\lambda_{\alpha\beta}$ constant and skew-symmetric in $(\alpha,\beta)$. Moreover, given the relation between left  and right  invariant vector fields, i.e.\ $R_\alpha = \textup{Ad}_{\vb*{u}}L_\alpha$, $C^\lambda_{\alpha \beta}=-c^\lambda_{\alpha \beta}$. 

In~\eqref{eq:renamed}, using the non-degeneracy of $q$ we get: 
\begin{equation}
    L^s_\rho\,(\partial_s k^{\alpha\gamma})=
q_{\delta\rho}q^{\alpha\beta}q^{\gamma\sigma}c^\delta_{\beta\sigma}-k^{\beta\gamma}c^\alpha_{\rho\beta}
-k^{\alpha\beta}c^\gamma_{\rho\beta}.
\end{equation}
Now we introduce 
\begin{equation}
    f^{\alpha \gamma}_\rho := q_{\delta\rho}q^{\alpha\beta}q^{\gamma\sigma}c^\delta_{\beta\sigma},
\end{equation}
i.e.\ $q$ is a homomorphism from $(\mathfrak g^*,f^{\alpha \beta}_\sigma)$ to $(\mathfrak g,c^\sigma_{\alpha \beta})$. From the fact that $k$ depends on the coordinate functions of the left invariant frame $\vb*{u}$, we set $\mathcal{L}_{\rho} = L^s_\rho \partial_s$
\begin{equation}\label{eq:k_system}
    \mathcal L_\rho k^{\alpha\gamma}=
k^{\gamma\beta}c^\alpha_{\rho\beta}
+k^{\beta\alpha}c^\gamma_{\rho\beta}+f^{\alpha \gamma}_\rho,
\end{equation}
where we used the skew-symmetry of $k$. 

\medskip

($ii$)~The proof is analogous to the previous case. 
We substitute the factorisations in~\eqref{eq:fact} into condition~\eqref{eq:pdncond3}, evaluated on $(\vb*{u},\vb*{v}) \equiv (\vb*{u}_{n-1},\vb*{u}_n)$. 
We introduce $\partial_s := \partial / \partial v^s$ and $\underline{\partial}_s := \partial / \partial u^s$. Differently from part $(i)$, in this case, we consider the expansion of $h^{ij}_n$ on the right  invariant vectors frame $\{R_\alpha\}$. After the substitution and a relabelling of the indices for some terms, condition~\eqref{eq:pdncond3} reads as: 
\begin{equation}\label{eq:condc_worked1}
    F^{ij\gamma}(\vb*{v})\,L^\ell_\gamma(\vb*{u}) = q^{\alpha\beta} \,q^{\gamma\delta} \,R^i_\beta(\vb*{v}) R^j_\delta(\vb*{v}) 
    \bigl[(\underline{\partial}_s L^\ell_\gamma) \, L^s_\alpha  
    -  (\underline{\partial}_s L^\ell_\alpha )    \, L^s_\gamma \bigr](\vb*{u}),
\end{equation}
where $F^{ij\gamma}$ is the $v$-dependent expression
\begin{equation}
\begin{split} 
    F^{ij\gamma}(\vb*{v}) &= 
    \bigl[(\partial_s \underline{k}^{\alpha\beta}) R^i_\alpha R^j_\beta
    +\underline{k}^{\alpha\beta}(\partial_s R^i_\alpha)  R^j_\beta
    +\underline{k}^{\alpha\beta} R^i_\alpha  (\partial_s R^j_\beta)\bigr]\, q^{\gamma\delta}\, R^s_\delta \\
    &~~
    + q^{\gamma\beta}\, \underline k^{\alpha \delta}\, R^s_\alpha\bigl( (\partial_s R^j_\beta)\,  R^i_\delta - (\partial_s R^i_\beta)\,  R^j_\delta \bigr)\,. 
\end{split} 
\end{equation}
The bracket in the RHS of~\eqref{eq:condc_worked1} is by definition 
\begin{equation}
    (\underline{\partial}_s L^\ell_\gamma) \, L^s_\alpha  
    -  (\underline{\partial}_s L^\ell_\alpha )    \, L^s_\gamma = \bigl[ L_\gamma , L_\alpha \bigr]^\ell = c^{\lambda}_{\gamma \alpha}\,L^\ell_{\lambda} = -c^{\lambda}_{\alpha \gamma}\,L^\ell_{\lambda}.  
\end{equation}
The equation~\eqref{eq:condc_worked1} becomes 
\begin{equation}\label{eq:condc_worked2}
    F^{ij\lambda}(\vb*{v})\,L^\ell_\lambda(\vb*{u}) + q^{\alpha\beta} \,q^{\gamma\delta} \,c^\lambda_{\alpha \gamma} \,R^i_\beta(\vb*{v}) R^j_\delta(\vb*{v})\, 
    L^\ell_\lambda(\vb*{u})=0,
\end{equation}
from which we can collect the coefficients of $L^\ell_\lambda(\vb*{u})$. We write the derivatives of the right  invariant frame applying the results of part $(i)$ as  
\begin{equation}
    \partial_s R^i_\alpha(\vb*{v}) = C^\lambda_{s \alpha}(\vb*{v}) R^i_{\lambda}(\vb*{v}) = -c^\lambda_{s \alpha}\,R^i_{\lambda}(\vb*{v}).
\end{equation} 
The equation~\eqref{eq:condc_worked2} for the coefficients of $L^\ell_\lambda(\vb*{u})$ becomes, after relabelling the indices: 
\begin{equation}
\begin{split} 
    &\bigl[(\partial_s \underline{k}^{\beta\delta}) {R^i_\beta R^j_\delta}
    -\underline{k}^{\alpha\delta}\,c^\beta_{s \alpha} {R^i_\beta  R^j_\delta}
    -\underline{k}^{\beta\alpha}\,c^\delta_{s \alpha} {R^i_\beta R^j_\delta }   \bigr]\, q^{\lambda\sigma}\, R^s_\sigma \\
    &~~
    + q^{\lambda\sigma}\, \underline k^{\alpha \delta}\,c^\beta_{s \sigma} R^s_\alpha \, {R^i_\beta \, R^j_\delta } 
    - q^{\lambda\sigma}\, \underline k^{\alpha \beta}\,c^\delta_{s \sigma}\,R^s_\alpha { R^i_\beta  \, R^j_\delta}  
    + q^{\alpha\beta} \,q^{\gamma\delta} \,c^\lambda_{\alpha \gamma} {R^i_\beta R^j_\delta} =0 .
\end{split} 
\end{equation}
Using the non-degeneracy of $R^i_\beta$ and $R^j_\delta$ and the skew-symmetry of $\underline{k}$, we get 
\begin{equation}\label{eq:condc_worked3}
\begin{split} 
    q^{\lambda\sigma}(\partial_s \underline{k}^{\beta\delta}) R^s_\sigma 
    + q^{\lambda\sigma}\, \underline k^{\alpha \delta}\,c^\beta_{\alpha \sigma} 
    +q^{\lambda\sigma}\,\underline{k}^{\alpha\beta}\,c^\delta_{\alpha \sigma} 
    + q^{\alpha\beta} \,q^{\gamma\delta} \,c^\lambda_{\alpha \gamma} =0 .
\end{split} 
\end{equation}
We can now introduce 
\begin{equation}
    \underline{f}^{\beta \delta}_\sigma :=-q_{\lambda \sigma}\, q^{\alpha\beta} \,q^{\gamma\delta} \,c^\lambda_{\alpha \gamma}
\end{equation}
i.e.\ ${}^t q$ is the homomorphism from $(\mathfrak g^*,f^{\alpha \beta}_\sigma)$ to 
$(\mathfrak g,c_{\alpha \beta}^\sigma)$. Since $\underline k$ is a function of 
$\vb*{v}$ only, we set $ \mathcal R_\sigma = R^s_\sigma \partial_s$,  
then~\eqref{eq:condc_worked3} reads as:
\begin{equation}\label{eq:system_underk}
    \mathcal R_\sigma \underline k^{\beta \delta} = 
    \underline{k}^{\beta\alpha}\,c^\delta_{\alpha \sigma} 
    +\underline k^{\delta \alpha}\,c^\beta_{\alpha \sigma} 
    +\underline{f}^{\beta \delta}_\sigma,
\end{equation}
after using again the skew-symmetry of $\underline k$. 

\medskip
$(iii)$ We can now compare the two structures we obtained by considering the fundamental relation between left and right invariant vector fields, i.e.\ $R_\alpha = \textup{Ad}_{\vb*{u}} L_\alpha$, in the expression of $h$ in~\eqref{eq:fact}. We get:
\begin{equation}
    h = k^{\alpha\beta}\,L_\alpha \otimes  L_\beta = \underline{k}^{\alpha\beta}\,R_\alpha \otimes  R_\beta = \underline k^{\alpha \beta}\,(\textup{Ad}_{\vb*{u}})^\mu_\alpha\,(\textup{Ad}_{\vb*{u}})^\nu_\beta L_\mu \otimes L_\nu,
\end{equation}
from which we identify the relation between $k$ and $\underline k$, i.e.
\begin{equation}
    \underline k^{\alpha \beta} = (\textup{Ad}_{\vb*{u}^{-1}}^{(2)})_{\mu \nu}^{\alpha \beta}\,k^{\mu \nu}, \qquad \textup{or } \underline k = \textup{Ad}_{\vb*{u}^{-1}}^{(2)}\, k.
\end{equation}
We evaluate the derivative of this expression at the identity $e$, where $\underline k^{\alpha \beta}(\vb*{e}) = k^{\alpha \beta}(\vb*{e}) = \kappa^{\alpha \beta}$ with $\kappa^{\alpha \beta}$ initial datum, since $R_\alpha = \text{Ad}_{\vb*{u}} L_\alpha$ and $\text{Ad}_{\vb*{e}}=\text{Id}$. We consider the Lie derivatives along the left and right invariant vector fields, i.e.\ $\mathcal L_\rho = L^s_\rho \partial_s$ and $\mathcal R_\rho = R^s_\rho \partial_s$, and set $\phi(\vb*{u})=\textup{Ad}_{\vb*{u}^{-1}}^{(2)}$. The right derivative is then:
\begin{equation}
    \mathcal R_\rho \underline k = (\mathcal R_\rho \phi) k + \phi (\mathcal R_\rho \,k).
\end{equation}
At $u=e$ we have $\phi(\vb*{e})=\text{Id}$ and $\mathcal R_\rho\,k(\vb*{e})=\mathcal L_\rho\,k(\vb*{e})$, hence 
\begin{equation}\label{eq:system_LR_identity}
    \mathcal R_\rho \underline k^{\alpha \beta}(\vb*{e}) = \bigl((\mathcal R_\rho \phi)(\vb*{e})\, k(\vb*{e})\bigr)^{\alpha \beta} + \mathcal L_\rho k^{\alpha \beta}(\vb*{e}). 
\end{equation}
The derivative of the adjoint action at the identity is 
\begin{equation}
(\mathcal R_\rho\text{Ad}_{\vb*{u}^{-1}})(\vb*{e}) = - \text{ad}_{e_\rho} \quad \implies \quad (\mathcal R_\rho\text{Ad}_{\vb*{u}^{-1}})(\vb*{e})^\lambda_{\sigma} = -c^\lambda_{\rho \sigma}, 
\end{equation}
therefore for the expression of $\phi$ we have 
\begin{equation}
    \bigl((\mathcal R_\rho\phi)(\vb*{e})\,k(\vb*{e}) \bigr)^{\alpha \beta} = -c^\alpha_{\rho \sigma} \kappa^{\sigma \beta} - c^\beta_{\rho \sigma} \kappa^{\alpha \sigma}. 
\end{equation}
From part $(i)$ the system in $k$~\eqref{eq:k_system} evaluated at $u=e$ reads as:
\begin{equation}
    \mathcal L_\rho k^{\alpha\beta}(\vb*{e}) = \kappa^{\beta\sigma}\, c^\alpha_{\rho\sigma} + \kappa^{\sigma\alpha}\, c^\beta_{\rho\sigma} + f^{\alpha\beta}_\rho.
\end{equation}
With this, we can rewrite equation~\eqref{eq:system_LR_identity} as 
\begin{equation}\label{eq:system_underk_identity_fromL}
    \mathcal R_\rho \underline k^{\alpha \beta}(\vb*{e}) = 2\kappa^{\beta\sigma}\, c^\alpha_{\rho\sigma} + 2\kappa^{\sigma\alpha}\, c^\beta_{\rho\sigma} + f^{\alpha\beta}_\rho, 
\end{equation}
after using the skew-symmetry of $\kappa$. From part $(ii)$, the system in $\underline k$~\eqref{eq:system_underk} evaluated at $\vb*{u}=\vb*{e}$ reads as: 
\begin{equation}\label{eq:system_underk_identity}
    \mathcal R_\rho \underline k^{\alpha\beta}(\vb*{e})
    = \underline k^{\beta\sigma}(\vb*{e})\, c^\alpha_{\rho\sigma}
    + \underline k^{\sigma\alpha}(\vb*{e})\, c^\beta_{\rho\sigma}
    + \underline f^{\alpha\beta}_\rho
    = \kappa^{\beta\sigma}\, c^\alpha_{\rho\sigma} + \kappa^{\sigma\alpha}\, c^\beta_{\rho\sigma} + \underline f^{\alpha\beta}_\rho.
\end{equation}
Equating \eqref{eq:system_underk_identity} and \eqref{eq:system_underk_identity_fromL} we obtain 
\begin{equation}
    \underline f^{\alpha\beta}_\rho = f^{\alpha\beta}_\rho + \kappa^{\beta\sigma}\, c^\alpha_{\rho\sigma} + \kappa^{\sigma\alpha}\, c^\beta_{\rho\sigma}.
\end{equation}
Finally, we consider the inhomogeneous system in $k$~\eqref{eq:k_system}, linear with constant coefficients. Its solution splits into the general solution of the homogeneous system and a particular solution with vanishing initial datum. We have:
\begin{equation}\label{eq:k_general_solution}
    k^{\alpha\beta}(\vb*{u}) = p_0^{\alpha\beta}(\vb*{u}) + \bigl(\mathrm{Ad}^{(2)}_{\vb*{u}^{-1}} \kappa\bigr)^{\alpha\beta},
\end{equation}
with $\kappa^{\alpha\beta} \coloneqq k^{\alpha\beta}(\vb*{e})$ the initial datum, and $p_0^{\alpha\beta}$ the Poisson--Lie bivector of the bialgebra $(\mathfrak g, \mathfrak g^*)$, i.e.\ the unique solution of the inhomogeneous system
\begin{equation}\label{eq:p0_system}
    \mathcal L_\gamma p_0^{\alpha\beta}
    = c^\alpha_{\gamma\varepsilon}\, p_0^{\varepsilon\beta}
    + c^\beta_{\gamma\varepsilon}\, p_0^{\alpha\varepsilon}
    + f^{\alpha\beta}_\gamma,
    \qquad
    p_0^{\alpha\beta}(\vb*{e}) = 0.
\end{equation}

\end{proof}

Finally, the following Lemma translates condition~\eqref{eq:pdncond4}.

\begin{lemma}\label{lem:cond_d}
From~\eqref{eq:pdncond4}, the factorisation~\eqref{eq:fact} implies 
\begin{equation}\label{eq:GYBE_statement}
    f^{\alpha\beta}_\lambda\, k^{\lambda\gamma}
    + f^{\beta\gamma}_\lambda\, k^{\lambda\alpha}
    + f^{\gamma\alpha}_\lambda\, k^{\lambda\beta} = 0, 
\end{equation}
i.e.\ $k$ solves GYBE~\eqref{eq:GYBE}.
\end{lemma}

\begin{proof}
In the left invariant frame~\eqref{eq:vect_frame_L}, with the definition of $c^\lambda_{\alpha\beta}$~\eqref{eq:c_definition}, condition~\eqref{eq:pdncond4} reduces to
\begin{equation}\label{eq:on_site_reduced}
\begin{aligned}
k^{\alpha\delta}\, (\mathcal L_\delta k^{\beta\gamma})
+ k^{\beta\delta}\, (\mathcal L_\delta k^{\gamma\alpha})
+ k^{\gamma\delta}\, (\mathcal L_\delta k^{\alpha\beta})  + k^{\alpha\delta} k^{\beta\varepsilon}\, c^\gamma_{\delta\varepsilon}
+ k^{\beta\delta} k^{\gamma\varepsilon}\, c^\alpha_{\delta\varepsilon}
+ k^{\gamma\delta} k^{\alpha\varepsilon}\, c^\beta_{\delta\varepsilon} = 0,
\end{aligned}
\end{equation}
where $\mathcal L_\delta = L^s_\delta\, \partial_s$.
We substitute the system~\eqref{eq:k_system} for $k$,
into the first three terms of~\eqref{eq:on_site_reduced}. 
Using the skew-symmetry of $k$ and of $c$, together with the Jacobi identity for the structure constants $c^\sigma_{\alpha\beta}$ of $\mathfrak g$, we obtain the following expression in $f$-dependent terms only: condition~\eqref{eq:on_site_reduced} becomes
\begin{equation}\label{eq:GYBE_discrete}
    k^{\alpha\delta}\, f^{\beta\gamma}_\delta
    + k^{\beta\delta}\, f^{\gamma\alpha}_\delta
    + k^{\gamma\delta}\, f^{\alpha\beta}_\delta = 0,
\end{equation}
or equivalently,
\begin{equation}\label{eq:GYBE_discrete_alt}
    f^{\alpha\beta}_\lambda\, k^{\lambda\gamma}
    + f^{\beta\gamma}_\lambda\, k^{\lambda\alpha}
    + f^{\gamma\alpha}_\lambda\, k^{\lambda\beta} = 0.
\end{equation}

This is precisely the generalised classical Yang--Baxter equation~\eqref{eq:GYBE} for $k$ with respect to the cocycle~$f$. Evaluating it at the identity $\vb*{u} = \vb*{e}$, where $k(\vb*{e}) = \kappa$, we obtain the GYBE for the constant initial datum:
\begin{equation}\label{eq:GYBE_kappa}
    f^{\alpha\beta}_\lambda\, \kappa^{\lambda\gamma}
    + f^{\beta\gamma}_\lambda\, \kappa^{\lambda\alpha}
    + f^{\gamma\alpha}_\lambda\, \kappa^{\lambda\beta} = 0.
\end{equation}
Together with the comparison relation of part $(iii)$,
equation~\eqref{eq:GYBE_kappa} is equivalent to the statement that the deformed cocycle $\underline f$ defines a Lie algebra structure on $\mathfrak g^*$. 
\end{proof}

The most interesting case for applications is the quasi-Frobenius one, in which the 
isomorphism $q$ coincides (up to sign) with the non-degenerate classical $r$-matrix 
$k$. This is the  triangular case of \cite[\S3]{DubrovinNovikov1989}, which we 
now state in our notation.

\begin{corollary}[quasi-Frobenius]\label{thm:Dubrovin_quasiF} 
Let $r \in \mathfrak g \wedge \mathfrak g$ be a non-degenerate skew-symmetric solution of the~CYBE~\eqref{eq:CYBE}, and let $q = -k \equiv r$. Then the brackets~\eqref{eq:gen_onsite} and~\eqref{eq:gen_intersite} reduce to
    \begin{align}
    \label{eq:qF_onsite}
        h^{ij}_n&=\{ u^i_n , u^j_{n} \} = -r^{\alpha \beta}\big(L^i_\alpha(\vb*{u}_n)\, L^j_\beta(\vb*{u}_{n}) +  R^i_\alpha(\vb*{u}_n)\, R^j_\beta(\vb*{u}_{n}) \big)\,, \\[1mm]
    \label{eq:qF_intersite}
        g^{ij}_n&=\{ u^i_n , u^j_{n+1} \} = r^{\alpha \beta} L^i_\alpha(\vb*{u}_n)\, R^j_\beta(\vb*{u}_{n+1}) \,.
    \end{align}
\end{corollary}

The condition of being quasi-Frobenius in \Cref{thm:Dubrovin_quasiF} is far more restrictive than the one of
the full \Cref{thm:Dubrovin_general}. However, even for a given Lie algebra, the tackling
the general conditions requires to solve huge systems of coupled nonlinear algebraic equations,
which is, in general, a very difficult task. On the other hand, \Cref{thm:Dubrovin_quasiF}
requires only to solve the Schouten bracket condition, which is a manageable task with the help 
of computer algebra software. Moreover, as we will detail later, this classes of Lie algebras 
drawn lot of attention in the years after Dubrovin's original work~\cite{Dubrovin1989}, 
allowing us to build a plethora of examples of this type. 

\begin{remark}[Sklyanin bracket]\label{rem:onsite_vs_skly}
The on-site bracket~\eqref{eq:qF_onsite} is a Poisson--Lie bracket on $G$, but it is not the only one. For instance, the standard Sklyanin bracket~\cite{Sklyanin1982} (see also~\cite{Kosmann2004}) is another:
\begin{equation}
\{f,g\}_{\textup{Sk}} = r^{\alpha\beta}\bigl(L_\alpha (f) L_\beta (g) - R_\alpha (f) R_\beta (g)\bigr),
\end{equation}
differing from~\eqref{eq:qF_onsite} by a sign. This is not a convention, in fact~\eqref{eq:qF_onsite} is the unique Poisson--Lie bracket on $G$ compatible with the inter-site bracket determined by $q=r$, i.e.\ the unique one making the full lattice bracket satisfy the Jacobi identity.
\end{remark}

The application of~\Cref{thm:Dubrovin_general} for self-dual bialgebras is given by the non-abelian Lie algebra~$\mathfrak {aff}(1)$ considered by
Dubrovin~\cite[Ex.\ 1]{Dubrovin1989}, Parodi~\cite[Rem.\ 2.4.5]{ParodiThesis}, and  Casati--Valeri~\cite[Sec.\ 3.1.1]{CasatiValeri}. Dubrovin~\cite[Ex.\ 0]{Dubrovin1989} and Casati--Valeri~\cite[Sec.\ 3.1.2]{CasatiValeri} also consider the case of the self-dual abelian Lie algebra. In Dubrovin's paper~\cite{Dubrovin1989} \Cref{thm:Dubrovin_quasiF} appears as an example (Ex.\ 2) without any explicit case discussed.

\section{Quasi-Frobenius Lie algebras and their dDGP brackets}
\label{sec:examples}

We construct explicit examples of dDGP brackets by applying the geometric characterisation of~\Cref{thm:Dubrovin_quasiF}. The starting point is a quasi-Frobenius Lie algebra~$\mathfrak g$, then the construction proceeds as follows. 
We identify at least one non-degenerate skew-symmetric solution $r$ of the CYBE~\eqref{eq:CYBE} on $\mathfrak g$, then we choose a faithful representation $\rho$ of $\mathfrak g$ and use it to parametrise the associated Poisson--Lie group $G = \exp(\mathfrak g)$ and compute the left and right invariant vector fields $L_\alpha$ and $R_\beta$ on $G$. 
All these elements enter the construction of the on-site $h^{ij}_n$~\eqref{eq:qF_onsite} and inter-site $g^{ij}_n$~\eqref{eq:qF_intersite} brackets. Wherever the resulting expressions remain sufficiently compact, we give their form explicitly. 

The Section is organised in two parts. In~\Cref{sec:qf} we consider the four-dimensional 
case and classify all real quasi-Frobenius Lie algebras of dimension four. This
is the first non-trivial case, becase as underlined a the end of the previous
Section, for the two dimensional case one can only consider the abelian Lie
algebra and the non-abelian one $\mathfrak{aff}(1)$. To produce examples
in dimension four we will use 
the well-known classification of low-dimensional real Lie
algebras. 
In~\Cref{sec:filiform} we consider quasi-Frobenius $\mathbb{N}$-graded filiform 
Lie algebras, including two families in arbitrary even dimension $2k$.

\subsection{Four-dimensional Lie algebras}
\label{sec:qf}

We consider the problem of finding all \emph{real} 
quasi-Frobenius Lie algebras in dimension four using the well-known classification
classification of low dimensional real Lie algebras given by 
Mubarakzyanov~\cite{Mubarakzyanov1963a,Mubarakzyanov1963b,Mubarakzyanov1963c,Mubarakzyanov1966} 
and use them to produce a list of dDGP brackets through~\Cref{thm:Dubrovin_quasiF}.
To this end we use the list of indecomposable four-dimensional Lie algebras as 
presented in~\cite[Chap. 17]{SnobWinternitz2017book}, adding the decomposable 
ones through the list presented in~\cite[Chap. 16]{SnobWinternitz2017book}.
We recall that in~\cite{SnobWinternitz2017book} Lie algebras are addressed
in the following way: 
\begin{enumerate}[$(a)$]
    \item \textbf{nilpotent Lie algebras} are denoted as $\mathfrak{n}_{d,p}$, 
        where $d$ is the dimension, and $p$ a progressive numbering;
    \item \textbf{solvable Lie algebras} are denoted as $\mathfrak{s}_{d,p}$, 
        where $d$ is the dimension, and $p$ a progressive numbering;
    \item \textbf{simple Lie algebras} are called using the associated classical
        Lie algebra name, e.g.\ $\mathfrak{sl}(n,\mathbb{K})$ for the
        special linear algebra of $n\times n$ matrices over field $\mathbb{K}$
        and so on so forth.
\end{enumerate}
For the sake of completeness, we highlight the presence of free parameters in the algebras with the notation $\mathfrak{g}(\vb*{\alpha})$, where $\vb*{\alpha}=(\alpha_1,\ldots,\alpha_K)$ is the vector of free parameters. We exclude from the search the abelian Lie algebra as it was considered earlier in~\Cref{ex:constg}.

Throughout this subsection we will use the following conventions:
\begin{enumerate}
    \item the basis of the Lie algebra $\mathfrak{g}$ will be denoted as $e_1$, 
        $e_2$, $e_3$, $e_4$;
    \item we will denote a representation of $\mathfrak{g}$ as $\rho\colon\mathfrak{g}\longrightarrow\gl(N,\mathbb{R})$;
    \item the associated Poisson--Lie group will be built as $G=\exp(\mathfrak{g})$,
        as exponentiation of one-parameter subgroups;
    \item the group $G$ will be parametrised by the coordinates $u^1_n$, $u^2_n$, 
        $u^3_n$, $u^4_n$.
\end{enumerate}

A list of the admissible Lie algebras can be found in~\Cref{tab:alglist}, while
in the next paragraphs we will detail the construction. In total we found ten
quasi-Frobenius Lie algebras, three of them being decomposable and seven being
indecomposable. Interestingly enough, all these Lie algebras are either nilpotent,
solvable, or direct sums of nilpotent and solvable. No simple Lie algebra appears
as summand. This is for instance different
from the continuous case of $1+0$ Hamiltonian operators where simple 
Lie algebras play a preeminent r\^ole in the description, see~\cite{GOSV_lie}.

 {\small
\begin{table}[hbt]
    \centering
    \begin{tabular}{c  c c c}
         & Algebra & Non-zero commutation relations & Centre
         \\[1mm]
         \midrule[1pt] \\[-3mm]
         \multirow{5}{*}{\rotatebox{90}{Decomposable}}
         & $\mathfrak{s}_{2,1}\oplus 2\mathfrak{n}_{1,1}$ 
         & $\comm{e_2}{e_1} = e_1$ & $\langle e_3,e_4\rangle$
         \\[2mm]
         \cmidrule(lr){2-4} \\[-3mm]
         & $\mathfrak{s}_{2,1}\oplus\mathfrak{s}_{2,1}$
         & $\comm{e_2}{e_1} = e_1$, $\comm{e_4}{e_3} = e_3$ 
         & $\Set{0}$
         \\[2mm]
         \cmidrule(lr){2-4} \\[-3mm]
         & $\mathfrak{n}_{3,1} \oplus \mathfrak{n}_{1,1}$
         & $\comm{e_2}{e_3} = e_1$ 
         & $\langle e_1,e_4\rangle$
         \\[1mm]
         \midrule[1pt] \\[-3mm]
         \multirow{15}{*}{\rotatebox{90}{Indecomposable}}
         & $\mathfrak{n}_{4,1}$
         & $\comm{e_2}{e_4} = e_1$, $\comm{e_3}{e_4} = e_2$ 
         & $\langle e_1\rangle$
         \\[2mm]
         \cmidrule(lr){2-4} \\[-3mm]
         & $\mathfrak{s}_{4,1}$
         & $\comm{e_4}{e_2} = e_1$, $\comm{e_4}{e_3} = e_3$ 
         & $\langle e_1\rangle$
         \\[2mm]
         \cmidrule(lr){2-4} \\[-3mm]
         & $\mathfrak{s}_{4,8}(a)$
         & $\comm{e_2}{e_3} = e_1$, $\comm{e_4}{e_1} = (1+a)e_1$,
         & \multirow{2}{*}{$\Set{0}$}
         \\[1mm]
         & $-1<a\leq 1$, $a\neq 0$ 
         & $\comm{e_4}{e_2} = e_2$, $\comm{e_4}{e_3} = a e_3$ 
         \\[1mm]
         \cmidrule(lr){2-4} \\[-3mm]
         & $\mathfrak{s}_{4,9}(\alpha)$, 
         & $\comm{e_2}{e_3} = e_1$, $\comm{e_4}{e_1} = 2\alpha e_1$, 
         & \multirow{2}{*}{$\Set{0}$}
         \\[1mm]
         & $\alpha>0$
         & $\comm{e_4}{e_2} = \alpha e_2 - e_3$, $\comm{e_4}{e_3} = e_2 + \alpha e_3$ 
         & 
         \\[1mm]
         \cmidrule(lr){2-4} \\[-3mm]
         & $\mathfrak{s}_{4,10}$
         & $\comm{e_2}{e_3} = e_1$, $\comm{e_4}{e_1} = 2 e_1$, $\comm{e_4}{e_2} = e_2$, $\comm{e_4}{e_3} = e_2 + e_3$ 
         & $\Set{0}$
         \\[2mm]
         \cmidrule(lr){2-4} \\[-3mm]
         & $\mathfrak{s}_{4,11}$
         & $\comm{e_2}{e_3} = e_1$, $\comm{e_4}{e_1} = e_1$, $\comm{e_4}{e_2} = e_2$ 
         & $\Set{0}$
         \\[2mm]
         \cmidrule(lr){2-4} \\[-3mm]
         & $\mathfrak{s}_{4,12}$
         & $\comm{e_3}{e_1} = e_1$, $\comm{e_3}{e_2} = e_2$, $\comm{e_4}{e_1} = -e_2$, $\comm{e_4}{e_2} = e_1$ 
         & $\Set{0}$
         \\[2mm]
         \bottomrule
    \end{tabular}
    \caption{List of 4-dimensional Lie algebras admitting a non-degenerate skew-symmetric \(r\)-matrix.}
    \label{tab:alglist}
\end{table}
}

We recall that to build the dDGP bracket, the computation of the left and right vector fields 
of the Poisson--Lie group is needed. To compute these objects we use the following
result:

\begin{proposition}
    Let $\mathfrak{g} = \langle e_1,\ldots,e_d\rangle$ be Lie algebra over the 
    field  $\mathbb{K}$ and $\rho \colon \mathfrak{g}\to \mathfrak{gl}(N,\mathbb{K})$ 
    a faithful, i.e.\ injective, representation. 
    Then, the right invariant vector fields on the Lie group $G=\exp (\mathfrak{g})$
    are generated by the vectors $R_i \in TG$, $i=1,\ldots,d$ such that:
    \begin{equation}
            R_i(\rho(g)) = \rho(e_i) \rho(g),
            \quad \forall g\in G,
            \label{eq:vRbasis}
    \end{equation}
    while the left invariant vector fields 
    are generated by the vectors $L_i \in TG$, $i=1,\ldots,d$ such that:
    \begin{equation}
            L_i(\rho(g)) = \rho(g)\rho(e_i),
            \quad \forall g\in G,
            \label{eq:vLbasis}
    \end{equation}
    where the action of vectors fields on matrices is component-wise. 
    \label{prop:lrinv}
\end{proposition}

The proof of this fact is a trivial application of the definitions,
see e.g.~\cite[\S1.4]{Olver1986}.

To apply~\Cref{prop:lrinv} it is required to know of a faithful 
representation of the underlying Lie algebra. It is well known from Ado's 
theorem~\cite[Chap.~VI]{Jacobson1962} that every finite-dimensional Lie algebra in 
characteristic zero admits a faithful representation. The representations of the Lie 
algebras of dimension lower than four have been systematically produced 
in~\cite{Ghanam_etal2005} using the list of these algebras as presented 
in~\cite{Montreal1976}. However, the list in~\cite{Montreal1976} is slightly different 
from the one in~\cite[Chap.~17]{SnobWinternitz2017book}, as it is built with different 
criteria. We refer to~\cite[Chap.~15]{SnobWinternitz2017book} for some additional 
comments on the differences between the two lists.

Since we adhere to the list of indecomposable four-dimensional real Lie algebras 
in~\cite[Chap.~15]{SnobWinternitz2017book}, we will build the necessary faithful 
representation using the following strategy: if the centre of the algebra is trivial, 
we take the adjoint representation as the faithful representation; if the centre is non-trivial, since, following~\Cref{tab:alglist}, it is at most one-dimensional, 
we use the following result:
\begin{theorem}[\cite{Ghanam_etal2005}, Thm.~3.1]
    Suppose that the $d$-dimensional Lie algebra $\mathfrak{g}$ has a 
    codimension one abelian ideal. Choose a basis such that the
    codimension one abelian ideal is $\mathfrak{i}=\langle e_1,\ldots, e_{d-1}\rangle$. 
    Then the map 
    $\rho\colon\mathfrak{g}\longrightarrow\mathfrak{gl}(d, \mathbb{R})$ whose
    action on the generators is the following:
    \begin{subequations}
        \begin{align}
            \rho(e_i) &= E_{i,d}, \quad i=1,\ldots,d-1,
            \\
            \rho(e_d) &= (\ad(e_d)_{k,l})_{k,l}^d,
        \end{align}
    \end{subequations}
    where $\ad(e_d)$ is the adjugate of $e_d$,
    is a faithful representation of $\mathfrak{g}$ as a subalgebra of 
    $\mathfrak{gl}(d, \mathbb{R})$.
    \label{thm:repr}
\end{theorem}

In the case of decomposable algebras we use the direct sum of representations,
see e.g.~\cite[\S 1.1]{FultonHarris1991}.

\subsubsection{$\mathfrak{s}_{2,1}\oplus 2\mathfrak{n}_{1,1}$}

Let us consider the Lie algebra
$\mathfrak{s}_{2,1}\oplus 2\mathfrak{n}_{1,1}$
whose only non-zero commutation relation is:
\begin{equation}
    \comm{e_2}{e_1} = e_1.
\end{equation}

\begin{remark}
    We recall that the Lie algebra $\mathfrak{s}_{2,1}$ is often
    denoted as $\mathfrak{aff}(1,\RR)$, as it is the Lie algebra of the
    affine group on the real line. 
\end{remark}

From a direct computation it is possible to show that the following
skew-symmetric matrix:
\begin{equation}
r_{\mathfrak{s}_{2,1}\oplus 2\mathfrak{n}_{1,1}} =
\begin{pmatrix}
    0 & c_1 & c_2 & c_3 \\
    -c_1 & 0 & 0 & 0 \\
    -c_2 & 0 & 0 & c_4 \\
    -c_3 & 0 & -c_4 & 0
\end{pmatrix}
\end{equation}
is a solution of the CYBE~\eqref{eq:CYBE}.

Next, we consider a representation in the space of $4\times 4$
matrices
$\rho\colon \mathfrak{s}_{2,1}\oplus 2\mathfrak{n}_{1,1}\longrightarrow\gl(4,\mathbb{R})$ acting as follows on the elements of the basis:
\begin{equation}
    \begin{gathered}
\rho(e_1)=
\begin{pmatrix}
0&1&0&0\\0&0&0&0\\0&0&0&0\\0&0&0&0
\end{pmatrix},\quad
\rho(e_2)=
\begin{pmatrix}
1&0&0&0\\0&0&0&0\\0&0&0&0\\0&0&0&0
\end{pmatrix},
\\
\rho(e_3)=
\begin{pmatrix}
0&0&0&0\\0&0&0&0\\0&0&1&0\\0&0&0&0
\end{pmatrix},\quad
\rho(e_4)=
\begin{pmatrix}
0&0&0&0\\0&0&0&0\\0&0&0&0\\0&0&0&1
\end{pmatrix}.            
    \end{gathered}
\end{equation}
The representation $\rho$ is clearly faithful, as it is obtained as direct sum of
faithful representations (observe that for $\mathfrak{s}_{2,1}$ the adjint 
representation is faithful and the abelian algebra has a diagonal representation). 
This gives the following  (local) parametrisation of the group $S_{2,1}\times \RR^2$:
\begin{equation}
    S_{2,1}\times \RR^2
    =
    \Set{
\begin{pmatrix}
\exp(u^2_n) & \exp(u^2_n) u_n^1 & 0 & 0 \\
0 & 1 & 0 & 0 \\
0 & 0 & \exp(u^3_n) & 0 \\
0 & 0 & 0 & \exp(u^4_n)
\end{pmatrix} | u_n^i \in \RR}.
\end{equation}
So, again by a direct computation we obtain the following left and right invariant
vector fields

in matrix form:
\begin{equation}
L =
\begin{pmatrix}
    \exp(-u^2_n) & 0 & 0 & 0 \\
    0 & 1 & 0 & 0 \\
    0 & 0 & 1 & 0 \\
    0 & 0 & 0 & 1
\end{pmatrix},
\qquad
R = 
\begin{pmatrix}
1 & 0 & 0 & 0 \\
-u^1_n & 1 & 0 & 0 \\
0 & 0 & 1 & 0 \\
0 & 0 & 0 & 1
\end{pmatrix}.
\end{equation}

The on-site $h^{ij}_n$ and inter-site $g^{ij}_n$ elements of the dDGP bracket in \Cref{thm:Dubrovin_quasiF} are:
\begin{subequations}
\begin{align}
h_n&=\left(\begin{array}{ c c c c}
0 & -{c_1} ({\mathrm e}^{-u^2_{n}}+1) & -{c_2} ({\mathrm e}^{-u^2_{n}}+1) & -{c_3} ({\mathrm e}^{-u^2_{n}}+1) 
\\[1mm]
 {c_1} ({\mathrm e}^{-u^2_{n}}+1) & 0 & {c_2} {u^{1}}_{n} & {c_3} {u^{1}}_{n} 
\\[1mm]
 {c_2} ({\mathrm e}^{-u^2_{n}}+1) & -{c_2} {u^{1}}_{n} & 0 & -2 {c_4}  
\\[1mm]
 {c_3} ({\mathrm e}^{-u^2_{n}}+1) & -{c_3} {u^{1}}_{n} & 2 {c_4}  & 0 
\end{array}\right), \\[2mm]
g_n&=\left(\begin{array}{ cccc}
0 & {\mathrm e}^{-u^2_{n}} {c_1}  & {\mathrm e}^{-u^2_{n}} {c_2}  & {\mathrm e}^{-u^2_{n}} {c_3}  
\\[1mm]
 -{c_1}  & {c_1} u^1_{n+1} & 0 & 0 
\\[1mm]
 -{c_2}  & {c_2} u^1_{n+1} & 0 & {c_4}  
\\[1mm]
 -{c_3}  & {c_3} u^1_{n+1} & -{c_4}  & 0 
\end{array}\right).
\end{align}
\end{subequations}

\subsubsection{$\mathfrak{s}_{2,1}\oplus\mathfrak{s}_{2,1}$}

Let us consider the Lie algebra
$\mathfrak{s}_{2,1}\oplus\mathfrak{s}_{2,1}$
whose non-zero commutation relation are:
\begin{equation}
    \comm{e_2}{e_1} = e_1, \quad \comm{e_4}{e_3} = e_3.
\end{equation}

From a direct computation it is possible to show that the following
skew-symmetric matrix:
\begin{equation}
r_{\mathfrak{s}_{2,1}\oplus\mathfrak{s}_{2,1}} =
\begin{pmatrix}
    0 & c_1 & c_2 & 0 \\
    -c_1 & 0 & 0 & 0 \\
    -c_2 & 0 & 0 & c_3 \\
    0 & 0 & -c_3 & 0
\end{pmatrix}
\end{equation}
is a solution of the CYBE~\eqref{eq:CYBE}.

Next, we consider a representation in the space of $4\times 4$
matrices
$\rho\colon \mathfrak{s}_{2,1}\oplus \mathfrak{s}_{2,1}\longrightarrow\gl(4,\mathbb{R})$ acting as follows on the elements of the basis:
\begin{equation}
    \begin{gathered}
\rho(e_1)=
\begin{pmatrix}
0&1&0&0\\0&0&0&0\\0&0&0&0\\0&0&0&0
\end{pmatrix},\quad
\rho(e_2)=
\begin{pmatrix}
1&0&0&0\\0&0&0&0\\0&0&0&0\\0&0&0&0
\end{pmatrix},
\\
\rho(e_3)=
\begin{pmatrix}
0&0&0&0\\0&0&0&0\\0&0&0&1\\0&0&0&0
\end{pmatrix},\quad
\rho(e_4)=
\begin{pmatrix}
0&0&0&0\\0&0&0&0\\0&0&1&0\\0&0&0&0
\end{pmatrix}.            
    \end{gathered}
\end{equation}
The representation $\rho$ is obtained easily as direct sum of faithful 
representations of $\mathfrak{s}_{2,1}$ and it is clearly faithful. This gives the following 
(local) parametrisation of the group $S_{2,1}\times S_{2,1}$:
\begin{equation}
    S_{2,1}\times S_{2,1}
    =
    \Set{
\begin{pmatrix}
\exp(u^2_n) & \exp(u^2_n) u_n^1 & 0 & 0 \\[1mm]
0 & 1 & 0 & 0 \\[1mm]
0 & 0 & \exp(u^4_n) & \exp(u^4_n)u_n^3 \\[1mm]
0 & 0 & 0 & 1
\end{pmatrix} | u_n^i \in \RR}.
\end{equation}
By a direct computation we obtain the following left and right invariant
vector fields

in matrix form:
\begin{equation}
L =
\begin{pmatrix}
    \exp(-u^2_n) & 0 & 0 & 0 \\[1mm]
    0 & 1 & 0 & 0 \\[1mm]
    0 & 0 & \exp(-u_n^4) & 0 \\[1mm]
    0 & 0 & 0 & 1
\end{pmatrix},
\qquad
R = 
\begin{pmatrix}
1 & 0 & 0 & 0 \\[1mm]
-u^1_n & 1 & 0 & 0 \\[1mm]
0 & 0 & 1 & 0 \\[1mm]
0 & 0 & -u_n^3 & 1
\end{pmatrix}.
\end{equation}

The on-site $h^{ij}_n$ and inter-site $g^{ij}_n$ elements of the dDGP bracket in \Cref{thm:Dubrovin_quasiF} are:
\begin{subequations}
\begin{align} 
h_n&=\left(\begin{array}{ cccc}
0 & -{c_1} ({\mathrm e}^{-u^2_{n}}+1) & -{c_2} ({\mathrm e}^{-(u^2_{n}+u^4_{n})} +1) & {c_2} u^3_{n} 
\\[1mm]
 {c_1} ({\mathrm e}^{-u^2_{n}}+1) & 0 & {c_2} u^1_{n} & -{c_2} u^1_{n} u^3_{n} 
\\[1mm]
 {c_2} ({\mathrm e}^{-(u^2_{n}+u^4_{n})}+1) & -{c_2} u^1_{n} & 0 & -{c_3} ({\mathrm e}^{-u^4_{n}}+1) 
\\[1mm]
 -{c_2} u^3_{n} & {c_2} u^1_{n} u^3_{n} & {c_3} ({\mathrm e}^{-u^4_{n}}+1) & 0 
\end{array}\right), \\[2mm]
g_n&=\left(\begin{array}{ cccc}
0 & {\mathrm e}^{-u^2_{n}} {c_1}  & {\mathrm e}^{-u^2_{n}} {c_2}  & -{\mathrm e}^{-u^2_{n}} u^3_{n+1} {c_2}  
\\[1mm]
 -{c_1}  & {c_1} u^1_{n+1} & 0 & 0 
\\[1mm]
 -{\mathrm e}^{-u^4_{n}} {c_2}  & {\mathrm e}^{-u^4_{n}} u^1_{n+1} {c_2}  & 0 & {\mathrm e}^{-u^4_{n}} {c_3}  
\\[1mm]
 0 & 0 & -{c_3}  & {c_3} u^3_{n+1} 
\end{array}\right).
\end{align}
\end{subequations}

\subsubsection{$\mathfrak{n}_{3,1} \oplus \mathfrak{n}_{1,1}$}

Let us consider the Lie algebra
$\mathfrak{n}_{3,1} \oplus \mathfrak{n}_{1,1}$
whose only non-zero commutation relation is:
\begin{equation}
    \comm{e_2}{e_3} = e_1.
\end{equation}
From a direct computation it is possible to show that the following
skew-symmetric matrix:
\begin{equation}
r_{\mathfrak{n}_{3,1} \oplus \mathfrak{n}_{1,1}} =
\begin{pmatrix}
0 & c_{1} & c_{2} & c_{3} \\
-c_{1} & 0 & 0 & c_{4} \\
-c_{2} & 0 & 0 & c_{5} \\
-c_{3} & -c_{4} & -c_{5} & 0
\end{pmatrix}
\end{equation}
is a solution of the CYBE~\eqref{eq:CYBE}.

\begin{remark}
    We recall that the Lie algebra $\mathfrak{n}_{3,1}$ is often
    called the Heisenberg algebra and denoted as $\mathfrak{h}(1)$,
    see~\cite[Chap.\ 11]{SnobWinternitz2017book}.
\end{remark}

Next, we consider a representation in the space of $3\times 3$
matrices
$\rho\colon \mathfrak{n}_{3,1} \oplus \mathfrak{n}_{1,1}\longrightarrow\gl(3,\mathbb{R})$ acting as follows on the elements of the basis:
\begin{equation}
    \begin{gathered}
\rho(e_1)=
\begin{pmatrix}
0&0&1\\0&0&0\\0&0&0
\end{pmatrix},\quad
\rho(e_2)=
\begin{pmatrix}
0&1&0\\0&0&0\\0&0&0
\end{pmatrix},
\\
\rho(e_3)=
\begin{pmatrix}
0&0&0\\0&0&1\\0&0&0
\end{pmatrix},\quad
\rho(e_4)=
\begin{pmatrix}
1&0&0\\0&1&0\\0&0&1
\end{pmatrix}.            
    \end{gathered}
\end{equation}
The representation $\rho$ is clearly faithful, obtained from the faithful 
representation of the Heisenberg algebra which is constructed applying 
\Cref{thm:repr}, see also~\cite[\S 3.1]{Ghanam_etal2005}, and adding 
the identity matrix ($\rho(e_4)$). This gives the following 
(local) parametrisation of the group $N_{3,1}\times \RR$:
\begin{equation}
    N_{3,1}\times \RR
    =
    \Set{
    \exp(u^4_n)\begin{pmatrix}
        1 & u^2_n & u^1_n
        \\[2mm]
        0 & 1 & u^3_n
        \\[2mm]
        0 & 0 & 1
\end{pmatrix} | u_n^i \in \RR}.
\end{equation}
By a direct computation we obtain the following left and right invariant
vector fields
 
in matrix form:
\begin{equation}
L =
\begin{pmatrix}
1 & u^3_n & 0 & 0 \\
0 & 1 & 0 & 0 \\
0 & 0 & 1 & 0 \\
0 & 0 & 0 & 1
\end{pmatrix},
\qquad
R =
\begin{pmatrix}
1 & 0 & 0 & 0 \\
0 & 1 & 0 & 0 \\
u^2_n & 0 & 1 & 0 \\
0 & 0 & 0 & 1
\end{pmatrix}.
\end{equation}
The on-site $h^{ij}_n$ and inter-site $g^{ij}_n$ elements of the dDGP bracket in \Cref{thm:Dubrovin_quasiF} are: 
\begin{subequations}
    \begin{align}
    h_n & = \left(
\begin{array}{cccc}
 0 & -2 c_{1} & -2 c_{2} & -2 c_{3}-c_{4} u_{n}^{3} \\[1mm]
 2 c_{1} & 0 & c_{1} u_{n}^{2} & -2 c_{4} \\[1mm]
 2 c_{2} & -c_{1} u_{n}^{2} & 0 & -c_{3} u_{n}^{2}-2 c_{5} \\[1mm]
 2 c_{3}+c_{4} u_{n}^{3} & 2 c_{4} & c_{3} u_{n}^{2}+2 c_{5} & 0 \\[1mm]
\end{array}
\right),
    \\[2mm]
    \setlength\arraycolsep{-5pt}
    g_n & = \left(
\begin{array}{cccc}
 -c_{1} u_{n}^{3} & c_{1} & c_{2}-c_{1} u_{n}^{3} u_{n+1}^{2} & c_{3}+c_{4} u_{n}^{3} \\[1mm]
 -c_{1} & 0 & -c_{1} u_{n+1}^{2} & c_{4} \\[1mm]
 -c_{2} & 0 & -c_{2} u_{n+1}^{2} & c_{5} \\[1mm]
 -c_{3} & -c_{4} & -c_{3} u_{n+1}^{2}-c_{5} & 0 \\[1mm]
\end{array}
\right).
    \end{align}
\end{subequations}

\subsubsection{$\mathfrak{n}_{4,1}$}

Let us consider the Lie algebra
$\mathfrak{n}_{4,1}$
whose non-zero commutation relations are:
\begin{equation}
    \comm{e_2}{e_4} = e_1, \quad \comm{e_3}{e_4} = e_2.
\end{equation}

From a direct computation it is possible to show that the following
skew-symmetric matrix:
\begin{equation}
r_{\mathfrak{n}_{4,1}} =
\begin{pmatrix}
    0 & c_1 & c_2 & c_3 \\
    -c_1 & 0 & c_4 & 0 \\
    -c_2 & -c_4 & 0 & 0 \\
    -c_3 & 0 & 0 & 0
\end{pmatrix}
\end{equation}
is a solution of the CYBE~\eqref{eq:CYBE}.

Next, since $Z(\mathfrak{n}_{4,1}) = \langle e_1\rangle$ we can use \Cref{thm:repr}
to build a faithful representation in the space of $4\times 4$
matrices
$\rho\colon \mathfrak{n}_{4,1}\longrightarrow\gl(4,\mathbb{R})$:
\begin{equation}
    \begin{gathered}
\rho(e_1)=
\begin{pmatrix}
0&0&0&1\\0&0&0&0\\0&0&0&0\\0&0&0&0
\end{pmatrix},\quad
\rho(e_2)=
\begin{pmatrix}
0&0&0&0\\0&0&0&1\\0&0&0&0\\0&0&0&0
\end{pmatrix},
\\
\rho(e_3)=
\begin{pmatrix}
0&0&0&0\\0&0&0&0\\0&0&0&1\\0&0&0&0
\end{pmatrix},\quad
\rho(e_4)=
\begin{pmatrix}
0&-1&0&0\\0&0&-1&0\\0&0&0&0\\0&0&0&0
\end{pmatrix}.            
    \end{gathered}
\end{equation}
This gives the following (local) parametrisation of the group~$N_{4,1}$:
\begin{equation}
    N_{4,1}
    =
    \Set{
\begin{pmatrix}
1 & -u^4_n & \frac{1}{2}(u^4_n)^2 & u^1_n-u^4_n u^2_n+\frac{1}{2}(u^4_n)^2 u^3_n \\[1mm]
0 & 1 & -u^4_n & u^2_n-u^4_n u^3_n \\[1mm]
0 & 0 & 1 & u^3_n \\[1mm]
0 & 0 & 0 & 1
\end{pmatrix} | u_n^i \in \RR}.
\end{equation}
By a direct computation we obtain the following left and right invariant
vector fields

in matrix form:
\begin{equation}
L =
\begin{pmatrix}
1 & u^4_n & \frac{1}{2}(u^4_n)^2 & 0 \\[1mm]
0 & 1 & u^4_n & 0 \\[1mm]
0 & 0 & 1 & 0 \\[1mm]
0 & 0 & 0 & 1
\end{pmatrix},
\qquad
R = 
\begin{pmatrix}
1 & 0 & 0 & 0 \\[1mm]
0 & 1 & 0 & 0 \\[1mm]
0 & 0 & 1 & 0 \\[1mm]
u^2_n & u^3_n & 0 & 1
\end{pmatrix}.
\end{equation}
The on-site $h^{ij}_n$ and inter-site $g^{ij}_n$ elements of the dDGP bracket in \Cref{thm:Dubrovin_quasiF} are: 
\begin{subequations}
\begin{align}
h_n&= \left(
\begin{array}{cccc}
 0 & * & * & * \\[1mm]
 2 c_{1}+c_{2} u_{n}^{4}+\frac{1}{2}c_{4} (u_{n}^{4})^2 & 0 & * & * \\[1mm]
 2 c_{2}+c_{4} u_{n}^{4} & 2 c_{4} & 0 & * \\[1mm]
 c_{1} u_{n}^{3}+2 c_{3} & -c_{1} u_{n}^{2} & -c_{2} u_{n}^{2}-c_{4} u_{n}^{3} & 0 \\[1mm]
\end{array}
\right),   \\[2mm]
g_n&= \left(
\begin{array}{cccc}
 -\frac{1}{2} u_{n}^{4} (2 c_{1}+c_{2} u_{n}^{4}) & c_{1}-\frac{c_{4} (u_{n}^{4})^2}{2} & c_{2}+c_{4} u_{n}^{4} & f(u_n^4,u^2_{n+1},u^3_{n+1}) \\[1mm]
 -c_{1}-c_{2} u_{n}^{4} & -c_{4} u_{n}^{4} & c_{4} 
 & -u^2_{n+1}(c_1+c_2 u_{n}^{4}) - c_4 u_{n}^{4} u_{n+1}^{3}
 \\[1mm]
 -c_{2} & -c_{4} & 0 & -c_{2} u_{n+1}^{2}-c_{4} u_{n+1}^{3} \\[1mm]
 -c_{3} & 0 & 0 & -c_{3} u_{n+1}^{2} \\[1mm]
\end{array}
\right),    
\end{align} 
\end{subequations}
with $f(u_n^4,u^2_{n+1},u^3_{n+1}) = c_{1} (u_{n+1}^{3}-u_{n}^{4} u_{n+1}^{2})-\frac{1}{2} (u_{n}^{4})^2 (c_{2} u_{n+1}^{2}+c_{4} u_{n+1}^{3})+c_{3}$.

\subsubsection{$\mathfrak{s}_{4,1}$}
Let us consider the Lie algebra
$\mathfrak{s}_{4,1}$
whose non-zero commutation relations are:
\begin{equation}
        \comm{e_4}{e_2} = e_1, \quad \comm{e_4}{e_3} = e_3.
\end{equation}

From a direct computation it is possible to show that the following
skew-symmetric matrix:
\begin{equation}
r_{\mathfrak{s}_{4,1}} =
\begin{pmatrix}
    0 & c_1 & c_2 & 0 \\
    -c_1 & 0 & c_3 & 0 \\
    -c_2 & -c_3 & 0 & c_4 \\
    0 & 0 & -c_4 & 0
\end{pmatrix}
\end{equation}
is a solution of the CYBE~\eqref{eq:CYBE}.

Next, since $Z(\mathfrak{s}_{4,1})=\langle e_1\rangle$, we can build a 
faithful representation in the space of $4\times 4$ matrices
$\rho\colon \mathfrak{s}_{4,1}\longrightarrow\gl(4,\mathbb{R})$ using
\Cref{thm:repr}. The obtained representation $\rho$ acts 
as follows on the elements of the basis:
\begin{equation}
    \begin{gathered}
\rho(e_1)=
\begin{pmatrix}
0&0&0&1\\0&0&0&0\\0&0&0&0\\0&0&0&0
\end{pmatrix},\quad
\rho(e_2)=
\begin{pmatrix}
0&0&0&0\\0&0&0&1\\0&0&0&0\\0&0&0&0
\end{pmatrix}, 
\\
\rho(e_3)=
\begin{pmatrix}
0&0&0&0\\0&0&0&0\\0&0&0&1\\0&0&0&0
\end{pmatrix},\quad
\rho(e_4)=
\begin{pmatrix}
0&1&0&0\\0&0&0&0\\0&0&1&0\\0&0&0&0
\end{pmatrix}.            
    \end{gathered}
\end{equation}
This gives the following 
(local) parametrisation of the group $S_{4,1}$:
\begin{equation}
    S_{4,1}
    =
    \Set{
\begin{pmatrix}
1 & u^4_n & 0 & u^2_n u^4_n + u^1_n \\[2mm]
0 & 1 & 0 & u^2_n \\[2mm]
0 & 0 & \exp(u^4_n) & \exp(u^4_n) u^3_n \\[2mm]
0 & 0 & 0 & 1
\end{pmatrix}
 | u_n^i \in \RR}.
\end{equation}
By a direct computation we obtain the following left and right invariant
vector fields

in matrix form:
\begin{equation}
L =
\begin{pmatrix}
1 & -u^4_n & 0 & 0 \\[1mm]
0 & 1 & 0 & 0 \\[1mm]
0 & 0 & \exp(-u^4_n) & 0 \\[1mm]
0 & 0 & 0 & 1
\end{pmatrix}
\qquad
R = 
\begin{pmatrix}
1 & 0 & 0 & 0 \\[1mm]
0 & 1 & 0 & 0 \\[1mm]
0 & 0 & 1 & 0 \\[1mm]
-u^2_n & 0 & -u^3_n & 1
\end{pmatrix}.
\end{equation}

The on-site $h^{ij}_n$ and inter-site $g^{ij}_n$ elements of the dDGP bracket in \Cref{thm:Dubrovin_quasiF} are: 
\begin{subequations} 
\begin{align}
h_n&=\left(
\begin{array}{cccc}
 0 & * & * & * \\[1mm]
 2 c_{1} & 0 & * & * \\[1mm]
 \text{e}^{-u^{4}_{n}} \bigl(c_{2} \text{e}^{u^{4}_{n}}+c_{2}-c_{3} u^{4}_{n}\bigr) & c_{3} \bigl(\text{e}^{-u^{4}_{n}}+1\bigr) & 0 & * \\[1mm]
 -c_{2} u^{3}_{n} & c_{1} u^{2}_{n}-c_{3} u^{3}_{n} & c_{2} u^{2}_{n}+c_{4} \text{e}^{-u^{4}_{n}}+c_{4} & 0 \\[1mm]
\end{array}
\right), \\[2mm]
g_n&=\left(
\begin{array}{cccc}
 c_{1} u^{4}_{n} & c_{1} & c_{2}-c_{3} u^{4}_{n} & u^{3}_{n+1} (c_{3} u^{4}_{n}-c_{2})-c_{1} u^{2}_{n+1} u^{4}_{n} \\[1mm]
 -c_{1} & 0 & c_{3} & c_{1} u^{2}_{n+1}-c_{3} u^{3}_{n+1} \\[1mm]
 -c_{2} \text{e}^{-u^{4}_{n}} & -c_{3} \text{e}^{-u^{4}_{n}} & 0 & \text{e}^{-u^{4}_{n}} (c_{2} u^{2}_{n+1}+c_{4}) \\[1mm]
 0 & 0 & -c_{4} & c_{4} u^{3}_{n+1} \\[1mm]
\end{array}
\right).
\end{align} 
\end{subequations}

\subsubsection{$\mathfrak{s}_{4,8}(a)$}

Let us consider the Lie algebra
$\mathfrak{s}_{4,8}(a)$
whose non-zero commutation relations are:
\begin{equation}
    \comm{e_2}{e_3} = e_1, \quad \comm{e_4}{e_1} = (1+a)e_1,
    \comm{e_4}{e_2} = e_2, \quad \comm{e_4}{e_3} = a e_3.
\end{equation}

From a direct computation it is possible to show that the following
skew-symmetric matrix:
\begin{equation}
r_{\mathfrak{s}_{4,8}(a)} =
\begin{pmatrix}
    0 & c_1 & c_2 & -\dfrac{c_3}{1+a} \\[1mm]
    -c_1 & 0 & c_3 & 0 \\[1mm]
    -c_2 & -c_3 & 0 & 0 \\[1mm]
    \dfrac{c_3}{1+a} & 0 & 0 & 0
\end{pmatrix}
\end{equation}
is a solution of the CYBE~\eqref{eq:CYBE}.

Next, since $Z(\mathfrak{s}_{4,8}(a))=\Set{0}$, we consider as a representation 
the adjoint representation in the space of $4\times 4$
matrices
$\rho\colon \mathfrak{s}_{4,8}(a)\longrightarrow\gl(4,\mathbb{R})$ acting as follows on the elements of the basis:
\begin{equation}
\begin{aligned}
\rho(e_1) &= \begin{pmatrix}
0 & 0 & 0 & -1-\alpha \\
0 & 0 & 0 & 0 \\
0 & 0 & 0 & 0 \\
0 & 0 & 0 & 0
\end{pmatrix}, &\quad 
\rho(e_2) &= \begin{pmatrix}
0 & 0 & 1 & 0 \\
0 & 0 & 0 & -1 \\
0 & 0 & 0 & 0 \\
0 & 0 & 0 & 0
\end{pmatrix}, \\[1mm]
\rho(e_3) &= \begin{pmatrix}
0 & -1 & 0 & 0 \\
0 & 0 & 0 & 0 \\
0 & 0 & 0 & -\alpha \\
0 & 0 & 0 & 0
\end{pmatrix}, &\quad 
\rho(e_4) &= \begin{pmatrix}
1+\alpha & 0 & 0 & 0 \\
0 & 1 & 0 & 0 \\
0 & 0 & \alpha & 0 \\
0 & 0 & 0 & 0
\end{pmatrix}.
\end{aligned}
\end{equation}
This gives the following 
(local) parametrisation of the group $S_{4,8}$:
\begin{equation}
    S_{4,8}
    =
    \Set{ g(u_n^1,u_n^3,u_n^3,u_n^4) 
 | u_n^i \in \RR},
\end{equation}
where:
\begin{equation} 
g=
\exp(u^4_n(1+\alpha))\begin{pmatrix}
 1& - u^3_n &  u^2_n & u^2_n u^3_n-(1+\alpha) u^1_n \\[2mm]
0 & \exp(-\alpha) & 0 & -\exp(-\alpha) u^2_n \\[2mm]
0 & 0 & \exp(-\alpha) & -\alpha \exp(-\alpha) u^3_n \\[2mm]
0 & 0 & 0 & \exp(-u^4_n(1+\alpha))
\end{pmatrix}.
\end{equation}
By a direct computation we obtain the following left and right invariant vector fields

in matrix form:
\begin{subequations}
    \begin{align}
L &=
\begin{pmatrix}
\exp(-u_n^4(1+\alpha)) & \exp(-u_n^4) u_n^3 & 0 & 0 \\[1mm]
0 & \exp(-u_n^4) & 0 & 0 \\[1mm]
0 & 0 & \exp(-\alpha u_n^4) & 0 \\[1mm]
0 & 0 & 0 & 1
\end{pmatrix},
\\[2mm]
R &=
\begin{pmatrix}
1 & 0 & 0 & 0 \\[1mm]
0 & 1 & 0 & 0 \\[1mm]
u_n^2 & 0 & 1 & 0 \\[1mm]
-(1+\alpha) u_n^1 & -u_n^2 & -\alpha u_n^3 & 1
\end{pmatrix}.        
    \end{align}
\end{subequations}
The on-site $h^{ij}_n$ and inter-site $g^{ij}_n$ elements of the dDGP bracket in~\Cref{thm:Dubrovin_quasiF} are not sufficiently compact to be displayed explicitly here.

\subsubsection{$\mathfrak{s}_{4,9}$}

Let us consider the Lie algebra
$\mathfrak{s}_{4,9}(\alpha)$
whose non-zero commutation relations are:
\begin{equation}
    \comm{e_2}{e_3} = e_1, \comm{e_4}{e_1} = 2\alpha e_1,
    \comm{e_4}{e_2} = \alpha e_2 - e_3, \comm{e_4}{e_3} = e_2 + \alpha e_3, 
\end{equation}

From a direct computation it is possible to show that the following
skew-symmetric matrix:
\begin{equation}
r_{\mathfrak{s}_{4,9}(\alpha)} =
\begin{pmatrix}
    0 & c_1 & c_2 & c_3 \\
    -c_1 & 0 & -2\alpha c_3 & 0 \\
    -c_2 & 2\alpha c_3 & 0 & 0 \\
    -c_3 & 0 & 0 & 0
\end{pmatrix}
\end{equation}
is a solution of the CYBE~\eqref{eq:CYBE}.

Next, since $Z(\mathfrak{s}_{4,9}(\alpha))=\Set{0}$, we consider as a representation 
the adjoint representation in the space of $4\times 4$
matrices
$\rho\colon \mathfrak{s}_{4,9}(\alpha)\longrightarrow\gl(4,\mathbb{R})$ acting as follows on the elements of the basis:
\begin{equation}
\begin{gathered}
\rho(e_1) = \begin{pmatrix}
0 & 0 & 0 & -2\alpha \\
0 & 0 & 0 & 0 \\
0 & 0 & 0 & 0 \\
0 & 0 & 0 & 0
\end{pmatrix}, \quad 
\rho(e_2) = \begin{pmatrix}
0 & 0 & 1 & 0 \\
0 & 0 & 0 & -\alpha \\
0 & 0 & 0 & 1 \\
0 & 0 & 0 & 0
\end{pmatrix}, \\[1mm]
\rho(e_3) = \begin{pmatrix}
0 & -1 & 0 & 0 \\
0 & 0 & 0 & -1 \\
0 & 0 & 0 & -\alpha \\
0 & 0 & 0 & 0
\end{pmatrix}, \quad 
\rho(e_4) = \begin{pmatrix}
2\alpha & 0 & 0 & 0 \\
0 & \alpha & 1 & 0 \\
0 & -1 & \alpha & 0 \\
0 & 0 & 0 & 0
\end{pmatrix}.
\end{gathered}
\end{equation}
This gives the following 
(local) parametrisation of the group $S_{4,9}$:
\begin{equation}
    S_{4,9}
    =
    \Set{ g_4(u_n^4)g_3(u_n^3)g_2(u_n^2)g_1(u_n^1)
 | u_n^i \in \RR},
\end{equation}
where:
\begin{equation}
\begin{split}
g_1 &=
\begin{pmatrix}
1 & 0 & 0 & -2\alpha u_n^1 \\[1mm]
0 & 1 & 0 & 0 \\[1mm]
0 & 0 & 1 & 0 \\[1mm]
0 & 0 & 0 & 1
\end{pmatrix}, ~~ 
g_2 =
\begin{pmatrix}
1 & 0 & u_n^2 & \frac{1}{2}(u_n^2)^2 \\[1mm]
0 & 1 & 0 & -\alpha u_n^2 \\[1mm]
0 & 0 & 1 & u_n^2 \\[1mm]
0 & 0 & 0 & 1
\end{pmatrix}, ~~ 
g_3 =
\begin{pmatrix}
1 & -u_n^3 & 0 & \frac{1}{2}(u_n^3)^2 \\[1mm]
0 & 1 & 0 & -u_n^3 \\[1mm]
0 & 0 & 1 & -\alpha u_n^3 \\[1mm]
0 & 0 & 0 & 1
\end{pmatrix}, \\[2mm]
g_4 &= 
\exp(\alpha u_n^4)
\begin{pmatrix}
\exp(\alpha u_n^4) & 0 & 0 & 0 \\[1mm]
0 & \cos(u_n^4) & \sin(u_n^4) & 0 \\[1mm]
0 & -\sin(u_n^4) & \cos(u_n^4) & 0 \\[1mm]
0 & 0 & 0 & \exp(-\alpha u_n^4)
\end{pmatrix}.
\end{split}
\end{equation}
Again by a direct computation we obtain the following left and right invariant
vector fields
 
in matrix form: 
\begin{subequations}
    \begin{align}
L &=
\exp(-\alpha u_n^4)
\begin{pmatrix}
\exp(-\alpha u_n^4) & \cos(u_n^4) u_n^3 & -\sin(u_n^4) u_n^3 & 0 \\[1mm]
0 & \cos(u_n^4) & -\sin(u_n^4) & 0 \\[1mm]
0 & \sin(u_n^4) & \cos(u_n^4) & 0 \\[1mm]
0 & 0 & 0 & \exp(\alpha u_n^4)
\end{pmatrix},
\\[1mm]
R &=
\begin{pmatrix}
1 & 0 & 0 & 0 \\[1mm]
0 & 1 & 0 & 0 \\[1mm]
u_n^2 & 0 & 1 & 0 \\[1mm]
\frac{1}{2}(u_n^2)^2-\frac{1}{2}(u_n^3)^2 -2\alpha u_n^1 & -\alpha u_n^2-u_n^3 & -\alpha u_n^3+u_n^2 & 1
\end{pmatrix}.        
    \end{align}
\end{subequations} 
The on-site $h^{ij}_n$ and inter-site $g^{ij}_n$ elements of the dDGP bracket in~\Cref{thm:Dubrovin_quasiF} are not sufficiently compact to be displayed explicitly here.

\subsubsection{$\mathfrak{s}_{4,10}$}

Let us consider the Lie algebra
$\mathfrak{s}_{4,10}$
whose non-zero commutation relations are:
\begin{equation}
            \comm{e_2}{e_3} = e_1, \quad \comm{e_4}{e_1} = 2 e_1, 
            \quad 
            \comm{e_4}{e_2} = e_2, \quad \comm{e_4}{e_3} = e_2 + e_3.
\end{equation}

From a direct computation it is possible to show that the following
skew-symmetric matrix:
\begin{equation}
r_{\mathfrak{s}_{4,10}} =
\begin{pmatrix}
    0 & c_1 & c_2 & -\frac{1}{2} c_3 \\[1mm]
    -c_1 & 0 & c_3 & 0 \\[1mm]
    -c_2 & -c_3 & 0 & 0 \\[1mm] 
    \frac{1}{2} c_3 & 0 & 0 & 0
\end{pmatrix}
\end{equation}
is a solution of the CYBE~\eqref{eq:CYBE}.

Next, since $Z(\mathfrak{s}_{4,10})=\Set{0}$, we consider as a representation 
the adjoint representation in the space of $4\times 4$
matrices
$\rho\colon \mathfrak{s}_{4,10}\longrightarrow\gl(4,\mathbb{R})$ acting as follows on the elements of the basis:
\begin{equation}
\begin{aligned}
\rho(e_1) &= \begin{pmatrix}
0 & 0 & 0 & -2 \\
0 & 0 & 0 & 0 \\
0 & 0 & 0 & 0 \\
0 & 0 & 0 & 0
\end{pmatrix}, &\quad 
\rho(e_2) &= \begin{pmatrix}
0 & 0 & 1 & 0 \\
0 & 0 & 0 & -1 \\
0 & 0 & 0 & 0 \\
0 & 0 & 0 & 0
\end{pmatrix}, \\[1mm]
\rho(e_3) &= \begin{pmatrix}
0 & -1 & 0 & 0 \\
0 & 0 & 0 & -1 \\
0 & 0 & 0 & -1 \\
0 & 0 & 0 & 0
\end{pmatrix}, &\quad 
\rho(e_4) &= \begin{pmatrix}
2 & 0 & 0 & 0 \\
0 & 1 & 1 & 0 \\
0 & 0 & 1 & 0 \\
0 & 0 & 0 & 0
\end{pmatrix}.
\end{aligned}
\end{equation}
This gives the following 
(local) parametrisation of the group $S_{4,9}(1)$:
\begin{equation}
    S_{4,10}
    =
    \Set{ g(u_n^1,u_n^2,u_n^3,u_n^4) 
 | u_n^i \in \RR},
\end{equation}
where:
\begin{equation}
g=
\exp(u_n^4)
\begin{pmatrix}
\exp(u_n^4) & -\exp(u_n^4) u_n^3 & \exp(u_n^4) u_n^2 & \exp(u_n^4)\left(-2u_n^1+u_n^2 u_n^3+\frac{1}{2}(u_n^3)^2\right) \\[2mm]
0 & 1 &  u_n^4 & -u_n^2-u_n^3-u_n^4 u_n^3 \\[2mm]
0 & 0 & 1 & - u_n^3 \\[2mm]
0 & 0 & 0 & \exp(-u_n^4)
\end{pmatrix}.
\end{equation}
By a direct computation we obtain the following left and right invariant
vector fields 

in matrix form:
\begin{subequations}
    \begin{align}
L &=
\exp(-u_n^4)
\begin{pmatrix}
\exp(-u_n^4) & u_n^3 & -u_n^4 u_n^3 & 0 \\[2mm]
0 & 1 & -u_n^4 & 0 \\[2mm]
0 & 0 & 1 & 0 \\[2mm]
0 & 0 & 0 & \exp(u_n^4)
\end{pmatrix},
\\[2mm]
R&=
\begin{pmatrix}
1 & 0 & 0 & 0 \\[1mm]
0 & 1 & 0 & 0 \\[1mm]
u_n^2 & 0 & 1 & 0 \\[1mm]
-2u_n^1-\frac{1}{2}(u_n^3)^2 & -u_n^2-u_n^3 & -u_n^3 & 1
\end{pmatrix}.        
    \end{align}
\end{subequations}
The on-site $h^{ij}_n$ and inter-site $g^{ij}_n$ elements of the dDGP bracket in~\Cref{thm:Dubrovin_quasiF} are not sufficiently compact to be displayed explicitly here.

\subsubsection{$\mathfrak{s}_{4,11}$}

Let us consider the Lie algebra
$\mathfrak{s}_{4,11}$
whose non-zero commutation relations are:
\begin{equation}
        \comm{e_2}{e_3} = e_1, 
        \quad 
        \comm{e_4}{e_1} = e_1,
        \quad
        \comm{e_4}{e_2} = e_2.
\end{equation}

From a direct computation it is possible to show that the following
skew-symmetric matrix:
\begin{equation}
r_{\mathfrak{s}_{4,11}} =
\begin{pmatrix}
    0 & c_1 & c_2 & -c_3 \\
    -c_1 & 0 & c_3 & 0 \\
    -c_2 & -c_3 & 0 & 0 \\
    c_3 & 0 & 0 & 0
\end{pmatrix}
\end{equation}
is a solution of the CYBE~\eqref{eq:CYBE}.

Next, since $Z(\mathfrak{s}_{4,11})=\Set{0}$, we consider as a representation 
the adjoint representation in the space of $4\times 4$
matrices
$\rho\colon \mathfrak{s}_{4,11}\longrightarrow\gl(4,\mathbb{R})$ acting as follows on the elements of the basis:
\begin{equation}
\begin{aligned}
\rho(e_1) &= \begin{pmatrix}
0 & 0 & 0 & -1 \\
0 & 0 & 0 & 0 \\
0 & 0 & 0 & 0 \\
0 & 0 & 0 & 0
\end{pmatrix}, &\quad 
\rho(e_2) &= \begin{pmatrix}
0 & 0 & 1 & 0 \\
0 & 0 & 0 & -1 \\
0 & 0 & 0 & 0 \\
0 & 0 & 0 & 0
\end{pmatrix}, \\[1mm]
\rho(e_3) &= \begin{pmatrix}
0 & -1 & 0 & 0 \\
0 & 0 & 0 & 0 \\
0 & 0 & 0 & 0 \\
0 & 0 & 0 & 0
\end{pmatrix}, &\quad 
\rho(e_4) &= \begin{pmatrix}
1 & 0 & 0 & 0 \\
0 & 1 & 0 & 0 \\
0 & 0 & 0 & 0 \\
0 & 0 & 0 & 0
\end{pmatrix}.
\end{aligned}
\end{equation}
This gives the following 
(local) parametrisation of the group $S_{4,11}$:
\begin{equation}
    S_{4,11}
    =
    \Set{ \exp(u_n^4)\begin{pmatrix}
1 & - u_n^3 &  u_n^2 & u_n^2 u_n^3-u_n^1 \\[1mm]
0 & 1 & 0 & - u_n^2 \\[1mm]
0 & 0 & \exp(-u_n^4) & 0 \\[1mm]
0 & 0 & 0 & \exp(-u_n^4)
\end{pmatrix}
 | u_n^i \in \RR}.
\end{equation}
By a direct computation we obtain the following left and right invariant
vector fields

in matrix form: 
\begin{align}
L =
\begin{pmatrix}
\exp(-u_n^4) & \exp(-u_n^4) u_n^3 & 0 & 0 \\[1mm]
0 & \exp(-u_n^4) & 0 & 0 \\[1mm]
0 & 0 & 1 & 0 \\[1mm]
0 & 0 & 0 & 1
\end{pmatrix}, \qquad %
R =
\begin{pmatrix}
1 & 0 & 0 & 0 \\[1mm]
0 & 1 & 0 & 0 \\[1mm]
u_n^2 & 0 & 1 & 0 \\[1mm]
-u_n^1 & -u_n^2 & 0 & 1
\end{pmatrix}.        
    \end{align}
The on-site $h^{ij}_n$ and inter-site $g^{ij}_n$ elements of the dDGP bracket in \Cref{thm:Dubrovin_quasiF} are: 
\begin{subequations} 
\begin{align}
h_n&= \left(
\begin{array}{cccc}
 0 & * & * & * \\[1mm]
 c_{1} (\text{e}^{-2 u^{4}_{n}}+1 ) & 0 & * & * \\[1mm]
 \text{e}^{-u^{4}_{n}} (c_{2}+c_{3} u^{3}_{n})+c_{2} & -c_{1} u^{2}_{n}+c_{3} \text{e}^{-u^{4}_{n}}+c_{3} & 0 & * \\[1mm]
 -c_{3} (\text{e}^{u^{4}_{n}}+1)-c_{1} u^{2}_{n} & c_{1} u^{1}_{n} & c_{2} u^{1}_{n}-c_{1} (u^{2}_{n})^2 & 0 \\[1mm]
\end{array}
\right)   , \\[2mm]
g_n&= \text{e}^{-u^{4}_{n}}\left(
\begin{array}{cccc}
 -c_{1} u^{3}_{n}  & c_{1}  &  u^{3}_{n} (c_{3}-c_{1} u^{2}_{n+1})+c_{2} & -c_{1} (u^{2}_{n+1}-u^{1}_{n+1} u^{3}_{n})-c_{3} \\[1mm]
 -c_{1}  & 0 &  (c_{3}-c_{1} u^{2}_{n+1}) & c_{1} u^{1}_{n+1}  \\[1mm]
 -c_{2} \text{e}^{u^{4}_{n}} & -c_{3} \text{e}^{u^{4}_{n}} & -c_{2} u^{2}_{n+1} \text{e}^{u^{4}_{n}} & c_{2} u^{1}_{n+1} \text{e}^{u^{4}_{n}}+c_{3} u^{2}_{n+1} \text{e}^{u^{4}_{n}} \\[1mm]
 c_{3}\text{e}^{u^{4}_{n}} & 0 & c_{3} u^{2}_{n+1}\text{e}^{u^{4}_{n}} & -c_{3} u^{1}_{n+1}\text{e}^{u^{4}_{n}} \\[1mm]
\end{array}
\right) .
\end{align}
\end{subequations}

\subsubsection{$\mathfrak{s}_{4,12}$}

Let us consider the Lie algebra
$\mathfrak{s}_{4,12}$
whose non-zero commutation relations are:
\begin{equation}
        \comm{e_3}{e_1} = e_1, 
        \quad 
        \comm{e_3}{e_2} = e_2,
        \quad
        \quad \comm{e_4}{e_1} = -e_2, 
        \quad 
        \comm{e_4}{e_2} = e_1.
\end{equation}

From a direct computation it is possible to show that the following
skew-symmetric matrix:
\begin{equation}
r_{\mathfrak{s}_{4,12}} =
\begin{pmatrix}
    0 & c_1 & c_3 & -c_2 \\
    -c_1 & 0 & c_2 & c_3 \\
    -c_3 & -c_2 & 0 & 0 \\
    c_2 & -c_3 & 0 & 0
\end{pmatrix}
\end{equation}
is a solution of the CYBE~\eqref{eq:CYBE}.

Next, since $Z(\mathfrak{s}_{4,12})=\Set{0}$, we consider as a representation 
the adjoint representation in the space of $4\times 4$
matrices
$\rho\colon \mathfrak{s}_{4,12}\longrightarrow\gl(4,\mathbb{R})$ acting as follows on the elements of the basis:{ 
\begin{eqnarray} 
&\rho(e_1) = \begin{pmatrix}
0 & 0 & -1 & 0 \\
0 & 0 & 0 & 1 \\
0 & 0 & 0 & 0 \\
0 & 0 & 0 & 0
\end{pmatrix},~
&\rho(e_2) = \begin{pmatrix}
0 & 0 & 0 & -1 \\
0 & 0 & -1 & 0 \\
0 & 0 & 0 & 0 \\
0 & 0 & 0 & 0
\end{pmatrix},~ \\[1mm]
&\rho(e_3) = \begin{pmatrix}
1 & 0 & 0 & 0 \\
0 & 1 & 0 & 0 \\
0 & 0 & 0 & 0 \\
0 & 0 & 0 & 0
\end{pmatrix},
&\rho(e_4) = \begin{pmatrix}
0 & 1 & 0 & 0 \\
-1 & 0 & 0 & 0 \\
0 & 0 & 0 & 0 \\
0 & 0 & 0 & 0
\end{pmatrix}.
\end{eqnarray}  }
This gives the following 
(local) parametrisation of the group $S_{4,12}$:
\begin{equation}
    S_{4,12}
    =
    \Set{ g_4(u_n^4)g_3(u_n^3)g_2(u_n^2)g_1(u_n^1)
 | u_n^i \in \RR},
\end{equation}
where:
\begin{equation}
    \begin{aligned}
g_1 &=
\begin{pmatrix}
1 & 0 & -u_n^1 & 0 \\[1mm]
0 & 1 & 0 & u_n^1 \\[1mm]
0 & 0 & 1 & 0 \\[1mm]
0 & 0 & 0 & 1
\end{pmatrix}, 
&
g_2 &=
\begin{pmatrix}
1 & 0 & 0 & -u_n^2 \\[1mm]
0 & 1 & -u_n^2 & 0 \\[1mm]
0 & 0 & 1 & 0 \\[1mm]
0 & 0 & 0 & 1
\end{pmatrix}, \\[1mm]
g_3 &=
\begin{pmatrix}
\exp(u_n^3) & 0 & 0 & 0 \\[1mm]
0 & \exp(u_n^3) & 0 & 0 \\[1mm]
0 & 0 & 1 & 0 \\[1mm]
0 & 0 & 0 & 1
\end{pmatrix}, 
&
g_4 &=
\begin{pmatrix}
\cos(u_n^4) & \sin(u_n^4) & 0 & 0 \\[1mm]
-\sin(u_n^4) & \cos(u_n^4) & 0 & 0 \\[1mm]
0 & 0 & 1 & 0 \\[1mm]
0 & 0 & 0 & 1
\end{pmatrix}.
\end{aligned}
\end{equation}
By a direct computation we obtain the following left and right invariant
vector fields
 
in matrix form:
\begin{subequations}
    \begin{align}
L &=
\begin{pmatrix}
\exp(-u_n^3)\cos(u_n^4) & -\exp(-u_n^3)\sin(u_n^4) & 0 & 0 \\[2mm]
\exp(-u_n^3)\sin(u_n^4) & \exp(-u_n^3)\cos(u_n^4) & 0 & 0 \\[2mm]
0 & 0 & 1 & 0 \\[2mm]
0 & 0 & 0 & 1
\end{pmatrix},
\quad 
R =
\begin{pmatrix}
1 & 0 & 0 & 0 \\[2mm]
0 & 1 & 0 & 0 \\[2mm]
-u_n^1 & -u_n^2 & 1 & 0 \\[2mm]
-u_n^2 & u_n^1 & 0 & 1
\end{pmatrix}.        
    \end{align}
\end{subequations}
The on-site $h^{ij}_n$ and inter-site $g^{ij}_n$ elements of the dDGP bracket in~\Cref{thm:Dubrovin_quasiF} are not sufficiently compact to be displayed explicitly here.

\medskip

\subsection{\texorpdfstring{$\mathbb{N}$-graded filiform Lie algebras}{fili}}
\label{sec:filiform}

In this subsection we construct the dDGP bracket associated with quasi-Frobenius 
\emph{$\mathbb{N}$-graded filiform Lie algebras}. We do so because quasi-Frobenius 
$\mathbb{N}$-graded filiform Lie algebras are completely classified, 
see~\cite{Millionschikov2004}. This classification is reported, with our notation, 
in \Cref{tab:sympngradfili}. Moreover, although these algebras do not exhaust all 
quasi-Frobenius filiform Lie algebras, it is known that they are the 
``building blocks'' of such structures: in a fixed even dimension $2k$, a general 
quasi-Frobenius filiform Lie algebra is constructed from them by 
deformation~\cite{Millionshchikov2006}.

\begin{table}[h]
\centering
\begin{tabular}{c c c}
   ~Dimension~ & ~~Algebra~~ & ~~Non-zero commutation relations~~
    \\[1mm]
    \midrule[1pt] \\[-3mm]
    4 & $\mathfrak{n}_{4,1}$ & $[e_2, e_4] = e_1,\; [e_3, e_4] = e_2$ 
    \\[2mm]
    \midrule \\[-3mm]
    \multirow{2}{*}{6} & $\mathfrak{n}_{6,1}$ & 
    $[e_1, e_i] = e_{i+1},\; i = 2, \ldots, 5$
    \\[2mm]
    & $\mathfrak{V}_6$ & $[e_i, e_j] = (j - i)e_{i+j},\; i + j \leq 6$ 
    \\[1mm]
    \midrule \\[-3mm]
    \multirow{2}{*}{8} & $\mathfrak{n}_{8,1}$ & 
    $[e_i, e_{2k}] = e_{i-1},\; i = 2, \ldots, 7$ 
    \\[2mm]
    & $\mathfrak{g}_{8}(\alpha)$ & \eqref{eq:commg8a} 
    \\[1mm]
    \midrule \\[-3mm]
    \multirow{2}{*}{10} & $\mathfrak{n}_{10,1}$ & 
    $[e_1, e_i] = e_{i+1},\; i = 2, \ldots, 9$
    \\[2mm]
    & $\mathfrak{g}_{10}(\alpha)$ &
    \eqref{eq:commg10a}
    \\[1mm]
    \midrule \\[-3mm]
    \multirow{2}{*}{$2k \geq 12$} & $\mathfrak{n}_{2k,1}$ & 
    $[e_i, e_{2k}] = e_{i-1},\; i = 2, \ldots, 2k - 1$ 
    \\[2mm]
    & $\mathfrak{V}_{2k}$ & 
    $[e_i, e_j] = (j - i)e_{i+j},\; i + j \leq 2k$ 
    \\[1mm]
    \bottomrule
\end{tabular}
    \caption{$\mathbb{N}$-graded filiform quasi-Frobenius Lie algebras}
    \label{tab:sympngradfili}
\end{table}

To give some context we will recall now the main definitions of this theory,
while for a complete description of such structures we refer to the aforementioned
works~\cite{Millionschikov2004,Millionshchikov2006}.
Intuitively speaking a filiform Lie algebra
is a nilpotent Lie algebra whose lower central series is as long as possible 
for its dimension. To be more precise, given an $d$-dimensional nilpotent 
Lie algebra $\mathfrak{g}$, its lower central series is the descending
series of ideals:
\begin{equation}
   \mathfrak{g}^{1} \supset \mathfrak{g}^{2} \supset \cdots \supset \mathfrak{g}^{k} \supset \cdots 
\end{equation}
defined recursively by:
\begin{equation}
    \mathfrak{g}^{k}
    =
    \begin{cases}
        \mathfrak{g} & k=1,
        \\
        \comm{\mathfrak{g}^{k-1}}{\mathfrak{g}}, & k>1.
    \end{cases}
\end{equation}
The Lie algebra $\mathfrak{g}$ is \emph{a ($k$-step) nilpotent Lie algebra} if there exists a $k$ such 
that $\mathfrak{g}^{k-1}\neq \Set{0}$ and $\mathfrak{g}^{k}=\Set{0}$. In particular,
the Lie algebra $\mathfrak{g}$ is called \emph{filiform} if $k=d$, i.e.\ a filiform
$d$-dimensional Lie algebra is a $d$-step nilpotent Lie algebra.

Moreover, a Lie algebra $\mathfrak{g}$ is said to be \emph{$\mathbb{N}$-graded}
if it decomposes as a direct sum of subspaces indexed by the natural numbers,
\begin{equation}
    \mathfrak{g} = \bigoplus_{i \in \mathbb{N}} \mathfrak{g}_i,    
\end{equation}
such that the Lie bracket respects the grading:
\begin{equation}
    [\mathfrak{g}_i, \mathfrak{g}_j] \subseteq \mathfrak{g}_{i+j}.    
\end{equation}
A \emph{filiform $\mathbb{N}$-graded Lie algebra} is a filiform Lie algebra 
that admits such a grading. In the finite-dimensional case, the grading typically 
takes the form:
\begin{equation}
    \mathfrak{g} = \bigoplus_{i=1}^{n} \mathfrak{g}_i,    
\end{equation}
with $\dim \mathfrak{g}_i = 1$ for all $i \le n$ and $\mathfrak{g}_i = \Set{0}$ for $i > n$.
The brackets are then determined by
\begin{equation}   
    [\mathfrak{g}_1, \mathfrak{g}_\alpha] = \mathfrak{g}_{\alpha+1}, \quad \alpha \geq 2.    
\end{equation}
These algebras are automatically nilpotent. 

Throughout this subsection we will make use of some results about
nilpotent Lie algebras. The first one is about their representation theory:

\begin{theorem}[\cite{Burde_etal2009c}]
    Let $\mathfrak{g}$ be a $k$-step nilpotent Lie algebra. 
    Define a weight function on elements in by the lower central series of 
    $\mathfrak{g}$. Then a faithful representation of $\mathfrak{g}$ is 
    $U\mathfrak{g}/(U\mathfrak{g})^{k+1}$, where $(U\mathfrak{g})^{k+1}$ is the 
    (two-sided) ideal of generated by all monomials of weight at least $k$.
    Moreover, expanding the ideal $I$ keeping the property that 
    $I\cap Z(\mathfrak{g}) = \emptyset$, the resulting quotient remains 
    faithful and is a \emph{minimal faithful representation of $\mathfrak{g}$} 
    in the sense that it has no faithful submodules or quotients.
    \label{thm:nilpotrepr}
\end{theorem}

\begin{remark}
    We remark that the condition of minimality obtained from~\Cref{thm:nilpotrepr}
    does not guarantee that obtain representation is the smallest dimensional 
    faithful representation of $\mathfrak{g}$.
\end{remark}

The second one is about the construction of left and right invariant vector
fields on the associated Lie group $G=\exp(\mathfrak{g})$, and is taken 
from~\cite{Magazev_etal2015} and adapted to our notations:

\begin{proposition}
    Let us assume we are given an $n$-dimension Lie algebra $\mathfrak{g}$
    over the field $\mathbb{K}$. Define 
    $M_i=\exp(-x_i \ad(e_i))$ and the matrices:
    \begin{equation}
        F^{(k)} = M_{1}\cdots M_{k-1}, \quad k=2,\ldots,n+1. 
    \end{equation}
    Then the matrix $\Omega = (\Omega_{i,k})_{i,k=1}^{n}$ whose elements are:
    \begin{equation}
        \Omega_{i,k}
        =
        \begin{cases}
            \delta_{1,k} & i=1,
            \\
            (F^{(k})_{i,k} & i>1,
        \end{cases}
    \end{equation}
    is such that the its columns of its inverse are the \emph{right  invariant vector
    fields} for the Lie group $G=\exp(\mathfrak{g})$. Moreover, the matrix:
    \begin{equation}
        \Sigma = - (F^{(n+1)})^{-1}\Omega,
    \end{equation}
    is such that the columns of its inverse are the \emph{left invariant vector
    fields} for the Lie group $G=\exp(\mathfrak{g})$.
    \label{prop:ivf}    
\end{proposition}

The construction of left and right invariant vector fields presented in~\Cref{prop:ivf} 
is particularly easy in the case of nilpotent Lie algebras because by 
Engel's theorem~\cite[\S II.3]{Jacobson1962} the matrices of the adjoint action 
are nilpotent, hence their exponential are reduced to finite sums.

\subsubsection{$\mathfrak{g}_{8}(\alpha)$}
\label{subsec:g8}

Let us consider the Lie algebra
$\mathfrak{g}_{8}(\alpha)$
whose non-zero commutation relations are:
\begin{equation}
    \label{eq:commg8a}
    \begin{array}{llll}
        [e_1, e_2] = e_3, & [e_1, e_3] = e_4, & 
        [e_1, e_4] = e_5, & [e_1, e_5] = e_6,
            \\[1mm]
            [e_1, e_6] = e_7, & [e_1, e_7] = e_8, 
            &
            [e_2, e_3] = (2+\alpha)e_5, & [e_2, e_4] = (2+\alpha)e_6,
            \\[1mm]
            [e_2, e_5] = (1+\alpha)e_7, 
            &
        [e_2, e_6] = \alpha e_8, & [e_3, e_4] = e_7, & [e_3, e_5] = e_8,
    \end{array}
\end{equation}
where $\alpha\neq -5/2,-2,-1/2,1/2$.

From a direct computation it is possible to show that such an
algebra admits following pair of skew-symmetric solutions of
the CYBE~\eqref{eq:CYBE}:
\begin{subequations}
    \begin{align}
        r_{\mathfrak{g}_{8}(\alpha)}^{(1)}
        &=\begin{pmatrix}
            0 & 0 & 0 & 0 & 0 & 0 & 0 &
            \dfrac{f_1c_1}{2\alpha+5}
            \\[.5em]
            * & 0 & 0 & 0 & 0 & 0 & c_1 & 0 
            \\[.5em]
            * & * & 0 & 0 & 0 &
            \dfrac{f_1c_1}{2\alpha+2} & 0 & c_2 
            \\[.5em]
            * & * & * & 0 &
            \dfrac{f_1c_1}{3} & 0 & f_2
& c_4 \\[.5em]
* & * & * & * & 0 & c_3 &
\dfrac{(-2\alpha-5)c_4}{2\alpha-1} & c_5 \\[.5em]
* & * & * & * & * & 0 & c_6 & c_7 \\[.5em]
* & * & * & * & * & * & 0 & c_8 \\[.5em]
* & * & * & * & * & * & * & 0
\end{pmatrix},
        \\
        r_{\mathfrak{g}_{8}(\alpha)}^{(2)} &=
        \begin{pmatrix}
0 & 0 & 0 & 0 & 0 & 0 & 0 & c_9 \\[.5em]
* & 0 & 0 & 0 & 0 & 0 &
\dfrac{(2\alpha+5)c_9}{(2+\alpha)(2\alpha-1)} & c_{10} \\[.75em]
* & * & 0 & 0 & 0 &
\dfrac{(2\alpha+5)c_9}{2\alpha+2} &
\dfrac{\alpha c_{10}(2\alpha+5)}{2\alpha+2} &
f_3 \\[.75em]
* & * & * & 0 &
\dfrac{(2\alpha+5)c_9}{3} &
f_3 c_{10} &
f_5
& c_4 \\[.5em]
* & * & * & * & 0 & c_3 & c_{11} & c_5 \\[.5em]
* & * & * & * & * & 0 & c_6 & c_7 \\[.5em]
* & * & * & * & * & * & 0 & c_8 \\[.5em]
* & * & * & * & * & * & * & 0
\end{pmatrix}
    \end{align}
\end{subequations}
where:
    \begin{align*}
        f_1 &= 2\alpha^2+3\alpha-2,
        &
        f_2 &= \dfrac{4\left(\left(\alpha^2+\frac{9}{2}\alpha+5\right)c_2+\frac{3}{2}c_3\right)(1+\alpha)}{6\alpha^2+9\alpha-6}, \\
        f_3 &= \dfrac{N_{3}}{(4\alpha^2+9\alpha+2)(2\alpha+5)c_9 c_{10}}
        &
        f_4 &= \dfrac{(4\alpha^3+16\alpha^2+11\alpha-10)}{6\alpha+6} \\
        f_5 &= \dfrac{N_{5}}{6 c_9 (1+\alpha) c_{10} (4\alpha^2+9\alpha+2)(2\alpha-1)}, &
        \end{align*}
and the numerators $N_3$, $N_5$:
\begin{align*}        
        N_{3} &
        \begin{aligned}[t]
        &=
        4\alpha^5 c_{10}^3 + 24\alpha^4 c_{10}^3 + 43\alpha^3 c_{10}^3
        + (12 c_{10}^3 -6 c_9 c_{10} c_3)\alpha^2
        \\
        &+ (6 c_9^2(c_4 + c_{11})-20 c_{10}^3 - 3 c_9 c_{10} c_3)\alpha
        - 6 c_9 c_{10} c_3 + 15 c_9^2\bigl(c_4 - \frac{1}{5}c_{11}\bigr),    
        \end{aligned}
        \\
        N_{5} &
        \begin{aligned}[t]
            &=
            16\alpha^7 c_{10}^3 + 80\alpha^6 c_{10}^3 + 80\alpha^5 c_{10}^3
            - (24 c_9 c_{10} c_3 + 100 c_{10}^3)\alpha^4
            \\
            &+ \bigl[(24 c_4 + 24 c_{11})c_9^2 - 12 c_9 c_{10} c_3 - 85 c_{10}^3\bigr]\alpha^3
            + \bigl[(108 c_4 + 36 c_{11})c_9^2 + 36 c_9 c_{10} c_3 + 92 c_{10}^3\bigr]\alpha^2
            \\
            &+ (144 c_9^2 c_4 + 12 c_9 c_{10} c_3 - 20 c_{10}^3)\alpha
            + (60 c_4 - 12 c_{11})c_9^2 - 12 c_9 c_{10} c_3
        \end{aligned}
    \end{align*}

Next, since $Z(\mathfrak{g}_{8}(\alpha))=\langle e_8\rangle$, but the maximal
abelian ideal is $\mathfrak{I}_5 = \langle e_4,e_5,e_6,e_7,e_8\rangle$,
we cannot build a faithful representation using~\Cref{thm:repr}. On the other
hand, we can use~\Cref{thm:nilpotrepr} and its implementation in 
\texttt{SageMath}~\cite{sagemath}. This yields a faithful representation 
in the space of $22\times 22$ matrices
$\rho\colon \mathfrak{g}_{8}(\alpha)\longrightarrow\gl(22,\mathbb{R})$,
whose explicit form we omit for sake of brevity. Such a representation
also gives a (local) parametrisation of the group $G_{8}(\alpha)$ which
for very the same reason we omit.
By a direct computation we obtain the left and right invariant vector fields are explicitly given in Appendix~\ref{app:left_right_vectors_long}.

\subsubsection{$\mathfrak{g}_{10}(\alpha)$}
\label{subsec:g10}
Let us consider the Lie algebra
$\mathfrak{g}_{10}(\alpha)$
whose non-zero commutation relations are:
\begin{equation}
    \label{eq:commg10a}
    \begin{array}{lll}
        [e_1, e_2] = e_3, & 
        [e_1, e_3] = e_4, &
        [e_1, e_4] = e_5,
        \\[2mm]
        [e_1, e_5] = e_6, &
        [e_1, e_6] = e_7, &
        [e_1, e_7] = e_8,
        \\[2mm]
        [e_1, e_8] = e_9, &
        [e_1, e_9] = e_{10}, &
        \\[2mm]
        [e_2, e_3] = (2 + \alpha)e_5, &
        [e_2, e_4] = (2 + \alpha)e_6, &
        [e_2, e_5] = (1 + \alpha)e_7,
        \\[2mm]
        [e_2, e_6] = \alpha e_8, &
        [e_2, e_7] = \dfrac{2\alpha^2 + 3\alpha - 2}{2\alpha + 5} e_9, &
        [e_2, e_8] = \dfrac{2\alpha^2 + \alpha - 1}{2\alpha + 5} e_{10},
        \\[3mm]
        [e_3, e_4] = e_7, &
        [e_3, e_5] = e_8, &
        [e_3, e_6] = \dfrac{2\alpha + 2}{2\alpha + 5} e_9, 
        \\[3mm]
        [e_3, e_7] = \dfrac{2\alpha - 1}{2\alpha + 5} e_{10}, &
        [e_4, e_5] = \dfrac{3}{2\alpha + 5} e_9, &
        [e_4, e_6] = \dfrac{3}{2\alpha + 5} e_{10},
    \end{array}
\end{equation}
where $\alpha\neq -5/2,-2,-1/2,1/2$ and it is not a solution of the
following polynomial equations:
\begin{equation}
    2\alpha^3+2\alpha^2+3 =0,
    \qquad
    4\alpha^3+8\alpha^2-8\alpha-21 =0.
\end{equation}

\begin{remark}
    We observe that, despite many of the commutation relations
    of the Lie algebra $\mathfrak{g}_{10}(\alpha)$~\eqref{eq:commg10a} are the same 
    as the Lie algebra $\mathfrak{g}_{8}(\alpha)$~\eqref{eq:commg8a},
    $\mathfrak{g}_{8}(\alpha)$~\eqref{eq:commg10a} itself is not
    a Lie subalgebra of $\mathfrak{g}_{10}(\alpha)$.
\end{remark}

From a direct computation it is possible to show that such an
algebra admits five non-degenerate skew-symmetric solutions of
the CYBE~\eqref{eq:CYBE}. These solutions are rather cumbersome,
so we omit their explicit form.

Next, since $Z(\mathfrak{g}_{10}(\alpha))=\langle e_{10}\rangle$, but the maximal
abelian ideal is $\mathfrak{I}_6 = \langle e_5,e_6,e_7,e_8,e_9,e_{10}\rangle$,
we cannot build a faithful representation using~\Cref{thm:repr}. On the other
hand, we can use~\Cref{thm:nilpotrepr} and its implementation in 
\texttt{SageMath}~\cite{sagemath}. This yields a faithful representation 
in the space of $48\times 48$ matrices
$\rho\colon \mathfrak{g}_{10}(\alpha)\longrightarrow\gl(48,\mathbb{R})$,
whose explicit form we omit for sake of brevity. Such a representation
also gives a (local) parametrisation of the group $G_{10}(\alpha)$ which
for very the same reason we omit.
Again by a direct computation we obtain the following left invariant
vector fields, which are explicitly given in Appendix~\ref{app:left_right_vectors_long}.

\subsubsection{The family $\mathfrak{n}_{2k,1}$}

Consider the Lie algebra $\mathfrak{n}_{2k,1}=\langle e_1, e_2, \dots, e_{2k}\rangle$, 
whose only non-zero commutation relations are:
\begin{equation}
    \comm{e_i}{e_{2k}} = e_{i-1} \qquad i=2, \ldots, 2k-1.   
    \label{eq:commn2k1}
\end{equation}
Alternatively these commutations relations can be written as:
\begin{equation}
    \comm{e_i}{e_{j}} = \delta_{j,2k} e_{i-1}, 
    \qquad 2 \leq i \le 2k-1.
    \label{eq:commn2k1bis}
\end{equation}
Observe that the element $e_1$ is central, and in fact that 
$Z(\mathfrak{n}_{2k,1}) = \langle e_1 \rangle$.

\begin{remark}
    We remark that this algebra is not present in this form in~\cite{Millionschikov2004}, 
    but corresponds to the Lie algebra $\mathfrak{m}_{0}(2k)$ through the
    Lie algebra isomorphism:
    \begin{equation}
        \begin{tikzcd}[row sep =tiny]
        \phi_{2k}\colon\mathfrak{n}_{2k,1} \arrow{rr} && \mathfrak{m}_{0}(2k)
        \\
        (e_1,\ldots,e_{2k}) \arrow[mapsto]{rr} && (e_{2k},\ldots,e_1).
    \end{tikzcd}
    \end{equation}
    We choose to present it in this form because in this basis it clearly
    a generalisation of the algebra $\mathfrak{n}_{4,1}$ which
    was addressed in the previous subsection.
\end{remark}

Since $\mathfrak{n}_{2k,1}$ is symplectic with simplectic form~\cite{Millionschikov2004}:
\begin{equation}
    \omega_{2k} = \sum_{i=1}^{k}  (-1)^{i} e^i \wedge e^{2k+1-i},
\end{equation}
we know already it is quasi-Frobenius. In what follows we compute a
quasi-triangular non-degenerate $r$-matrix for $\mathfrak{n}_{2k,1}$ with
$3k-2$ free parameters, thus extending then the results
of~\cite{Millionschikov2004}. This is the content of the following result:

\begin{proposition}
    The Lie algebra $\mathfrak{n}_{2k,1}$ admits the following non-degenerate
    skew-symmetric $r$-matrix:
    \begin{equation}
        r_{\mathfrak{n}_{2k,1}} = 
        \sum_{j=2}^{2k} a_j e_1 \wedge e_j
        +
        \sum_{\ell=2}^{k}\sum_{i=2}^{\ell}
        b_\ell  (-1)^{i+\ell} e_i \wedge e_{2\ell+1-i}.
        \label{eq:rn2k1}
    \end{equation}
\end{proposition}

\begin{proof}
    The proof consists in showing that the element 
    $r_{\mathfrak{n}_{2k,1}}\in \Lambda^2 \mathfrak{n}_{2k,1}$ 
    defined in equation~\eqref{eq:rn2k1} satisfies the CBYE~\eqref{eq:CYBE}.
    We will consider the CBYE written in terms of the Schouten bracket, i.e.\
    $\SB{r,r}=0$. In particular we will use two properties of the Schouten bracket, i.e.\
    the bilinearity and the graded Leibniz rule:
    \begin{equation}
        \SB{r_1,r_2\wedge r_3} =
        \SB{r_1,r_2}\wedge r_3 + (-1)^{(k_1-1)k_2} r_2 \wedge \SB{r_1,r_3},
        \quad
        r_i \in \Lambda^{k_i} V.
    \end{equation}
    In particular, the graded Leibniz rule implies on
    the base elements of $\Lambda ^2 V$ the following writing:
    \begin{equation}
        \begin{aligned}
        \SB{ e_i \wedge e_j, e_k \wedge e_\ell }
        &=
        \comm{ e_i}{ e_k } \wedge e_j \wedge e_\ell
        - \comm{ e_i}{ e_\ell } \wedge e_j \wedge e_k
        \\
        &- \comm{ e_j}{ e_k } \wedge e_i \wedge e_\ell
        + \comm{ e_j}{ e_\ell } \wedge e_i \wedge e_k.            
        \end{aligned}
        \label{eq:eiejekel}
    \end{equation}
    Using the bilinearity we obtain:
    \begin{equation}
        \begin{aligned}
            \SB{r_{\mathfrak{n}_{2k,1}},r_{\mathfrak{n}_{2k,1}}}    
            &=
            \sum_{j,j'=2}^{2k} a_ja_{j'} \SB{e_1 \wedge e_j,e_1 \wedge e_{j'}}
            \\
            &+
            2\sum_{j=2}^{2k} \sum_{\ell'=2}^{k}\sum_{i'=2}^{\ell'}
            a_j b_{\ell'}  (-1)^{i'+\ell'} \SB{e_1 \wedge e_j, e_{i'} \wedge e_{2\ell'+1-i'}}
            \\
            &+
            \sum_{\ell=2}^{k}\sum_{i=2}^{\ell}\sum_{\ell'=2}^{k}\sum_{i'=2}^{\ell'}
            b_{\ell} b_{\ell'}  (-1)^{i+\ell+i'+\ell'} \SB{e_{i} \wedge e_{2\ell+1-i},e_{i'} \wedge e_{2\ell'+1-i'}}
        \end{aligned}
        \label{eq:schn2k1}
    \end{equation}
    We now prove that all the Schouten brackets in the sums vanish.
    
    Using~\eqref{eq:eiejekel} and~\eqref{eq:commn2k1} we obtain:
    \begin{equation}
        \begin{aligned}
            \SB{e_1 \wedge e_j,e_1 \wedge e_{j'}}
            &=
            \comm{ e_1}{ e_1 } \wedge e_j \wedge e_{j'}
            - \comm{ e_1}{ e_{j'} } \wedge e_j \wedge e_1
            \\
            &- \comm{ e_j}{ e_1 } \wedge e_1 \wedge e_{j'}
            + \comm{ e_j}{ e_{j'} } \wedge e_1 \wedge e_1
            \\
            &=- \comm{ e_1}{ e_{j'} } \wedge e_j \wedge e_1
            - \comm{ e_j}{ e_1 } \wedge e_1 \wedge e_{j'}
            =0,
        \end{aligned}
    \end{equation}
    since $e_1$ is a central element. So, the elements in the first 
    sum of~\eqref{eq:schn2k1} vanish.
    
    In a similar way using this time
    the compact form of the commutation rules~\eqref{eq:commn2k1bis}:
    \begin{equation}
        \begin{aligned}
        \SB{e_1 \wedge e_j, e_{i'} \wedge e_{2\ell'+1-i'}}
        &= [ e_1, e_{i'} ] \wedge e_j \wedge e_{2\ell'+1-i'} 
        - [ e_1, e_{2\ell'+1-i'} ] \wedge e_j \wedge e_{i'} 
        \\
        & - [ e_j, e_{i'} ] \wedge e_1 \wedge e_{2\ell'+1-i'} 
        + [ e_j, e_{2\ell'+1-i'} ] \wedge e_1 \wedge e_{i'}
        \\
        &=- \delta_{i',2k}e_{j-1}  \wedge e_1 \wedge e_{2\ell'+1-i'} 
        + \delta_{2\ell'+1-i',2k} e_{j-1} \wedge e_1 \wedge e_{i'}.
        \end{aligned}
        \label{eq:n2k1add2}
    \end{equation}
    The first summand in~\eqref{eq:n2k1add2} is non null only
    if $i'=2k$, but $2\leq i'\leq \ell'$ with $2\leq \ell'\leq k$, so $i'<2k$
    always, and the term vanishes. On the other hand the second summand 
    in~\eqref{eq:n2k1add2} is non null only
    if $2m'+1-i'=2k$, i.e.\ $i'=1+2(m'-k)$, but:
    \begin{equation}
        5 -2k \leq 1+2(m'-k) \leq 1.
    \end{equation}
    Since $k\geq 1$ it means that the $5-2k\leq 1$, and then the index
    $i'$ in fact out of range, implying that also the the second summand 
    in~\eqref{eq:n2k1add2} vanishes, i.e.:
    \begin{equation}
        \SB{e_1 \wedge e_j, e_{i'} \wedge e_{2m'+1-i'}}
        =0,
        \label{eq:n2k1add2bis}
    \end{equation}
    and the elements in the second sum of~\eqref{eq:schn2k1} vanish.

    Finally, let us consider the elements in the third and final
    sum of~\eqref{eq:schn2k1} using again
    the compact form of the commutation rules~\eqref{eq:commn2k1bis}:
    \begin{equation}
        \begin{aligned}
        \SB{e_{i} \wedge e_{2\ell+1-i},e_{i'} \wedge e_{2\ell'+1-i'}}
        &= [ e_i, e_{i'} ] \wedge e_{2\ell+1-i} \wedge e_{2\ell'+1-i'}
        - [ e_i, e_{2\ell'+1-i'} ] \wedge e_{2\ell+1-i} \wedge e_{i'} 
        \\
        &- [ e_{2\ell+1-i}, e_{i'} ] \wedge e_i \wedge e_{2\ell'+1-i'}
        + [ e_{2\ell+1-i}, e_{2\ell'+1-i'} ] \wedge e_i \wedge e_{i'}
        \\
        &\hspace*{-3ex}= \delta_{i',2k} e_{i-1} \wedge e_{2\ell+1-i} \wedge e_{2\ell'+1-i'}
        - \delta_{2\ell'+1-i',2k} e_{i-1} \wedge e_{2\ell+1-i} \wedge e_{i'} 
        \\
        &\hspace*{-3ex}+ \delta_{2\ell+1-1,2k} e_{i-1} \wedge e_i \wedge e_{2\ell'+1-i'}
        + \delta_{2\ell'+1-i',2k} e_{2\ell-i} \wedge e_i \wedge e_{i'}.
        \end{aligned}
        \label{eq:n2k1add3}
    \end{equation}
    With the same reasoning all the terms in~\eqref{eq:n2k1add3} vanish,
    i.e.\ $\SB{e_1 \wedge e_j, e_{i'} \wedge e_{2m'+1-i'}}=0$.
    
    Thus, we proved that $\SB{r_{\mathfrak{n}_{2k,1}},r_{\mathfrak{n}_{2k,1}}}\equiv 0$
    and the proof ends.
\end{proof}

To complete the construction in~\Cref{thm:Dubrovin_quasiF} we need to compute
the left and right invariant vector fields of $\mathfrak{n}_{2k,1}$. In principle,
this can be be done through matrix representations, as $\mathfrak{n}_{2k,1}$
admits a maximal abelian ideal of dimension $2k-1$, namely
ideal $\mathfrak{i}_{2k-1} = \langle e_1, e_2, \dots, e_{2k-1} \rangle$.
A simple application of~\Cref{thm:repr} gives the following faithful representation
in the space of $2k\times 2k$ matrices~$\rho\colon \mathfrak{n}_{2k,1} \longrightarrow \mathfrak{gl}(2k, \mathbb{R})$:
    \begin{align}
        \rho(e_i) &= E_{i, 2k},\quad  i=1, \ldots, 2k-1,
        \qquad \qquad
        \rho(e_{2k}) = -\sum_{j=1}^{2k-1} E_{j, j+1}.    
    \end{align}
However, it is simpler to consider the construction provided by~\Cref{prop:ivf}.
Indeed, we have the following values for the adjoint of the base elements
of $\mathfrak{n}_{2k,1}$:
    \begin{align}\label{eq:adjn2k1}
        \ad(e_1) &= O_{2k},
        \qquad 
        \ad(e_i) = E_{i-1,2k},
        \quad  
        i=2, \ldots, 2k-1,
        \qquad 
        \ad(e_{2k}) = -\sum_{j=1}^{2k-1} E_{j, j+1}.
    \end{align}
Observe that $\ad(e_i)$, $i=2,\ldots,2k-1$, are nilpotent matrices of order~two, 
and $\ad(e_{2k})$ 
is a nilpotent matrix of order~$2k-1$. Using, as always $u_n^i$ as coordinates 
on $N_{2k,1}=\exp (\mathfrak{n}_{2k,1})$ this last observation readily implies:
\begin{subequations}
    \begin{align}
        M_1 &= I_{2k},
        \qquad
        M_i = I_{2k} - u^{i}_nE_{i-1,2k},
        \quad i=2,\ldots, 2k-1,
        \\[1mm]
        M_{2k} &= 
        \sum_{i=1}^{2k-1}\sum_{j=0}^{2j-1-i} \frac{(u_n^{2k})^j}{j!}\, E_{i,i+j} 
        + E_{2k,2k}.
    \end{align}%
\end{subequations}
Using the rules of multiplication of elementary matrices we obtain:
\begin{subequations}
    \begin{align}
        F^{(2)} &= I_{2k},
        \qquad 
        F^{(i)} = I_{2k} - \sum_{j=2}^{i} u^{j}_nE_{j-1,2k},
        \quad i=3,\ldots, 2k,
        \\[1mm]
        F^{(2k+1)} &= 
        \sum_{i=1}^{2k-1}\sum_{j=0}^{2j-1-i} \frac{(u_n^{2k})^j}{j!} E_{i,i+j} 
        - \sum_{j=2}^{2k-1} u^{j}_nE_{j-1,2k}+ E_{2k,2k}.
    \end{align}%
\end{subequations}
This yields the following form for the matrices $\Omega$ and $\Sigma$:
\begin{subequations}
    \begin{align}
        \Omega &= I_{2k} - \sum_{j=2}^{2k-1} u^{j}_nE_{j-1,2k},
        \\
        \Sigma &= 
        \sum_{i=1}^{2k-1}\sum_{j=0}^{2j-1-i} \frac{(-u_n^{2k})^j}{j!}E_{i,i+j} 
        + E_{2k,2k}.
    \end{align}%
\end{subequations}
To invert the matrices $F^{(k)}$ we used that the matrices $\ad(e_i)$~\eqref{eq:adjn2k1} 
are traceless matrices (elements of $\mathfrak{sl}(2k,\mathbb{R})$), so that
their matrix exponential (and their products) are elements of $\SL(2k,\mathbb{R})$.
Finally, inverting and transposing we obtain the matrices of the invariant vector
fields:
\begin{subequations}
    \begin{align}
        L &= \sum_{i=1}^{2k-1}\sum_{j=0}^{2j-1-i} \frac{(u_n^{2k})^j}{j!}E_{i+j,i} 
        + E_{2k,2k},
        \\
        R &= I_{2k} + \sum_{j=2}^{2k-1} u^{j}_nE_{j-1,2k},
    \end{align}
\end{subequations}

\subsubsection{The family $\mathfrak{V}_{2k}$}

Consider the Lie algebra $\mathfrak{V}_{2k}=\langle e_1, e_2, \dots, e_{2k}\rangle$, 
whose only non-zero commutation relations are:
\begin{equation}
    \comm{e_i}{e_{j}} = (j-i)e_{i+j} \qquad  i+j \leq 2k.   
    \label{eq:commV2k}
\end{equation}
We can consider $k\geq 3$ because the case $k=1$ is trivial, and 
the case $k=2$ is in fact isomorphic to the algebra $\mathfrak{n}_{4,1}$
through the isomorphism:
\begin{equation}
    \begin{tikzcd}[row sep =tiny]
        \iota_4\colon\mathfrak{n}_{4,1} \arrow{rr} && \mathfrak{V}_{4}
        \\
        (e_1,e_2,e_3,e_4) \arrow[mapsto]{rr} && (e_4,e_3,-e_2,2e_1).
    \end{tikzcd}
\end{equation}

It is known that $\mathfrak{V}_{2k}$ is symplectic with simplectic form~\cite{Millionschikov2004}:
\begin{equation}
    \omega_{2k} = \sum_{i=1}^{k} (2k+1-2i) \, e^i \wedge e^{2k-1-i}.
\end{equation}
By duality, we obtain (up to a constant multiple) the following non-degenerate $r$-matrix:
\begin{equation}
    r_{\mathfrak{V}_{2k}} = \sum_{i=1}^{k} \frac{1}{2k+1-2i} \, e_i \wedge e_{2k-1-i}.
\end{equation}

\begin{remark}
    By direct computation for low $k$ we were able to build several
    other parametric families of non-degenerate $r$-matrices. However, differently
    from the case of the $\mathfrak{n}_{2k,1}$ we were not able to devise
    a common pattern. So, we leave the discussion of other possible non-degenerate $r$-matrices
    to future works.
\end{remark}

Therefore, to build the associated dDGP bracket we need to construct the invariant vector fields
of $\mathfrak{V}_{2k}$. In this case, there is no obvious representation of this Lie algebra,
and the application of~\Cref{thm:nilpotrepr} is not trivial for general $k$. For this
reason, we apply again~\Cref{prop:lrinv}. In this case we have:
\begin{equation}
    \ad(e_i)=\sum_{j=1}^{2k-i}(j-i)E_{j+i,j}, \qquad i=1,\dots,2k.    
\end{equation}
Note that from the above formula $\ad(e_{2k})=0$, which is consistent
with it being a central element.
Because the matrices $\ad(e_i)$ act on the standard basis of $\mathbb{R}^{2k}$,
$\vb*{v}_1$, \ldots, $\vb*{v}_{2k}$, as:
\begin{equation}
    \ad(e_i)\vb*{v}_{k} = 
    \begin{cases}
        (k-i)\vb*{v}_{k+i}, & \text{if $1\leq k\leq 2k-i$}
        \\
        \vb*{0}, & \text{otherwise},
    \end{cases}
    \label{eq:v2kadact}
\end{equation}
we have that the matrices $\ad(e_i)$ are nilpotent matrices of nilpotency order:
\begin{equation}
    \nil (\ad(e_i)) =
    \begin{cases}
        2k-1, & \text{if $i=1$},
        \\[2mm]
        \left\lfloor \dfrac{2k-1}{i} \right\rfloor+1, & \text{if $2\leq i\leq 2k$}.
        \end{cases}
\end{equation}
Thanks to this observation and the previous formula~\eqref{eq:v2kadact} we obtain:
\begin{equation}
    M_i = \exp(-u^i\ad(e_i))
    =
    I_{2k}+
    \sum_{\ell=1}^{\lfloor (2k-1)/i\rfloor}
    \sum_{j=1}^{2k-\ell i}
    \frac{(-u^i)^\ell}{\ell!}
    \left(
    \prod_{m=1}^{\ell}\bigl(j+(m-2)i\bigr)
    \right)
    E_{j+\ell i,j}.    
\end{equation}

This yields the following formula for elements of the matrices $F^{(i)}$:
\begin{equation}
(F^{(i)})_{p,q}=
\begin{cases}
1, & p=q,\\[1mm]
0, & p<q,\\[1mm]
\displaystyle
\sum_{\substack{k_1,\dots,k_{i-1}\ge 0\\
\sum_{r=1}^{i-1} r k_r=p-q}}
\left(
\prod_{r=1}^{i-1}
\frac{(-u^r)^{k_r}}{k_r!}
\right)
\prod_{r=1}^{i-1}
\prod_{m=0}^{k_r-1}
\left(
q+\sum_{s=r+1}^{i-1}s k_s+m r-r
\right),
& p>q,    
\end{cases}
\label{eq:Fiv2k}
\end{equation}
with the convetion that an empty product is $1$. Moreover, if
$p>q$ and $p-q$ are representable as 
$\sum_{r=1}^{i-1} r k_r$, the sum is empty and $(F^{(i)})_{p,q}\equiv0$.

\begin{proof}[Proof of eq.~\eqref{eq:Fiv2k}]
    Let us observe that since $M_i = \exp(-u^i \ad(e_i))$ we have:
    \begin{equation}
        F^{(i)}
        =
        \sum_{k_1,\dots,k_{i-1}\ge 0}
        \left(
        \prod_{r=1}^{i-1}
        \frac{(-u^r)^{k_r}}{k_r!}
        \right)
        \ad(e_1)^{k_1}\ad(e_2)^{k_2}\cdots \ad(e_{i-1})^{k_{i-1}}.        
    \end{equation}
    Since the matrices $\ad(e_i)$ are nilpotent the sum is finite. Now the
    key is to apply the product $P_{k_1,\ldots,k_{i-1}}=\ad(e_1)^{k_1}\ad(e_2)^{k_2}\cdots \ad(e_{i-1})^{k_{i-1}}$
    to the vector of the basis $\vb*{v}_q$. Observe that the blocks with higher indices 
    act first. Let us use again~\eqref{eq:v2kadact}: this implies that one has the following
    form of the action
    \begin{equation}
        P_{k_1,\ldots,k_{i-1}}\vb*{v}_{q}
        =
        \prod_{r=1}^{i-1}
        \prod_{m=0}^{k_r-1}
        \left(
        q+\sum_{s=r+1}^{i-1}s k_s+m r-r
        \right)\vb*{v}_{q+d},
        \quad
        d=\sum_{r=1}^{i-1} r k_r.
    \end{equation}
    Let us observe that $d\geq 0$ because it is a sum of non-negative integers.
    This readily implies that for $p<q$ the entry $(F^{(i)})_{p,q}\equiv0$ since there
    is no solution to the Diophantine equation $\sum_{r=1}^{i-1} r k_r = d = p-q <0$. For 
    $p=q$, the only solution of $\sum_{r=1}^{i-1} r k_r=d=p-q=0$  is obtained for $k_r=0$.
    In such a case the product collapses to $1$. 
    So, for $p>q$, let us put $d=p-q$. If there is no solution to $\sum_{r=1}^{i-1} r k_r=d$, 
    the sum is empty and gives $(F^{(i)})_{p,q}\equiv0$. Otherwise we have:
    \begin{equation}
        (F^{(i)})_{p,q}
        \equiv
\sum_{\substack{k_1,\dots,k_{i-1}\ge 0\\
\sum_{r=1}^{i-1} r k_r=d}}
\left(
\prod_{r=1}^{i-1}
\frac{(-u^r)^{k_r}}{k_r!}
\right)
\prod_{r=1}^{i-1}
\prod_{m=0}^{k_r-1}
\left(
q+\sum_{s=r+1}^{i-1}s k_s+m r-r
\right).
    \end{equation}
    This concludes the proof.
\end{proof}

Formula~\eqref{eq:Fiv2k} allows us to write down the matrices $\Omega$
and $\Sigma$: 
\begin{subequations}
    \begin{align}
        \Omega_{i,k} &=
\begin{dcases}
0, & 1 = k <i,\\[1mm]
0, & 1\le i<k,\\[1mm]
1, & 1\le i=k,\\[1mm]
\displaystyle
\sum_{\substack{\ell_1,\dots,\ell_{k-1}\ge 0\\
\sum_{r=1}^{k-1} r\ell_r = i-k}}
\left(
\prod_{r=1}^{k-1}
\frac{(-u^r)^{\ell_r}}{\ell_r!}
\right)
\prod_{r=1}^{k-1}
\prod_{m=0}^{\ell_r-1}
\left(
k+\sum_{s=r+1}^{k-1}s\ell_s+m r-r
\right),
& 1<k<i,
\end{dcases}
    \\[3mm]
    \Sigma_{i,k}&=
\begin{dcases}
0, & i<k,\\[1mm]
-1, & i=k,\\[1mm]
\displaystyle
-\sum_{\substack{\ell_k,\dots,\ell_{n-1}\ge 0\\[1mm]
\sum_{r=k}^{n-1} r\ell_r=i-k}}
\left(
\prod_{r=k}^{n-1}
\frac{(u^r)^{\ell_r}}{\ell_r!}
\right)
\prod_{r=k}^{n-1}
\prod_{m=0}^{\ell_r-1}
\left(
k+\sum_{s=r+1}^{n-1}s\ell_s+m r-r
\right),
& i>k.
\end{dcases}
    \end{align}
    \label{eq:OmSigV2k}%
\end{subequations}
Inverting these matrices we obtain the desired invariant vector fields.

We were not able to provide a full closed form for the general case, but for
instance for $k=3$ we obtain from~\eqref{eq:OmSigV2k}:
\begin{subequations}
    \begin{align}
        \Omega &=
        \begin{pmatrix}
1 & 0 & 0 & 0 & 0 & 0\\[1mm]
0 & 1 & 0 & 0 & 0 & 0\\[1mm]
0 & -u^1 & 1 & 0 & 0 & 0\\[1mm]
0 & (u^1)^2 & -2u^1 & 1 & 0 & 0\\[1mm]
0 & -(u^1)^3 & 3(u^1)^2-u^2 & -3u^1 & 1 & 0\\[1mm]
0 & (u^1)^4 & -4(u^1)^3+4u^1u^2 & 6(u^1)^2-2u^2 & -4u^1 & 1
\end{pmatrix},
        \\[3mm]
        \Sigma &=
        \begin{pmatrix}
-1 & 0 & 0 & 0 & 0 & 0\\[1mm]
0 & -1 & 0 & 0 & 0 & 0\\[1mm]
u^2 & 0 & -1 & 0 & 0 & 0\\[1mm]
2u^3 & 0 & 0 & -1 & 0 & 0\\[1mm]
\frac{1}{2}(u^2)^2+3u^4 & u^3 & 0 & 0 & -1 & 0\\[1mm]
4u^5 & 2u^4 & 0 & 0 & 0 & -1
\end{pmatrix},
    \end{align}
\end{subequations}
implying:
\begin{subequations}
    \begin{align}
    L &=
\begin{pmatrix}
-1 & 0 & 0 & 0 & 0 & 0\\[1mm]
0 & -1 & 0 & 0 & 0 & 0\\[1mm]
-u^2 & 0 & -1 & 0 & 0 & 0\\[1mm]
-2u^3 & 0 & 0 & -1 & 0 & 0\\[1mm]
-\frac{1}{2}(u^2)^2-3u^4 & -u^3 & 0 & 0 & -1 & 0\\[1mm]
-4u^5 & -2u^4 & 0 & 0 & 0 & -1
\end{pmatrix},
    \\[3mm]
    R &=
    \begin{pmatrix}
1 & 0 & 0 & 0 & 0 & 0\\[1mm]
0 & 1 & u^1 & (u^1)^2 & (u^1)^3+u^1u^2 & (u^1)^4+2(u^1)^2u^2\\[1mm]
0 & 0 & 1 & 2u^1 & 3(u^1)^2+u^2 & 4(u^1)^3+4u^1u^2\\[1mm]
0 & 0 & 0 & 1 & 3u^1 & 6(u^1)^2+2u^2\\[1mm]
0 & 0 & 0 & 0 & 1 & 4u^1\\[1mm]
0 & 0 & 0 & 0 & 0 & 1
\end{pmatrix}.
    \end{align}
\end{subequations}
Other particular cases can be constructed similarly.

\section{Conclusions and outlook}
\label{sec:conclusions}

In this paper, following Parodi's work~\cite{ParodiThesis}, we presented 
a complete proof of Dubrovin's characterisation of non-degenerate
dDGP brackets~\cite{Dubrovin1989}, and used it to construct new 
explicit examples. 

To be more precise, we gave a detailed proof of the characterisation of 
non-degenerate~dDGP  brackets of order $(1)$ (on-site and inter-site) 
in terms of the Hamiltonianity conditions~\eqref{eq:pdncond} 
(\Cref{thm:dDNPBconds}). We also gave a geometric interpretation in terms 
of Poisson--Lie groups and Lie bialgebras (\Cref{thm:Dubrovin_general}), 
encoded by the pair $(q,k)$, where $q$ is a Lie isomorphism and $k$ is a 
skew-symmetric deformation solving the generalised classical Yang--Baxter 
equation~\eqref{eq:GYBE}. We then focused on the quasi-Frobenius case (\Cref{thm:Dubrovin_quasiF}), in which $(q,k)=(r,-r)$, and the pair 
reduces to a single non-degenerate classical $r$-matrix, i.e.\ a non-degenerate
solution of the classical Yang--Baxter equation~\eqref{eq:CYBE}. 
We used this result to produce explicit examples: an exhaustive classification 
in dimension four for the admissible algebras, and the two one-parameter 
families $\mathfrak n_{2k,1}$ and $\mathfrak V_{2k}$ in arbitrary dimension.

Dubrovin's insight on the structure of the dDGP 
brackets proved to be a great source of examples. This is somewhat surprising,
since in his thesis Parodi~\cite[Remark 2.4.5]{ParodiThesis} observed that only 
one non-trivial example of application of \Cref{thm:Dubrovin_general}
was known, namely the one related to the non-abelian Lie algebra
$\mathfrak{aff}(1)\cong \mathfrak{s}_{2,1}$, which is reproduced also 
in~\cite[\S 3.1.1]{CasatiValeri} and already appeared in \cite[Ex.\ 1]{Dubrovin1989},
see also the discussion at the end of \Cref{sec:geom}.
In fact, as we commented in \Cref{sec:geom}
a direct application of \Cref{thm:Dubrovin_general} is very difficult,
as it requires the solution of several systems of coupled non-linear (quadratic)
equations. However, we observed that \Cref{thm:Dubrovin_quasiF}, offers
a natural and easier way to produce examples. This, coupled with
modern results on quasi-Frobenius Lie algebras, and Lie algebras
in general, paved the way to produce a wealth of examples, and especially
to prove the existence of non-trivial families of dDGP brackets in arbitrary even dimensions. To this end, it must be observed
that the results presented in \Cref{sec:examples} are just a starting point,
as there are several natural problems arising from our use of 
\Cref{thm:Dubrovin_quasiF}. For instance, it would be possible to use
the classification of six-dimensional Lie 
algebras~\cite[Chap.\ 19]{SnobWinternitz2017book} and filter all
quasi-Frobenius Lie algebras up to such a dimension. On the other hand,
quasi-Frobenius filiform Lie algebras are also classified up to dimension 
ten~\cite{Gomez_etal2001}\footnote{We have to observe that 
in~\cite{Millionshchikov2006} it is reported that there are few misprint
in~\cite{Gomez_etal2001}, so particular care is needed while using those results.}, 
and quasi-Frobenius nilpotent Lie algebras are known also up to dimension 
eight~\cite{AitAissaMansouri2026}.
In fact, it is also known that it is possible to produce
systematically quasi-Frobenius Lie algebras through a procedure
called \emph{(generalised) double extension}, see~\cite{Medina1985,Dardie1996a}. 
This procedure takes a quasi-Frobenius Lie algebra $\mathfrak{g}$ of dimension 
$d$, a derivation $\mathcal{D}$ compatible with the structure, and a two-cocycle 
$\gamma$  and gives a quasi-Frobenius Lie algebra $\mathfrak{g}'$ of dimension $d+2$.
Under some additional assumption on the derivation and the cocycle,
the obtained algebra is indecomposable, thus it is a genuinely new Lie algebra.
Therefore, we believe that it is possible
to build additional families of interesting dDGP brackets through \Cref{thm:Dubrovin_quasiF}
and a careful application of the generalised double extension procedure. In
particular, this procedure can be applied to the obtained classifications of 
quasi-Frobenius Lie algebras, hoping to produce other infinite families.

Finally, we observe that in~\cite[eq.\ (6)]{Dubrovin1989} Dubrovin  announced (without a proof)
that the continuum limit of the dDGP bracket of order $(1)=[1]+[0]$ is a Dubrovin--Novikov bracket, 
i.e.\ homogeneous of order~$1$. Later, Parodi~\cite[Remark 2.4.6]{ParodiThesis} 
referred to a continuum limit of order~$1+0$ (again without a proof).  In the forthcoming 
paper~\cite{contiDisc}, we will show that, in a more general continuum limit, the 
resulting differential operator is in fact non-homogeneous of type~$1+0$, reflecting 
the non-homogeneity of the discrete bracket. To the best of the authors' knowledge, 
this is the first systematic investigation of discrete differential-geometric 
structures in the continuum limit. The forthcoming work will establish 
a precise correspondence between the discrete and continuous frameworks, providing 
a rigorous proof of the correspondence suggested by the discrete-to-continuous analogy.

\section*{Acknowledgments}

We thank Prof.\ Jing Ping Wang, Prof.\ Francisco Herranz, Dr.\ Edoardo Peroni,
and Dr.\ Emanuele Sgroi for the insightful discussions during the preparation of
this paper. 

GG and PV acknowledges the financial support of GNFM of the Istituto Nazionale di 
Alta Matematica and  are partially funded by the research project Mathematical Methods 
in Non-Linear Physics (MMNLP) by the Commissione Scientifica Nazionale – Gruppo 4 – 
Fisica Teorica of the Istituto Nazionale di Fisica Nucleare (INFN), Sezione di Milano 
and Lecce respectively.	

PV also acknowledges the kind hospitality of the Dipartimento di Matematica 
``Federigo Enriquez'' of Universit\`a Statale di Milano where part of this manuscript
was prepared.

MDA has been partially supported by the ERC STARTING GRANT 2021 ``Hamiltonian 
Dynamics, Normal Forms and Water Waves'' (HamDyWWa), Project Number: 101039762. The 
Views and opinions expressed are however those of the authors only and do not 
necessarily reflect those of the European Union or the European Research Council. 
Neither the European Union nor the granting authority can be held responsible for 
them.

\appendix

\section{Dubrovin--Parodi proof of the Hamiltonian conditions}
\label{app:dp}

In this Appendix we give a pedagogical proof of the conditions
the Hamiltonian conditions for a dDGP bracket of order $(1)$, i.e.\ a proof of~\Cref{thm:dDNPBconds}.
This proof follows closely the original ideas of Dubrovin~\cite{Dubrovin1989} and 
Parodi~\cite{ParodiThesis}\footnote{During the proof we will emend a few typos 
we found in the original text.}.

The first step it to observe that since the dDGP bracket of order $(1)$ is local,
the Jacobi identity for the coordinate functions~\eqref{eq:jacbase} involve 
the local coordinates evaluated on the sites $(n,n+1,n+2)$, $(n,n,n+1)$, $(n-1,n,n)$,
and $(n,n,n)$. All other combinations either follow from these four or will 
vanish identically because of the locality conditions. This implies that
the Jacobi identity reduces to the following equations:
\begin{subequations}
    \begin{gather}
        \pb*{\pb*{u_n^i}{u_{n+1}^j}}{u_{n+2}^k}
        = \pb*{u_n^{i}}{\pb*{u_{n+1}^j}{u_{n+2}^k}},
        \label{eq:jacnnpnpp}
        \\
        \pb*{\pb*{u_n^i}{u_{n}^j}}{u_{n+1}^k}
        +
        \pb*{\pb*{u_n^j}{u_{n+1}^k}}{u_{n}^i}
        +
        \pb*{\pb*{u_{n+1}^k}{u_{n}^i}}{u_{n}^j}=0,
        \label{eq:jacnnnp}
        \\
        \pb*{\pb*{u_n^i}{u_{n}^j}}{u_{n-1}^k}
        +
        \pb*{\pb*{u_n^j}{u_{n-1}^k}}{u_{n}^i}
        +
        \pb*{\pb*{u_{n-1}^k}{u_{n}^i}}{u_{n}^j}=0,
        \label{eq:jacnmnn}
        \\
        \pb*{\pb*{u_n^i}{u_{n}^j}}{u_{n}^k}
        +
        \pb*{\pb*{u_n^j}{u_{n}^k}}{u_{n}^i}
        +
        \pb*{\pb*{u_{n}^k}{u_{n}^i}}{u_{n}^j}=0,
        \label{eq:jacnnn}
    \end{gather}
\end{subequations}
coming from the four possible cases respectively.
Observe that we used the locality condition in~\eqref{eq:jacnnpnpp} to
annihilate the term $\pb{u_n}{u_{n+2}}$.

The last condition~\eqref{eq:jacnnn} is the simplest one to deal with.
We identify again $g$ and $h$ as defined in~\eqref{eq:discretePBexp}:
\begin{equation}
    \pb*{h^{ij}_n}{u_{n}^k}
    +
    \pb*{h^{jk}_n}{u_{n}^i}
    +
    \pb*{h^{ki}_n}{u_{n}^j}=0,
    \label{eq:jacnnn1}
\end{equation}
which from the derivation property implies:
\begin{equation}
    \pdv{h^{ij}_n}{u_{n}^s}h^{sk}_n
    +
    \pdv{h^{jk}_n}{u_{n}^s}h^{si}_n
    +
    \pdv{h^{ki}_n}{u_{n}^s}h^{sj}_n
    =0.
    \label{eq:jacnnn2}
\end{equation}
This is nothing but equation~\eqref{eq:pdncond4}.

Next let us turn to the first condition~\eqref{eq:jacnnpnpp}, that we can rewrite as:
\begin{equation}
    \pb*{g_n^{ij}}{u_{n+2}^k}
        = \pb*{u_n^{i}}{g_{n+1}^{jk}}.
    \label{eq:jacnnpnpp1}
\end{equation}
For the sake of clarity, we show how to develop the brackets involved by taking the left hand side in~\eqref{eq:jacnnpnpp1} as an example: 
\begin{equation}
\begin{split} 
    \pb*{g_n^{ij}}{u_{n+2}^k} = \frac{\partial g_n^{ij}}{\partial u^s_n} \pb*{u_n^{s}}{u_{n+2}^k} + \frac{\partial g_{n}^{ij}}{\partial u^s_{n+1}} \pb*{u_{n+1}^{s}}{u_{n+2}^k} = \frac{\partial g_{n}^{ij}}{\partial u^s_{n+1}}\,g_{n+1}^{sk}\,.
\end{split} 
\end{equation}
Overall, the condition~\eqref{eq:jacnnpnpp1} implies the following differential constraint:
\begin{equation}
    \pdv{g_n^{ij}}{u_{n+1}^s}\,g_{n+1}^{sk}
    =
    g_{n}^{is}\, \pdv{g_{n+1}^{jk}}{u_{n+1}^s},
    \label{eq:jacnnpnpp2}
\end{equation}
then obtaining equation~\eqref{eq:pdncond1}.

In the condition~\eqref{eq:jacnnnp}, we identify again $g$ and $h$ as defined in~\eqref{eq:discretePBexp}:
\begin{equation}
    \underbrace{\pb*{h_n^{ij}}{u_{n+1}^k}}_{a^{ijk}}
    +
    \underbrace{\pb*{g_{n}^{jk}}{u_{n}^i}}_{b^{ijk}}
    -
    \underbrace{\pb*{g_{n}^{ik}}{u_{n}^j}}_{c^{ijk}}=0.
    \label{eq:jacnnnp1}
\end{equation}
Observe first that we have $c^{ijk}=b^{jik}$, so we have to compute 
$b^{ijk}$ only. Using the derivation property we obtain:
\begin{equation}
    a^{ijk} = \pdv{h^{ij}_{n}}{u_{n}^{s}}g_{n}^{sk},
    \label{eq:jacnnnp1a}
\end{equation}
and in the same way we have:
\begin{equation}
    b^{ijk} = 
    \pdv{g_{n}^{jk}}{u_{n}^{s}}h_{n}^{si}
    -
    \pdv{g_{n}^{jk}}{u_{n+1}^{s}}g_{n}^{si}\,. 
    \label{eq:jacnnnp1b}
\end{equation}
Therefore,~\eqref{eq:jacnnnp1} is equivalent to the following
expression:
\begin{equation}
    \pdv{h^{ij}_{n}}{u_{n}^{s}}g_{n}^{sk}
    +
    \pdv{g_{n}^{jk}}{u_{n}^{s}}h_{n}^{si}
    -
    \pdv{g_{n}^{jk}}{u_{n+1}^{s}}g_{n}^{is}
    -
    \pdv{g_{n}^{ik}}{u_{n}^{s}}h_{n}^{sj}
    +
    \pdv{g_{n}^{ik}}{u_{n+1}^{s}}g_{n}^{js}
    =0\,.
    \label{eq:jacnnnp2}
\end{equation}
This is nothing but equation~\eqref{eq:pdncond2}.

A similar proof holds for the third and final condition~\eqref{eq:jacnmnn}.
In particular, we identify $g$ and $h$ and obtain:
\begin{equation}
    \underbrace{\pb*{h_{n}^{ij}}{u_{n-1}^k}}_{\tilde{a}^{ijk}}
    -
    \underbrace{\pb*{g_{n-1}^{kj}}{u_{n}^i}}_{\tilde{b}^{ijk}}
    +
    \underbrace{\pb*{g_{n-1}^{ki}}{u_{n}^j}}_{\tilde{c}^{ijk}}
    =0.
    \label{eq:jacnmnn1}
\end{equation}
Again, we observe that we have $\tilde{c}^{ijk}=\tilde{b}^{jik}$, so we have to
compute $\tilde{b}^{ijk}$ only. With the derivation property we obtain:
\begin{equation}
    \tilde{a}^{ijk} = -\pdv{h^{ij}_{n}}{u_{n}^{s}}g_{n-1}^{ks},
    \label{eq:jacnmnn1a}
\end{equation}
and in the same way we have:
\begin{equation}
    \tilde{b}^{ijk} = 
    \pdv{g_{n-1}^{kj}}{u_{n-1}^{s}}
    g_{n-1}^{si}
    +
    \pdv{g_{n-1}^{kj}}{u_{n}^{s}}h_{n}^{si}\,.
    \label{eq:jacnmnn1b}
\end{equation}
Therefore, in the end~\eqref{eq:jacnmnn1} is equivalent to the following
expression:
\begin{equation}
    -\pdv{h^{ij}_{n}}{u_{n}^{s}}g_{n-1}^{ks}
    -\pdv{g_{n-1}^{kj}}{u_{n-1}^{s}}
    g_{n-1}^{si}
    -
    \pdv{g_{n-1}^{kj}}{u_{n}^{s}}h_{n}^{si}
    +\pdv{g_{n-1}^{ki}}{u_{n-1}^{s}}
    g_{n-1}^{sj}
    +
    \pdv{g_{n-1}^{ki}}{u_{n}^{s}}h_{n}^{sj}
    =0.
    \label{eq:jacnmnn2}
\end{equation}
Up to a sign, this is nothing but equation~\eqref{eq:pdncond3}. 
So, the proof of \Cref{thm:dDNPBconds} is complete.

\section{Invariant vector fields}
\label{app:left_right_vectors_long}
Here we report the left and right invariant vector fields for the 
algebras $\mathfrak{g}_{8}(\alpha)$ in~\ref{subsec:g8} and $\mathfrak{g}_{10}(\alpha)$ in~\ref{subsec:g10}.

\subsection{\texorpdfstring{$\mathfrak{g}_{8}(\alpha)$}{g8}}
The left and right invariant vector fields are:
\begin{subequations}
    \begin{align}
    \begin{split}
        L_1  &= \pdv{}{u_n^1} + u_n^2 \pdv{}{u_n^3} + u_n^3 \pdv{}{u_n^4} + \left(\frac{\alpha+2}{2}\,(u_n^2)^2 + u_n^4\right) \pdv{}{u_n^5} + u_n^5 \pdv{}{u_n^6} 
            \\ &+ \left(\frac{(\alpha+1)(\alpha+2)}{6}\,(u_n^2)^3 + \frac12 (u_n^3)^2 + u_n^6\right) \pdv{}{u_n^7} + \left(\frac{\alpha+2}{2}\,u_n^3 (u_n^2)^2 + u_n^7\right) \pdv{}{u_n^8},
    \end{split}
        \\
        L_2&
        \begin{aligned}[t]
                 &= \pdv{}{u_n^2} + (\alpha+2)\left(u_n^3 \pdv{}{u_n^5} + u_n^4 \pdv{}{u_n^6} \right)
                + (\alpha+1)u_n^5 \pdv{}{u_n^7} + \left(\frac{\alpha+2}{2}\,(u_n^3)^2 + \alpha u_n^6\right) \pdv{}{u_n^8},
        \end{aligned}
        \\
        L_3 &= \pdv{}{u_n^3} + u_n^4 \pdv{}{u_n^7} + u_n^5 \pdv{}{u_n^8},
        \\
        L_k &= \pdv{}{u_n^k}, \quad k = 4,5,6,7,8,
        \\ 
        R_j &= \pdv{}{u_n^j}, \quad j =1,8,
        \\
         \begin{split}
        R_2 &= \pdv{}{u_n^2} + u_n^1 \pdv{}{u_n^3} + \frac12 (u_n^1)^2 \pdv{}{u_n^4} + \left(\frac16 (u_n^1)^3 + (\alpha+2)u_n^1 u_n^2\right) \pdv{}{u_n^5} \\ 
        &+ \left(\frac1{24}(u_n^1)^4 + \frac{\alpha+2}{2}(u_n^1)^2 u_n^2\right) \pdv{}{u_n^6} \\ &+ \left(\frac1{120}(u_n^1)^5 + \frac{\alpha+1}{6}(u_n^1)^3 u_n^2 + \frac12 (u_n^1)^2 u_n^3 + \frac{(\alpha+1)(\alpha+2)}{2}u_n^1 (u_n^2)^2\right) \pdv{}{u_n^7} \\ &+ \left(\frac1{720}(u_n^1)^6 + \frac{\alpha}{24}(u_n^1)^4 u_n^2 + \frac16 (u_n^1)^3 u_n^3 + \frac{\alpha(\alpha+2)}{4}(u_n^1)^2 (u_n^2)^2 + (\alpha+2)u_n^1 u_n^2 u_n^3\right) \pdv{}{u_n^8}, \end{split} \\
        R_3 &\begin{aligned}[t] &= \pdv{}{u_n^3} + u_n^1 \pdv{}{u_n^4} + \left(\frac12 (u_n^1)^2 + (\alpha+2)u_n^2\right) \pdv{}{u_n^5} + \left(\frac16 (u_n^1)^3 + (\alpha+2)u_n^1 u_n^2\right) \pdv{}{u_n^6} \\ &+ \left(\frac1{24}(u_n^1)^4 + \frac{\alpha+1}{2}(u_n^1)^2 u_n^2 + u_n^1 u_n^3 + \frac{(\alpha+1)(\alpha+2)}{2}(u_n^2)^2\right) \pdv{}{u_n^7} \\ &+ \left(\frac1{120}(u_n^1)^5 + \frac{\alpha}{6}(u_n^1)^3 u_n^2 + \frac12 (u_n^1)^2 u_n^3 + \frac{\alpha(\alpha+2)}{2}u_n^1 (u_n^2)^2 + (\alpha+2)u_n^2 u_n^3\right) \pdv{}{u_n^8}, \end{aligned}
        \\
        R_4 &\begin{aligned}[t] &= \pdv{}{u_n^4} + u_n^1 \pdv{}{u_n^5} + \left(\frac12 (u_n^1)^2 + (\alpha+2)u_n^2\right) \pdv{}{u_n^6} %
        + \left(\frac16 (u_n^1)^3 + (\alpha+1)u_n^1 u_n^2 + u_n^3\right) \pdv{}{u_n^7} \\ &+ \left(\frac1{24}(u_n^1)^4 + \frac{\alpha}{2}(u_n^1)^2 u_n^2 + \frac{\alpha(\alpha+2)}{2}(u_n^2)^2 + u_n^1 u_n^3\right) \pdv{}{u_n^8}. \end{aligned}
        \\
        R_5 &\begin{aligned}[t] &= \pdv{}{u_n^5} + u_n^1 \pdv{}{u_n^6} + \left(\frac12 (u_n^1)^2 + (\alpha+1)u_n^2\right) \pdv{}{u_n^7} 
        + \left(\frac16 (u_n^1)^3 + \alpha u_n^1 u_n^2 + u_n^3\right) \pdv{}{u_n^8}, \end{aligned}
        \\
        R_6 &= \pdv{}{u_n^6} + u_n^1 \pdv{}{u_n^7} + \left(\frac12 (u_n^1)^2 + \alpha u_n^2\right) \pdv{}{u_n^8},
        \\
        R_7 &= \pdv{}{u_n^7} + u_n^1 \pdv{}{u_n^8},
    \end{align}
\end{subequations}

\subsection{\texorpdfstring{$\mathfrak{g}_{10}(\alpha)$}{g10}}
The left invariant vector fields are: 
\begin{subequations}
\begin{align}
    \begin{split}
    \label{eq:g10_L1}
    L_1 &= \pdv{}{u_n^1}
    + u_n^2 \pdv{}{u_n^3}
    + u_n^3 \pdv{}{u_n^4}
    + \left( \frac{\alpha+2}{2}(u_n^2)^2 + u_n^4 \right) \pdv{}{u_n^5}
    + u_n^5 \pdv{}{u_n^6} 
    \\
    &+ \left( \frac{\alpha^2+3\alpha+2}{6}(u_n^2)^3 + \frac12 (u_n^3)^2 + u_n^6 \right) \pdv{}{u_n^7}
    + \left( \frac{\alpha+2}{2}(u_n^2)^2 u_n^3 + u_n^7 \right) \pdv{}{u_n^8} \\
    &+ 
    \frac{1}{24}\dfrac{N_{1,9}
    }{2\alpha+5} \pdv{}{u_n^9} 
    + 
    \frac{1}{6}\frac{N_{1,10}}{2\alpha+5} \pdv{}{u_n^{10}},     
    \end{split}
    \\[1.2ex]
    L_2 &
    \begin{aligned}[t]
    &= \pdv{}{u_n^2}
    + (\alpha+2)u_n^3 \pdv{}{u_n^5}
    + (\alpha+2)u_n^4 \pdv{}{u_n^6}
    + (\alpha+1)u_n^5 \pdv{}{u_n^7}
    \\
    &+ \left( \frac{\alpha+2}{2}(u_n^3)^2 + \alpha u_n^6 \right) \pdv{}{u_n^8} 
    + \frac{(\alpha+2)\bigl((2\alpha-1)u_n^7 + 3u_n^3 u_n^4\bigr)}{2\alpha+5} \pdv{}{u_n^9} \\
    &
    + \frac{1}{2}\frac{(4\alpha^2+2\alpha-2)u_n^8 + (3\alpha+6)(u_n^4)^2}{2\alpha+5} \pdv{}{u_n^{10}},    
    \end{aligned}
    \\[1.2ex]
    L_3 &= \pdv{}{u_n^3}
    + u_n^4 \pdv{}{u_n^7}
    + u_n^5 \pdv{}{u_n^8}
    + \frac{2(\alpha+1)u_n^6}{2\alpha+5} \pdv{}{u_n^9}
    + \frac{(2\alpha-1)u_n^7}{2\alpha+5} \pdv{}{u_n^{10}}, \\[1.2ex]
    L_4 &= \pdv{}{u_n^4}
    + \frac{3u_n^5}{2\alpha+5} \left(\pdv{}{u_n^9}
    + \pdv{}{u_n^{10}} \right), 
    \\[1.2ex]
    L_k &= \pdv{}{u_n^k}, \quad k=5,6,7,8,9,10,
    \end{align}
\end{subequations}
where the numerators $N_{1,9}$ and $N_{1,10}$ in $L_1$~\eqref{eq:g10_L1} are
\begin{align*}
N_{1,9}&=
2\alpha^4(u_n^2)^4+9\alpha^3(u_n^2)^4+11\alpha^2(u_n^2)^4+36\alpha(u_n^2)^2 u_n^4
    -4(u_n^2)^4+72(u_n^2)^2 u_n^4+48\alpha u_n^8\\
    &+36(u_n^4)^2+120 u_n^8, \\[1mm] 
    N_{1,10}&=
        2\alpha^3(u_n^2)^3 u_n^3+5\alpha^2(u_n^2)^3 u_n^3+\alpha(u_n^2)^3 u_n^3
        +2\alpha(u_n^3)^3-2(u_n^2)^3 u_n^3-(u_n^3)^3+12\alpha u_n^9+30 u_n^9.
\end{align*}
The right invariant vector fields are: 
\begin{subequations}
    \begin{align} 
    R_j &= \pdv{}{u_n^j}, \quad j = 1,10, 
    \\[1.2ex]
     \label{eq:g10_R2}
    \begin{split}
    R_2 
    &= \pdv{}{u_n^2}
    + u_n^1 \pdv{}{u_n^3}
    + \frac{1}{2} (u_n^1)^2 \pdv{}{u_n^4}
    + \left( \frac16 (u_n^1)^3 + (\alpha+2)u_n^1 u_n^2 \right) \pdv{}{u_n^5} 
    \\
    & + \left( \frac1{24}(u_n^1)^4 + \frac{\alpha+2}{2}(u_n^1)^2 u_n^2 \right) \pdv{}{u_n^6}
    \\
    &+ \left( \frac1{120}(u_n^1)^5 + \frac{\alpha+1}{6}(u_n^1)^3 u_n^2 + \frac12 (u_n^1)^2 u_n^3 + \frac{(\alpha+1)(\alpha+2)}{2}u_n^1 (u_n^2)^2 \right) \pdv{}{u_n^7} 
    \\
    &+ \left( \frac1{720}(u_n^1)^6 + \frac{\alpha}{24}(u_n^1)^4 u_n^2 + \frac{\alpha(\alpha+2)}{4}(u_n^1)^2 (u_n^2)^2 + \frac16 (u_n^1)^3 u_n^3 + (\alpha+2)u_n^1 u_n^2 u_n^3 \right) \pdv{}{u_n^8} \\
    &  + \frac{u_n^1}{5040}  \frac{N_{2,9}}{2\alpha+5} \pdv{}{u_n^9} + \frac{u_n^1}{40320}  \frac{
    N_{2,10}}{2\alpha+5} \pdv{}{u_n^{10}},
\end{split}
    \\[1.2ex]
    \label{eq:g10_R3}
    \begin{split}
    R_3 
    &= \pdv{}{u_n^3}
    + u_n^1 \pdv{}{u_n^4}
    + \left( \frac12 (u_n^1)^2 + (\alpha+2)u_n^2 \right) \pdv{}{u_n^5}
    + \left( \frac16 (u_n^1)^3 + (\alpha+2)u_n^1 u_n^2 \right) \pdv{}{u_n^6} 
    \\
    &+ \left( \frac1{24}(u_n^1)^4 + \frac{\alpha+1}{2}(u_n^1)^2 u_n^2 + u_n^1 u_n^3 + \frac{(\alpha+1)(\alpha+2)}{2}(u_n^2)^2 \right) \pdv{}{u_n^7} 
    \\
    & + \left( \frac1{120}(u_n^1)^5 + \frac{\alpha}{6}(u_n^1)^3 u_n^2 + \frac{\alpha(\alpha+2)}{2}u_n^1 (u_n^2)^2 + \frac12 (u_n^1)^2 u_n^3 + (\alpha+2)u_n^2 u_n^3 \right) \pdv{}{u_n^8} 
    \\
    &+ \frac{1}{720} \frac{N_{3,9}}{2\alpha+5} \pdv{}{u_n^9}  + \frac{1}{5040} \frac{N_{3,10}}{2\alpha+5} \pdv{}{u_n^{10}},     
    \end{split}
    \\[1.2ex]
    \label{eq:g10_R4}
    \begin{split} 
    R_4 &= \pdv{}{u_n^4}
    + u_n^1 \pdv{}{u_n^5}
    + \left( \frac12 (u_n^1)^2 + (\alpha+2)u_n^2 \right) \pdv{}{u_n^6} + \left( \frac16 (u_n^1)^3 + (\alpha+1)u_n^1 u_n^2 + u_n^3 \right) \pdv{}{u_n^7} 
    \\
    & + \left( \frac1{24}(u_n^1)^4 + \frac{\alpha}{2}(u_n^1)^2 u_n^2 + \frac{\alpha(\alpha+2)}{2}(u_n^2)^2 + u_n^1 u_n^3 \right) \pdv{}{u_n^8} 
    \\
    & + \frac{1}{120} 
    \frac{N_{4,9}
    }{2\alpha+5} \pdv{}{u_n^9}  + \frac{1}{720} \frac{
    N_{4,10}
    }{2\alpha+5} \pdv{}{u_n^{10}},    
    \end{split}
    \\[1.2ex]
    \label{eq:g10_R5}
    \begin{split} 
    R_5 
    &= \pdv{}{u_n^5} + u_n^1 \pdv{}{u_n^6}
    + \left( \frac12 (u_n^1)^2 + (\alpha+1)u_n^2 \right) \pdv{}{u_n^7}+ \left( \frac16 (u_n^1)^3 + \alpha u_n^1 u_n^2 + u_n^3 \right) \pdv{}{u_n^8} 
    \\
    &+ \frac{1}{24} \frac{N_{5,9}
    }{2\alpha+5} \pdv{}{u_n^9} + \frac{1}{120} \frac{N_{5,10}
    }{2\alpha+5} \pdv{}{u_n^{10}},     
    \end{split}
    \\[1.2ex]
    \label{eq:g10_R6}
    \begin{split} 
    R_6 
    &= \pdv{}{u_n^6}
    + u_n^1 \pdv{}{u_n^7}
    + \left( \frac12 (u_n^1)^2 + \alpha u_n^2 \right) \pdv{}{u_n^8} 
    \\
    &+ \frac{1}{6} \frac{(12\alpha^2+18\alpha-12)u_n^1 u_n^2 + (2\alpha+5)(u_n^1)^3 + 12(\alpha+1)u_n^3}{2\alpha+5} \pdv{}{u_n^9}  + \frac{1}{24} \frac{N_{6,10}}{2\alpha+5} \pdv{}{u_n^{10}},    
    \end{split}
    \\[1.2ex]
    \begin{split} 
    R_7 
    &= \pdv{}{u_n^7}
+ u_n^1 \pdv{}{u_n^8}
+ \frac{1}{2} \frac{(4\alpha^2+6\alpha-4)u_n^2 + (2\alpha+5)(u_n^1)^2}{2\alpha+5} \pdv{}{u_n^9} \\
&\quad + \frac{1}{6} \frac{(12\alpha^2+6\alpha-6)u_n^1 u_n^2 + (2\alpha+5)(u_n^1)^3 + 12(\alpha-1)u_n^3}{2\alpha+5} \pdv{}{u_n^{10}},    
    \end{split}
     \\[1.2ex]
    R_8 &= \pdv{}{u_n^8}
    + u_n^1 \pdv{}{u_n^9}
+ \frac{1}{2} \frac{(4\alpha^2+2\alpha-2)u_n^2 + (2\alpha+5)(u_n^1)^2}{2\alpha+5} \pdv{}{u_n^{10}}, \\[1.2ex]
R_9 &= \pdv{}{u_n^9} + u_n^1 \pdv{}{u_n^{10}},
\end{align}
\end{subequations}
where $N_{2,9}$ and $N_{2,10}$ in $R_2$~\eqref{eq:g10_R2} are
\begin{align*}
    N_{2,9}&=
    1680\alpha^4(u_n^2)^3
    +840\alpha^3(u_n^1)^2(u_n^2)^2
    +84\alpha^2(u_n^1)^4 u_n^2
    +2\alpha(u_n^1)^6
    +7560\alpha^3(u_n^2)^3
    +2100\alpha^2(u_n^1)^2(u_n^2)^2
    \\
    &+126\alpha(u_n^1)^4 u_n^2
    +5(u_n^1)^6
    +5040\alpha^2 u_n^1 u_n^2 u_n^3
    +9240\alpha^2(u_n^2)^3
    +420\alpha(u_n^1)^3 u_n^3
    +420\alpha(u_n^1)^2(u_n^2)^2
    \\
    &-84(u_n^1)^4 u_n^2
    +15120\alpha u_n^1 u_n^2 u_n^3
    +420(u_n^1)^3 u_n^3
    -840(u_n^1)^2(u_n^2)^2
    +15120\alpha u_n^2 u_n^4
    +2520(u_n^1)^2 u_n^4
    \\
    &+10080 u_n^1 u_n^2 u_n^3
    -3360(u_n^2)^3
    +30240 u_n^2 u_n^4, \\[1mm]
    N_{2,10}&= 
    6720\alpha^4 u_n^1(u_n^2)^3
    +1680\alpha^3(u_n^1)^3(u_n^2)^2
    +112\alpha^2(u_n^1)^5 u_n^2
    +2\alpha(u_n^1)^7
    +16800\alpha^3u_n^1(u_n^2)^3 
    \\
    &+840\alpha^2(u_n^1)^3(u_n^2)^2
    +56\alpha(u_n^1)^5 u_n^2 +5(u_n^1)^7
    +40320\alpha^3(u_n^2)^2 u_n^3
    +13440\alpha^2(u_n^1)^2 u_n^2 u_n^3
    \\
    &+3360\alpha^2 u_n^1(u_n^2)^3
    +672\alpha(u_n^1)^4 u_n^3
    -840\alpha(u_n^1)^3(u_n^2)^2
    -56(u_n^1)^5 u_n^2
    +100800\alpha^2(u_n^2)^2 u_n^3
    \\
    &+6720\alpha(u_n^1)^2 u_n^2 u_n^3
    -6720\alpha u_n^1(u_n^2)^3
    -336(u_n^1)^4 u_n^3+60480\alpha u_n^1 u_n^2 u_n^4
    \\
    &+20160\alpha u_n^1(u_n^3)^2
    +20160\alpha(u_n^2)^2 u_n^3
    +5040(u_n^1)^3 u_n^4
    -6720(u_n^1)^2 u_n^2 u_n^3
    \\&+120960 u_n^1 u_n^2 u_n^4
    -10080 u_n^1(u_n^3)^2
    -40320(u_n^2)^2 u_n^3,    
\end{align*}
the numerators $N_{3,9}$ and $N_{3,10}$ in $R_3$~\eqref{eq:g10_R3} are 
\begin{align*}
    N_{3,9}&=240\alpha^4(u_n^2)^3
  +360\alpha^3(u_n^1)^2(u_n^2)^2
        +60\alpha^2(u_n^1)^4 u_n^2+2\alpha(u_n^1)^6+1080\alpha^3(u_n^2)^3
        +900\alpha^2(u_n^1)^2(u_n^2)^2
            \\
            &+90\alpha(u_n^1)^4 u_n^2+5(u_n^1)^6+1440\alpha^2 u_n^1 u_n^2 u_n^3
            +1320\alpha^2(u_n^2)^3+240\alpha(u_n^1)^3 u_n^3 +180\alpha(u_n^1)^2(u_n^2)^2
            \\
            &-60(u_n^1)^4 u_n^2+4320\alpha u_n^1 u_n^2 u_n^3+240(u_n^1)^3 u_n^3
            -360(u_n^1)^2(u_n^2)^2+2160\alpha u_n^2 u_n^4+1080(u_n^1)^2 u_n^4
            \\
            &+2880 u_n^1 u_n^2 u_n^3-480(u_n^2)^3+4320 u_n^2 u_n^4, \\[1mm]
    N_{3,10}&=
        1680\alpha^4 u_n^1(u_n^2)^3
        +840\alpha^3(u_n^1)^3(u_n^2)^2
        +84\alpha^2(u_n^1)^5 u_n^2
        +2\alpha(u_n^1)^7
        +4200\alpha^3 u_n^1(u_n^2)^3
        \\
        &+420\alpha^2(u_n^1)^3(u_n^2)^2
        +42\alpha(u_n^1)^5 u_n^2
        +5(u_n^1)^7
        +5040\alpha^3(u_n^2)^2 u_n^3
        +5040\alpha^2(u_n^1)^2 u_n^2 u_n^3
        \\
        &+840\alpha^2 u_n^1(u_n^2)^3
        +420\alpha(u_n^1)^4 u_n^3
        -420\alpha(u_n^1)^3(u_n^2)^2
        -42(u_n^1)^5 u_n^2
        +12600\alpha^2(u_n^2)^2 u_n^3
        \\
        &+2520\alpha(u_n^1)^2 u_n^2 u_n^3
        -1680\alpha u_n^1(u_n^2)^3
        -210(u_n^1)^4 u_n^3
        +15120\alpha u_n^1 u_n^2 u_n^4
        +5040\alpha u_n^1(u_n^3)^2
        \\
        &+2520\bigl(\alpha(u_n^2)^2 u_n^3
        +(u_n^1)^3 u_n^4
        -(u_n^1)^2 u_n^2 u_n^3\bigr)
        +30240 u_n^1 u_n^2 u_n^4
        -2520 u_n^1(u_n^3)^2-5040(u_n^2)^2 u_n^3,        
\end{align*}
the numerators $N_{4,9}$ and $N_{4,10}$ in $R_4$~\eqref{eq:g10_R4} are
\begin{align*}
    N_{4,9}&=
            120\alpha^3 u_n^1(u_n^2)^2
            +40\alpha^2(u_n^1)^3 u_n^2
            +2\alpha(u_n^1)^5
            +300\alpha^2 u_n^1(u_n^2)^2
            +60\alpha(u_n^1)^3 u_n^2
            +5(u_n^1)^5\\
            &+240\alpha^2 u_n^2 u_n^3
            +120\alpha(u_n^1)^2 u_n^3
            +60\alpha u_n^1(u_n^2)^2
            -40(u_n^1)^3 u_n^2
            +720\alpha u_n^2 u_n^3
            \\
            &+120(u_n^1)^2 u_n^3
            -120 u_n^1(u_n^2)^2
            +360 u_n^1 u_n^4
            +480 u_n^2 u_n^3 ,       \\[1mm]
    N_{4,10}&=
            240\alpha^4(u_n^2)^3
            +360\alpha^3(u_n^1)^2(u_n^2)^2
            +60\alpha^2(u_n^1)^4 u_n^2+2\alpha(u_n^1)^6
            +600\alpha^3(u_n^2)^3
            +180\alpha^2(u_n^1)^2(u_n^2)^2
            \\
            &+30\alpha(u_n^1)^4 u_n^2
            +5(u_n^1)^6
            +1440\alpha^2 u_n^1 u_n^2 u_n^3
            +120\alpha^2(u_n^2)^3
            +240\alpha(u_n^1)^3 u_n^3
            -180\alpha(u_n^1)^2(u_n^2)^2
            \\
            &-30(u_n^1)^4 u_n^2
            +720\alpha u_n^1 u_n^2 u_n^3
            -240\alpha(u_n^2)^3
            -120(u_n^1)^3 u_n^3
            +2160\alpha u_n^2 u_n^4
            +720\alpha(u_n^3)^2
            \\
            &+1080(u_n^1)^2 u_n^4
            -720 u_n^1 u_n^2 u_n^3
            +4320 u_n^2 u_n^4
            -360(u_n^3)^2,        
\end{align*}
the numerators $N_{5,9}$ and $N_{5,10}$ in $R_5$~\eqref{eq:g10_R5} are 
\begin{align*}
    N_{5,9}&=
    24\alpha^3(u_n^2)^2+24\alpha^2(u_n^1)^2 u_n^2 +2\alpha(u_n^1)^4 +60\alpha^2(u_n^2)^2
    +36\alpha(u_n^1)^2 u_n^2 +5(u_n^1)^4+48\alpha u_n^1 u_n^3\\
    &+12\alpha(u_n^2)^2
    -24(u_n^1)^2 u_n^2+48 u_n^1 u_n^3-24(u_n^2)^2+72 u_n^4  ,  \\[1mm]
    N_{5,10}&=
            120\alpha^3 u_n^1(u_n^2)^2+40\alpha^2(u_n^1)^3 u_n^2 +2\alpha(u_n^1)^5 
            +60\alpha^2 u_n^1(u_n^2)^2
            +20\alpha(u_n^1)^3 u_n^2+5(u_n^1)^5\\
            &+240\alpha^2 u_n^2 u_n^3
            +120\alpha(u_n^1)^2 u_n^3
            -60\alpha u_n^1(u_n^2)^2-20(u_n^1)^3 u_n^2+120\alpha u_n^2 u_n^3
            \\
            &-60(u_n^1)^2 u_n^3+360 u_n^1 u_n^4-120 u_n^2 u_n^3     ,   
\end{align*}
and the numerator $N_{6,10}$ in $R_6$~\eqref{eq:g10_R6} is 
\begin{align*}
    N_{6,10}&= 
    24\alpha^3(u_n^2)^2+24\alpha^2(u_n^1)^2u_n^2 +2\alpha(u_n^1)^4+12\alpha^2(u_n^2)^2
    +12\alpha(u_n^1)^2 u_n^2
    +5(u_n^1)^4
    +48\alpha u_n^1 u_n^3\\
    &-12\alpha(u_n^2)^2-12(u_n^1)^2 u_n^2-24 u_n^1 u_n^3+72 u_n^4.
\end{align*}

\addcontentsline{toc}{section}{References}
\printbibliography

\end{document}